\documentclass[a4paper,11pt]{article}

\usepackage{graphicx} 
\usepackage[margin=1.5cm]{geometry}
\usepackage{placeins}
\usepackage[english]{babel}
\usepackage{microtype}

\usepackage{amsmath,amsthm,amssymb}  

\usepackage{newtxtext}
\usepackage{newtxmath}

\usepackage{enumitem}
\setlist[1]{labelindent=\parindent}
\setlist[enumerate]{label=(\arabic*)}

\usepackage{mathtools}
\usepackage{graphicx}
\usepackage[colorlinks=true, allcolors=blue]{hyperref}
\usepackage[capitalize]{cleveref}
\usepackage{xspace}
\usepackage{thmtools}
\usepackage{thm-restate}
\usepackage{url}
\usepackage{authblk}
\usepackage{comment}
\usepackage{subcaption}
\usepackage{bm}
\usepackage{multirow}
\usepackage{todonotes}
\usepackage{booktabs}
\usepackage{soul}
\usepackage{balance}
\usepackage{authblk}
\usepackage{scrextend} 
\usepackage{algorithm}
\usepackage{algpseudocode}
\usepackage[bottom]{footmisc}
\usepackage[group-separator={,},group-minimum-digits=4]{siunitx}
\usepackage{enumitem}
\setlist[1]{labelindent=\parindent}
\setlist[enumerate]{label=(\arabic*)}

\usepackage{titlesec}
\titleformat*{\subparagraph}{\itshape}
\titlespacing*{\subparagraph}{\parindent}{\parskip + 0.3em}{1em}

\definecolor{CWIBlue}{rgb}{0.38, 0.478, 0.96}
\definecolor{CWIContrGreen}{rgb}{0.654, 0.704, 0.462}
\definecolor{CWIDGreen}{rgb}{0.372, 0.494, 0.419}
\definecolor{CWILGreen}{rgb}{0.603, 0.717, 0.486}
\definecolor{CWIRed}{rgb}{0.886, 0.188, 0.262}

\newif\ifhidetodos

\usepackage{marginnote}

\newcommand{\cbox}[2][yellow]{%
  \fcolorbox{#1}{white}{\parbox{\dimexpr\linewidth-2\fboxsep}{\strut #2\strut}}%
}

\newcommand{\customtodo}[3]{\textcolor{#2}{\(\blacktriangledown\)}\marginnote{\raggedright \textcolor{#2}{\textbf{#1:} #3}}}
\newcommand{\customtododisplay}[3]{\noindent\cbox[#2]{\textcolor{#2}{\textbf{#1:} #3}}}
\newcommand{\customtodoinline}[3]{\textcolor{#2}{\textcolor{#2}{\textbf{#1:} #3}}}
\newcommand{\INDSTATE}[1][1]{\State\hspace{\algorithmicindent}}

\ifhidetodos
    \renewcommand{\customtodo}[3]{}
    \renewcommand{\customtododisplay}[3]{}
    \renewcommand{\customtodoinline}[3]{}
\fi

\newtheorem{theorem}{Theorem}
\newtheorem{definition}{Definition}

\newtheorem{lemma}{Lemma}

\newtheorem{corollary}{Corollary}
\newtheorem{example}{Example}
\newtheorem{observation}{Observation}
\newtheorem{proposition}{Proposition}

\def\dd{\mathinner{.\,.}}
\newcommand{\cO}{\mathcal{O}}
\newcommand{\Output}{\mathrm{output}}

\newcommand{\ST}{\textsf{ST}}
\newcommand{\sd}{\textsf{sd}}
\newcommand{\str}{\textsf{str}}

\newcommand{\occ}{\textsf{occ}}
\newcommand{\LCA}{\textsf{LCA}\xspace}
\newcommand{\SM}{\textsf{SM}\xspace}
\newcommand{\NCSMQ}{\textsf{SM}\texttt{+}\xspace}
\newcommand{\WebKB}{\textsc{WebKB}\xspace}
\newcommand{\News}{\textsc{News}\xspace}
\newcommand{\Genes}{\textsc{Genes}\xspace}
\newcommand{\VIR}{\textsc{Vir}\xspace}
\newcommand{\DRQ}{$1$-\textsf{DR}\xspace}
\newcommand{\kDRQ}{$k$-\textsf{DR}\xspace}
\newcommand{\WebKBsample}{\textsc{WebKB-sam}\xspace}
\newcommand{\Genessample}{\textsc{Genes-sam}\xspace}

\newcommand{\Movielens}{\textsc{Movielens}\xspace} \newcommand{\BookCrossing}{\textsc{Book-Crossing}\xspace}
\newcommand{\Alibaba}{\textsc{Alibaba}\xspace}

\renewcommand{\epsilon}{\varepsilon}

\newcommand{\CSQ}{\textsf{CQS}\xspace}

\newcommand{\BAI}{\textsf{BA1}\xspace}
\newcommand{\BAII}{\textsf{BA2}\xspace}
\newcommand{\BAIII}{\textsf{BA3}\xspace}
\newcommand{\SCM}{\textsf{SCM}\xspace}

\newcommand{\CPM}{\textsf{UPM}\xspace} 
\newcommand{\RM}{\textsf{RM}\xspace} 
\newcommand{\DM}{\textsf{DM}\xspace}

\newcommand{\INFL}{\textsc{INFL}\xspace}
\newcommand{\HUM}{\textsc{HUM}\xspace}
\newcommand{\SARS}{\textsc{SARS}\xspace}
\newcommand{\rulesep}{\unskip\ \vrule\ }

\newcommand{\defproblem}[3]{
\vspace{2mm}
\noindent\fbox{
   \begin{minipage}{0.96\columnwidth}
   \textsc{#1}\\
   {\bf{Input:}} #2  \\
   {\bf{Output:}} #3
   \end{minipage}
   }
   \vspace{-2mm}
}

\newcommand{\defDSproblem}[3]{
\vspace{2mm}
\noindent\fbox{
   \begin{minipage}{0.96\columnwidth}
   \textsc{#1}\\
   {\bf{Preprocess:}} #2  \\
   {\bf{Query:}} #3
   \end{minipage}
   }
   \vspace{2mm}
}

\newcommand{\seq}[1]{\texttt{\shortstack[l]{#1}}}

\renewcommand{\thefootnote}{\fnsymbol{footnote}}

\usepackage[backend=biber,style=alphabetic,
 maxnames=10,
  url=false,     
  isbn=false     
            ]{biblatex}
\AtEveryBibitem{%
  \ifentrytype{book}{}{%
    \ifentrytype{incollection}{}{%
      \clearname{editor}%
      \clearlist{publisher}%
    }%
  }
  \clearfield{url}%
  \clearfield{urldate}%
  \clearfield{timestamp}
  \clearfield{month}
  \clearfield{series}
  \clearfield{eprint}%
  \clearfield{note}%
}

\title{Subtree Mode and Applications}
\author[1]{Jialong Zhou\thanks{The first two authors contributed equally to this work.}} 
\author[2,3]{Ben Bals\protect\footnotemark[1]} 
\author[3]{Matei Tinca}
\author[1]{Ai Guan}
\author[1]{\\Panagiotis Charalampopoulos}
\author[1]{Grigorios Loukides}
\author[2,3]{Solon P.\ Pissis}
\affil[1]{King's College London, London, UK \quad \textsuperscript{2}CWI, Amsterdam, The Netherlands \quad 
\textsuperscript{3}Vrije Universiteit, Amsterdam, The Netherlands}

\date{\vspace{-.5cm}}

\begin{document}
\maketitle

\renewcommand{\thefootnote}{\arabic{footnote}}

\begin{abstract}
The \emph{mode} of a collection of values (i.e., the most frequent
value in the collection) is a key summary statistic.
Finding the mode in a given \emph{range} of an array of values is thus of great importance, and constructing a data structure to solve this problem is in fact the well-known \emph{Range Mode} problem.
In this work, we introduce the \emph{Subtree Mode} (\SM) problem, the analogous problem in a \emph{leaf-colored tree}, where the task is to compute the most frequent color in the leaves of the subtree of a given node.
\SM is motivated by several applications in domains such as text analytics and  biology, where the data are hierarchical and can thus be represented as a (leaf-colored) tree. 
Our central contribution is a time-optimal algorithm for \SM that computes the answer for every node of an input $N$-node tree in $\cO(N)$ time.
We further show how our solution
can be adapted for \emph{node-colored}  trees, or for computing the $k$ most frequent colors, for any given $k =\cO(1)$, in the optimal $\cO(N)$ time.
Moreover, we prove that a similarly fast solution for 
when the input is a sink-colored directed acyclic graph instead of a leaf-colored tree is highly unlikely. Our experiments on real datasets with trees of up to $7.3$ billion nodes demonstrate that our algorithm is faster than baselines 
by at least one order of magnitude and much more space efficient.
They also show that it is effective in pattern mining, sequence-to-database search, and biology applications.
\end{abstract}

\setlength{\skip\footins}{5pt}  %

\section{Introduction}\label{sec:intro}

A key summary statistic of a collection of values is its  \emph{mode} (i.e.,\,the most frequent value in the collection)~\cite{DBLP:books/mk/HanKP2011}.
Finding the mode in a given \emph{range} of values of an array (e.g., a window of a sequence) is thus of great importance. In fact,  constructing a data structure to solve this problem is the well-known \emph{Range Mode}  (\RM) problem~\cite{DBLP:journals/njc/KrizancMS05,DBLP:journals/corr/abs-1101-4068,DBLP:conf/icalp/GreveJLT10,DBLP:conf/soda/WilliamsX20}. For instance, \RM allows finding the most frequently purchased item over a certain time period~\cite{sea23}, the most frequent $q$-gram (i.e., length-$q$ substring) occurring in a genomic region~\cite{DBLP:journals/bioinformatics/AyadPP18}, or the most frequent value of an attribute in a range of database tuples~\cite{DBLP:books/mk/HanKP2011}.

\paragraph{The \SM Problem.}
We introduce a natural variant of \RM that asks for the 
most frequent color in the leaves of subtrees of a leaf-colored tree. For example, suppose that we want to tag the folders of a  filesystem with the most frequent file type (e.g., image, document, etc.) contained in them to provide quick visual clues about the prevalent folder contents.
The folder structure can be modeled as a tree $\mathcal{T}$ with each non-empty   
folder being an internal node, its subfolders being its children, 
and each file being a leaf (attached to the  containing folder's node).   
Furthermore, each leaf of $\mathcal{T}$ is colored based on the type of the file it models. The mode for a node $v$ in $\mathcal{T}$ gives us the most frequent file type in the folder corresponding to $v$. We call this problem \emph{Subtree Mode} (\SM) and define it below. 

\defDSproblem{Subtree Mode (\SM)}
{A rooted tree $\mathcal{T}$ on $N$ nodes with every leaf colored from a set $\{0, \ldots, \Delta-1\}$ of colors (integers).}
{Given a node $v$ of $\mathcal{T}$, output the most frequent color~$c^{\max}_v$ in the leaves of the subtree rooted at $v$ (breaking ties arbitrarily).}

For simplicity, we refer to $c^{\max}_v$ as the \emph{mode} of node $v$.

\paragraph{Motivation.} We start by motivating \SM from a theory standpoint.
\SM can be easily reduced to \RM and by using the state-of-the-art algorithm for \RM~\cite{DBLP:journals/corr/abs-1101-4068}, we obtain an $\cO(N\sqrt{N})$-time baseline for \SM (see Section~\ref{sec:baselines} for the  details). 
Unfortunately, there is strong evidence suggesting that a significantly faster algorithm for \RM is rather unlikely~\cite{DBLP:journals/corr/abs-1101-4068}. This gives rise to two fundamental questions: (Q1) \emph{Can we solve \SM directly and significantly faster?} (Q2) \emph{Can we solve problems that generalize \SM to more complex graph types efficiently}? 

We next motivate \SM from a practical standpoint.
\SM is motivated by several  
applications in domains such as text analytics and biology. We sketch some of these below. 

\emph{SNP-based Phylogenetic Tree Annotation.}~Single-nucleotide polymorphisms (SNPs) are genetic variations at specific nucleotide sites within a species' genome. SNPs are linked to diseases such as cancer, Alzheimer’s disease, and other inherited disorders~\cite{Kostem2011,pnas04}. Owing to their physical proximity to disease-associated loci, SNP alleles are frequently co-inherited with pathogenic variants across generations, reflecting the principle of genetic linkage~\cite{Syvanen2001AccessingGV}.
Phylogenetic trees show evolutionary relationships between organisms; a leaf represents an organism from which a DNA sequence has been obtained, and an internal node represents a common ancestor of all the organisms that correspond to its leaf descendants~\cite{phylobook}. The leaves are often annotated manually with categorical values related to SNPs (e.g., types of diseases~\cite{10.1371/journal.pone.0005022}), and each internal node with a value summarizing the values of its subtree~\cite{10193,plospathogens,pnas2025,10.1093/nar/gkae268}. 
We can model such a phylogenetic tree as $\mathcal{T}$ and color each leaf of $\mathcal{T}$ according to the SNP-related value of its corresponding leaf in the phylogenetic tree. 
Then, the mode of $v$ identifies the most prevalent SNP-related value among the group of organisms corresponding to~$v$, supporting interpretation and hypothesis generation about evolutionary processes~\cite{mbe25}. Thus, \SM can reveal significant associations, such as 
 clade‑specific markers~\cite{Coll2014} and genotype–phenotype links~\cite{10.1093/bioinformatics/btad215,Arora2017OriginSyphilis},  between subtrees and SNP-related values. This is valuable in studies of pathogen spread~\cite{plospathogens} and antibiotic resistance evolution~\cite{Arora2017OriginSyphilis}. For example, coloring the phylogenetic tree in~\cite[Fig. 3b]{10193} by  ``continent of origin'' and solving \SM identifies subtrees where Ethiopian samples are prevalent, which is central in the study of~\cite{plospathogens}. Similarly, coloring the phylogenetic tree in~\cite[Fig. S1]{10193} by  SNP state at selected genomic positions (respectively, by antibiotic resistance status) identifies subtrees where particular SNP states (respectively,  resistant samples) predominate in the data of the study of~\cite{Arora2017OriginSyphilis}, which is crucial for detecting genotype-phenotype links. There are widely-used systems for annotating phylogenetic tree nodes with  values (e.g.,~\cite{itol6,10193,icytree}) but one must manually compute the most frequent colors (i.e., solve \SM) and perform the annotation in these systems,  whereas our approach is automatic and efficient.

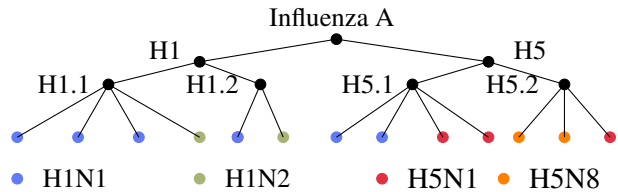
\begin{figure}[!ht]\centering
\scalebox{1.0}{
\begin{tikzpicture}[level distance=0.1cm, sibling distance=0.1cm]
\filldraw[black] (0,0) circle (2pt);\node[left] at (0.9,0.3) {\small{Influenza A}};
\filldraw[black] (-1.8,-0.3) circle (2pt); \node[left] at (-1.9,-0.15) {\small{H1}};
\filldraw[black] (2,-0.3) circle (2pt); \node[right] at (2.2,-0.15) {\small{H5}};
\filldraw[black] (-3,-0.6) circle (2pt);\node[left] at (-3.1,-0.6) {\small{H1.1}};
\filldraw[black] (-1,-0.6) circle (2pt);\node[left] at (-1.15,-0.6) {\small{H1.2}};
\filldraw[black] (1,-0.6) circle (2pt);\node[left] at (0.9,-0.6) {\small{H5.1}};
\filldraw[black] (3,-0.6) circle (2pt);\node[left] at (2.8,-0.6) {\small{H5.2}};

\filldraw[CWIBlue] (-4.2,-1.3) circle (2pt);
\filldraw[CWIBlue] (-3.4,-1.3) circle (2pt);
\filldraw[CWIBlue]   (-2.6,-1.3) circle (2pt);

\filldraw[CWIContrGreen] (-1.8,-1.3) circle (2pt);
\filldraw[CWIRed] (-1.3,-1.3) circle (2pt);

\filldraw[CWIBlue] (-1.3,-1.3) circle (2pt);
\filldraw[CWIContrGreen] (-0.7,-1.3) circle (2pt);

\filldraw[CWIBlue] (0,-1.3) circle (2pt);
\filldraw[CWIBlue] (0.6,-1.3) circle (2pt);
\filldraw[CWIRed]  (1.4,-1.3) circle (2pt);
\filldraw[CWIRed]  (2,-1.3) circle (2pt);

\filldraw[orange] (2.4,-1.3) circle (2pt);
\filldraw[orange] (3,-1.3) circle (2pt);
\filldraw[CWIRed]   (3.6,-1.3) circle (2pt);

\draw (0,0) -- (-1.8,-0.3);
\draw (0,0) -- (2,-0.3);
\draw (-1.8,-0.3) -- (-3,-0.6);
\draw (-1.8,-0.3) -- (-1,-0.6);
\draw (2,-0.3) -- (1,-0.6);
\draw (2,-0.3) -- (3,-0.6);

\draw (-3,-0.6) -- (-4.2,-1.3);
\draw (-3,-0.6) -- (-3.4,-1.3);
\draw (-3,-0.6) -- (-2.6,-1.3);
\draw (-3,-0.6) -- (-1.8,-1.3);

\draw (-1,-0.6) -- (-1.3,-1.3);
\draw (-1,-0.6) -- (-0.7,-1.3);

\draw (1,-0.6) -- (0,-1.3);
\draw (1,-0.6) -- (0.6,-1.3);
\draw (1,-0.6) -- (1.4,-1.3);
\draw (1,-0.6) -- (2,-1.3);

\draw (3,-0.6) -- (2.4,-1.3);
\draw (3,-0.6) -- (3,-1.3);
\draw (3,-0.6) -- (3.6,-1.3);

\begin{scope}[shift={(0,-0.55)}]
\filldraw[CWIBlue] (-4.2,-1.3) circle (2pt); \node[right] at (-4,-1.3) {\small{H1N1}};
\filldraw[CWIContrGreen] (-1.8,-1.3) circle (2pt);  \node[right] at (-1.6,-1.3) {\small{H1N2}};
\filldraw[CWIRed] (0.6,-1.3) circle (2pt); \node[right] at (0.8,-1.3) {H5N1};
\filldraw[orange] (2.2,-1.3) circle (2pt); \node[right] at (2.4,-1.3) {H5N8};
\end{scope}
\end{tikzpicture}}
\caption{SNP-based phylogenetic tree.}\label{fig:biological_example}
\end{figure}

\begin{example}
Fig.~\ref{fig:biological_example} shows an SNP-based phylogenetic tree alike the ones used to study the Influenza A virus in~\cite{10.1371/journal.pone.0005022}. 
The leaves are annotated with Influenza A subtypes (H1N1, H1N2, H5N1, and H5N8), and the internal nodes represent \emph{lineages} (subtrees sharing specific patterns of SNPs). An internal node with .1 and .2 corresponds to a subtype in the Eastern and Western Hemisphere, respectively. We model the tree in Fig.~\ref{fig:biological_example} as $\mathcal{T}$ in our \SM problem and assign colors $0$ (\textcolor{CWIBlue}{blue}), $1$ (\textcolor{CWIContrGreen}{green}), $2$  (\textcolor{CWIRed}{red}), and $3$ (\textcolor{orange}{orange}) to the leaves of subtype H1N1, H1N2, H5N1, and H5N8, respectively. Thus, the number of nodes $N$ of $\mathcal{T}$ is $20$, and the number of colors $\Delta$ is $4$. Given a query consisting of the node H1, \SM outputs $0$, as the subtree rooted at this node has more \textcolor{CWIBlue}{blue} leaves than \textcolor{CWIContrGreen}{green}.   
\end{example}

\emph{Top-$1$ Document Retrieval (\DRQ).}~In the \DRQ problem~\cite{DBLP:conf/soda/NavarroN12,DBLP:journals/jacm/HonSTV14, DBLP:journals/siamcomp/NavarroN17,DBLP:conf/soda/0001N25}, 
a collection $\mathcal{S}$ of $\Delta$ documents (strings), $S_0, \ldots, S_{\Delta-1}$, is given for preprocessing, and we are asked to answer queries of the following type: given a query pattern $P$, output the string in~$\mathcal{S}$ in which $P$ occurs most frequently as a substring.  
\DRQ can be reduced to \SM. We preprocess $\mathcal{S}$ by first constructing its \emph{suffix tree} (i.e., the compacted trie of the suffixes of the string $S_0\$_0\ldots S_{\Delta-1}\$_{\Delta-1}$, where $\$_i$, for each $i\in[0,\Delta)$, is a unique delimiter)~\cite{DBLP:conf/focs/Weiner73} and then coloring the leaves corresponding to the suffixes starting in $S_i\$_i$ with color $i$.
This is the leaf-colored tree in the constructed instance of \SM.
For a query pattern $P$, we spell $P$ on the suffix tree, arriving at a node $v$, compute the mode of $v$, say $i$, 
and output~$S_i$ as the answer. The example below illustrates how we solve \DRQ via \SM. 

 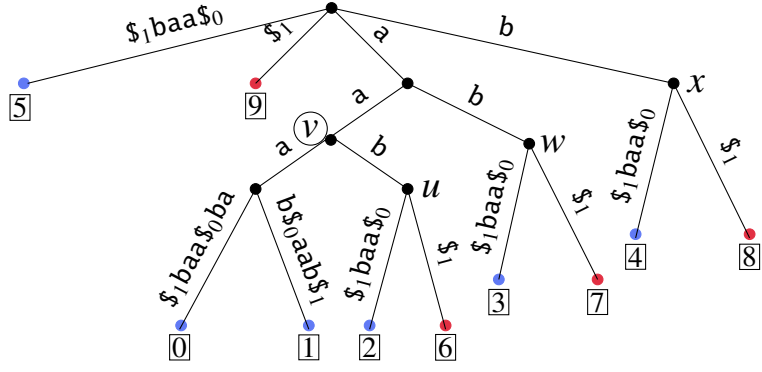
\begin{figure}[htbp]
\centering
\scalebox{1.0}{
 \begin{tikzpicture}[- , level distance=0.7cm]
\filldraw[black] (0,0) circle (2pt); 
\filldraw[black] (-0.01,-1.75) circle (2pt);
\node[left, font=\Large, draw, circle, inner sep=1pt] at (-0.05, -1.6) {$v$};

\filldraw[black] (2.6,-1.8) circle (2pt);
\node[right, font=\Large] at (2.6, -1.8) {$w$};

\filldraw[CWIBlue] (2.2,-3.6) circle (2pt); 
\filldraw[CWIRed] (3.5,-3.6) circle (2pt); 

\filldraw[black] (-1,-2.4) circle (2pt);  

\filldraw[black] (1,-2.4) circle (2pt);
\node[right, font=\Large] at (1.05, -2.4) {$u$};

\filldraw[CWIBlue] (-1.98,-4.21) circle (2pt); 
\filldraw[CWIBlue] (-0.3,-4.21) circle (2pt); 

\filldraw[CWIBlue] (-4.05,-1) circle (2pt); 

\filldraw[CWIRed] (-1,-1) circle (2pt); 
\filldraw[black] (1,-1) circle (2pt);  
\filldraw[black] (4.5,-1) circle (2pt);
\node[right, font=\Large] at (4.55, -1.01) {$x$};

\filldraw[CWIBlue] (0.5,-4.21) circle (2pt); 
\filldraw[CWIRed] (1.5,-4.21) circle (2pt);  

\filldraw[CWIBlue] (4,-3) circle (2pt); 
\filldraw[CWIRed] (5.5,-3) circle (2pt); 

  \draw (0,0)   -- (-4.0,-1) node[midway,sloped,above] {$\texttt{\$}_1$\texttt{baa}$\texttt{\$}_0$}; 
  \draw (0,0)   -- (-1,-1) node[midway,sloped,above] {$\texttt{\$}_1$};
  \draw (0,0)   -- (1,-1) node[midway,sloped,above] {$\texttt{a}$}; 
  \draw (0,0)   -- (4.5,-1) node[midway,sloped,above] {$\texttt{b}$}; 
  
  \draw (4.5,-1)   -- (4,-3.05) node[midway,sloped,above] {$\texttt{\$}_1\texttt{baa}\texttt{\$}_0$}; 
  \draw (4.5,-1)   -- (5.5,-3.05) node[midway,sloped,above] {
  $\texttt{\$}_1$}; 

    \draw (1,-1) -- (0,-1.7) node[midway,sloped,above] {$\texttt{a}$}; 
    \draw (1,-1)   -- (2.6,-1.8) node[midway,sloped,above] {$\texttt{b}$}; 

\draw (-0.01,-1.7)   -- (-1,-2.4) node[midway,sloped,above] {$\texttt{a}$}; 
\draw (-0.01,-1.7)   -- (1,-2.4) node[midway,sloped,above] {$\texttt{b}$};

\draw (1,-2.4)   -- (0.5,-4.26) node[midway,sloped,above] {$\texttt{\$}_1\texttt{baa}\texttt{\$}_0$}; 
\draw (1,-2.4)   -- (1.5,-4.26) node[midway,sloped,above] {$\texttt{\$}_1$};
         
        \draw (2.6,-1.8)   -- (2.2,-3.6) node[midway,sloped,above] {$\texttt{\$}_1\texttt{baa}\texttt{\$}_0$}; 
        \draw (2.6,-1.8)   -- (3.5,-3.6) node[midway,sloped,above] {$\texttt{\$}_1$}; 

              \draw (-1,-2.4) -- (-2,-4.26) node[midway,sloped,above] {
              $\texttt{\$}_1\texttt{baa}\texttt{\$}_0\texttt{ba}$}; 
              \draw (-1,-2.4) -- (-0.3,-4.26)node[midway,sloped,above] {$\texttt{b}\texttt{\$}_0\texttt{aab}\texttt{\$}_1$};  

\node[draw, rectangle, fill=white, inner sep=1pt] at (-4.1,-1.32) {5};
\node[draw, rectangle, fill=white, inner sep=1pt] at (-1,-1.3) {9};
\node[draw, rectangle, fill=white, inner sep=1pt] at (-2,-4.5) {0};
\node[draw, rectangle, fill=white, inner sep=1pt] at (-0.3,-4.5) {1};
\node[draw, rectangle, fill=white, inner sep=1pt] at (0.5,-4.5) {2};
\node[draw, rectangle, fill=white, inner sep=1pt] at (1.5,-4.5) {6};
\node[draw, rectangle, fill=white, inner sep=1pt] at (2.2,-3.9) {3};
\node[draw, rectangle, fill=white, inner sep=1pt] at (3.5,-3.9) {7};
\node[draw, rectangle, fill=white, inner sep=1pt] at (4,-3.3) {4};
\node[draw, rectangle, fill=white, inner sep=1pt] at (5.5,-3.3) {8};
\end{tikzpicture}}

\caption{The suffix tree of $S_0\$_0S_1\$_1=\texttt{aaaab}\texttt{\$}_0\texttt{aab}\texttt{\$}_1$.}\label{fig:suffixtree:intro}
\end{figure}

\begin{example}\label{example:drq}
Consider a collection of strings
$\mathcal{S}$ comprised of $S_0=\texttt{aaaab}$ and $S_1=\texttt{aab}$.
Fig.~\ref{fig:suffixtree:intro} shows the suffix tree of $S_0\$_0S_1\$_1$ with its leaves colored as follows: the leaves corresponding to the suffixes starting in $S_0\$_0$ 
are colored with $0$ (\textcolor{CWIBlue}{blue}) and the remaining ones with $1$ (\textcolor{CWIRed}{red}).  
For instance, the second leaf from the right  is colored $0$ (\textcolor{CWIBlue}{blue}) as its suffix $\texttt{b}\texttt{\$}_0\texttt{aab}\texttt{\$}_1$ starts in $S_0\$_0$.   
Consider the query pattern $P=\texttt{aa}$. By spelling~$P$ on the suffix tree, we arrive at node~$v$.
Assuming that we have a data structure for \SM, we obtain $0$, as $v$ has three leaves colored $0$ (\textcolor{CWIBlue}{blue}) and one leaf colored $1$ (\textcolor{CWIRed}{red}).
Then, we output~$S_0$ as the answer to the \DRQ query.
Indeed, $P$ occurs more often as a substring in $S_0$ compared to $S_1$.
\end{example}

\emph{Uniform Pattern Mining (\CPM)}.~
Consider two strings, one comprised of male-targeted ads and another comprised of female-targeted ads, and that we mine a pattern ``strong leader'' which occurs much more frequently in the former string. 
Ads with this pattern may perpetuate gender stereotypes and influence how opportunities or messages are presented to different ad viewer groups, causing discrimination and bias in
decision-making or perception. To prevent this,  
in the spirit of \emph{statistical parity~\cite{parity}}\footnote{This fairness measure requires the probability distributions of
outcomes to be similar across all subpopulations of a population.}, we can mine patterns with ``similar'' frequencies in all strings of an input collection by formulating and solving 
the following problem, called \emph{Uniform Pattern Mining} (\CPM):  
Given a collection~$\mathcal{S}$ of $\Delta$ strings, $S_0, \ldots, S_{\Delta-1}$, and an integer $\epsilon \geq 0$ specified based on domain knowledge, \CPM asks for all strings (patterns) whose frequencies in any pair of strings in~$\mathcal{S}$ differ by \emph{at most}~$\epsilon$.
When the strings in $\mathcal{S}$ represent different subpopulations of a user population and $\epsilon$ is ``small'', such  patterns prevent the discrimination of these subpopulations. On the other hand, when $\epsilon$ is ``large'', those patterns with large differences in their frequencies reveal behavioral preferences that prevail in user subpopulations (e.g., they may represent movie genres viewed by much more men than women).    
The \CPM problem is solved via a reduction to \SM.

\begin{example}\label{ex:UPM}
Consider a collection
$\mathcal{S}$ comprised of $S_0=\texttt{aaaab}$ and $S_1=\texttt{aab}$, and that $\epsilon=1$. 
The output of \CPM is $\{\texttt{aaaa},\texttt{aaaab},\texttt{aaab},\texttt{aab},
\texttt{ab},\texttt{b}\}$, as the difference between the frequency of each of these patterns in $S_0$ and in $S_1$ is at most $1$. For instance, the difference for $\texttt{ab}$ is $1-1\leq \epsilon$.  
\end{example}

\emph{Consistent Query String (\CSQ).}~In the \CSQ problem, a collection $\mathcal{S}$ of $\Delta$ strings, $S_0, \ldots, S_{\Delta-1}$, is given for preprocessing, and we are asked to
quantify how similar a query string $P$ is to the strings in $\mathcal{S}$.  
This can be achieved by counting the number of distinct $q$-grams $Q$ of $P$  whose frequency in $P$ is in the interval ${[\min_{S\in \mathcal{S}}|\occ_S(Q)|-\epsilon,} \max_{S\in \mathcal{S}}|\occ_S(Q)|+\epsilon]$,   where
$|\occ_S(Q)|$ is the frequency of $Q$ in string $S$ and $\epsilon$ is a user-specified parameter capturing approximate consistency. We call each such $q$-gram $\epsilon$-\emph{consistent} with $\mathcal{S}$.  
Clearly, $P$ is similar to the strings in $\mathcal{S}$ if most of its $q$-grams are $\epsilon$-consistent with $\mathcal{S}$.
The \CSQ problem is particularly
relevant for 
databases of highly-similar strings, which are common in genomics, as they are constructed over collections with a shared evolutionary history or a common function~\cite{shuy2017gisaid,DBLP:journals/nar/MistryCWQSSTPRR21}.
For example, the \CSQ problem can be used for identifying viral species, a core challenge in virology~\cite{vir1,vir2,vir3,vir4,vir5}. State-of-the-art methods for this task~\cite{vir3,vir4,vir5}  perform $q$-gram based sequence classification but need separate training for each desired $q$ and decide class membership for a pattern (e.g., SARS-Cov-2 vs. Influenza). On the contrary, a data structure for \CSQ is built once and queried for any $q$, and it provides numerical scores, which is more informative when sequences have errors, as is often the case in practice.

\begin{example}
Consider two DNA sequences,
one that is a genetic variant of SARS-Cov-2~\cite{ncbi_sars_cov_2_2024} and
another that is a genetic variant of the Influenza A virus~\cite{ncbi_influenza_a_2025}. 
Suppose that a biologist does not know whether each 
of these sequences is a genetic variant of SARS-CoV-2 and wants to check this.
They can use the first sequence as query $P_1$ in a database $\mathcal{S}$ comprised of $2,000$ different genetic variants of the SARS~CoV-2 virus~\cite{ncbi_sars_cov_2_2024}, and solve $\CSQ$ for $\epsilon=0$ and $q=4$.
They will find that $99.9\%$ of the $q$-grams of $P_1$ are $0$-consistent with $\mathcal{S}$, and conclude that~$P_1$ is very likely a genetic variant of SARS-Cov-2.
Then, if they repeat the same process with the second sequence as query~$P_2$, they will find that only $3.12\%$ of the $q$-grams of~$P_2$ are $0$-consistent with $\mathcal{S}$. 
Thus, they will conclude that $P_2$ is unlikely to be a genetic variant of SARS-Cov-2.  
\end{example}

\paragraph{Contributions.} In addition to introducing the \SM problem, our work makes the following specific contributions:

\begin{enumerate}

\item Our central contribution is the following theorem. 

\begin{theorem}\label{the:linear-time}
Given a tree $\mathcal{T}$ on $N$ nodes, we can construct, in $\cO(N)$ time, a data structure that can answer any \SM query in $\cO(1)$ time.
In particular, for a given node $v$, 
the query algorithm returns both the mode $c^{\max}_v$ and
its frequency~$f^{\max}_v$.
\end{theorem}
\Cref{the:linear-time} answers Q1 affirmatively.
The algorithm to construct the data structure in  \Cref{the:linear-time} computes the answer $(c_v^{\max}, f_v^{\max})$, for every node $v$ of the input tree $\mathcal{T}$. It works by first splitting $\mathcal{T}$ into a forest of $\Delta$ trees $\mathcal{T}_0,\ldots,\mathcal{T}_{\Delta-1}$, such that each tree~$\mathcal{T}_i$, for $i\in[0,\Delta)$, has all its leaves colored $i$. Every node of each tree~$\mathcal{T}_i$ is associated with \emph{one} node of $\mathcal{T}$. Then, the algorithm makes a bottom-up traversal of $\mathcal{T}$, and, for each node~$v$ of~$\mathcal{T}$, it combines the color frequency information of the children of $v$ in~$\mathcal{T}$ with the information coming from the nodes associated with $v$ in  the trees $\mathcal{T}_0,\ldots,\mathcal{T}_{\Delta-1}$. The efficiency of this algorithm is based on the fact that the total size of all trees
is $\cO(N)$, and on that it employs: (I) an efficient data structure for answering Lowest Common Ancestor (\LCA) queries~\cite{DBLP:conf/latin/BenderF00}; and (II) an efficient algorithm for tree traversal that exploits \LCA information~\cite{DBLP:conf/cpm/KasaiLAAP01}. See \cref{sec:linear}. 

\item We show how the above-mentioned string-processing problems, namely \DRQ, \CPM, and \CSQ, can be solved in optimal time using \Cref{the:linear-time} via linear-time reductions to \SM. 
We remark that via the reduction from \DRQ to \SM, we solve \DRQ by constructing in linear time a data structure that supports queries in optimal time.
The existing data structures~\cite{DBLP:conf/soda/NavarroN12,DBLP:journals/jacm/HonSTV14, DBLP:journals/siamcomp/NavarroN17,DBLP:conf/soda/0001N25} for the more general \kDRQ problem clearly work (with $k=1$) for our problem but, to the best of our knowledge, they do not admit 
an linear-time construction.
Their focus is on obtaining theoretically good space-query time trade-offs.
We also remark that the query and preprocessing time we achieve for the \CSQ problem are optimal. A baseline alternative approach is to construct $\Delta$ suffix trees, one for each string $S$ in the input collection $\mathcal{S}$, and to find if each $q$-gram $Q$ of the query is $\epsilon$-consistent with $\mathcal{S}$ after matching it to each suffix tree to compute $|\occ_S(Q)|$. This approach is prohibitively expensive, as it may take $\Omega(\Delta q|Q|)$ time for a query of length $|Q|$ and $\Delta$ is typically in the order of thousands. 
See \cref{sec:applications}.

\item We show two generalizations of \SM that can be solved efficiently using \Cref{the:linear-time}: (I) having a \emph{node-colored} tree  instead of a leaf-colored tree as input; and (II) finding the $k\geq 1$ most frequent colors instead of the most frequent color. Our result for generalization II implies a linear-space data structure for any $k=\cO(1)$ for \kDRQ, which answers queries in optimal time \emph{and} can be constructed in linear time. 
See \cref{sec:generalizations}.

\item We show that an analogous problem to \SM, where the input is a directed acyclic graph (DAG) instead of a tree and the sinks (nodes in the DAG with no outgoing edges) are colored instead of the tree leaves, is  
unlikely to be solved as fast as \SM. We do this by providing conditional lower bounds  answering Q2 negatively for this problem. 
See \cref{sec:lb}. 

\item We present experiments on $4$ real datasets from different domains showing that our algorithm is \emph{at least one order of magnitude faster} and uses \emph{significantly less memory} compared to three natural baselines. For example, it processes a dataset whose tree has over $7$ billion nodes in less than $20$ minutes, while the most time- and space-efficient baseline needs about $6.5$ hours and uses $28\%$ more memory. We also show the efficiency of our approach in \kDRQ and its
usefulness in \CPM, \CSQ, and a phylogenetic tree annotation application. 
 In \CPM, our approach discovers patterns that reveal behavioral preferences about movies, books, or products, which prevail in different user subpopulations and are reflected in the literature, in \CSQ it efficiently distinguishes between DNA sequences that belong to different entities, and in the last application it helps genotype-phenotype links detection. 
See \cref{sec:experiments}.  
\end{enumerate}

\cref{sec:back} provides the background,~\cref{sec:baselines} baselines, and~\cref{sec:related} the related work.
We conclude in Section~\ref{sec:conclusion}.

\section{Background}\label{sec:back}

\paragraph{From \RM to \SM}  Range Mode is by now a classic problem in data structures theory~\cite{DBLP:journals/njc/KrizancMS05,DBLP:conf/icalp/GreveJLT10,DBLP:journals/corr/abs-1101-4068,DBLP:journals/tcs/DurocherEMT15,DBLP:journals/algorithmica/DurocherSST16,DBLP:conf/esa/El-Zein0MS18,DBLP:conf/soda/WilliamsX20,DBLP:conf/icalp/SandlundX20,DBLP:conf/icalp/Gu0WX21}.
It is defined as follows.

\defDSproblem{Range Mode (\RM)}{An array $\mathcal{A}$ of $N$ elements colored from a set
$\{0, \ldots, \Delta-1 \}$ of colors (integers).}{Given an interval $[i,j]$, output the most frequent color among the colors in $\mathcal{A}[i\dd j]$ (breaking ties arbitrarily).}

\SM can be seen as a specialization of \RM on leaf-colored trees. 
The fundamental difference is that in \RM we have $\Theta(N^2)$ distinct queries (one per interval $[i,j]$), while in \SM we have $\cO(N)$ possible query intervals that form a laminar family, that is, for every two intervals,
either the intervals are disjoint or one is contained in the other.
Although our $\cO(N)$-time construction algorithm (\Cref{the:linear-time}) 
precomputes the answer of each possible query,
we opted to define \SM as a data structure (instead of an algorithmic) problem to be consistent with \RM. 
Moreover, it might be possible to have a data structure of $o(N)$ size supporting (near-)optimal \SM queries or tree updates. Natural variations of \SM (similar to those of \RM) output both the mode $c_v^{\max}$ 
and its frequency, 
or the least frequent color, $c_v^{\min}$, 
known as \emph{anti-mode}, and its frequency $f_v^{\min}$; note that for a node $v$, the anti-mode $c_v^{\min}$ may have $f_v^{\min}=0$. 

For simplicity, our algorithm is presented for the \SM variation that returns the mode and its frequency but, as we show, it can be modified to compute instead the anti-mode and its frequency.

\paragraph{Strings.}~An \emph{alphabet} $\Sigma$ is a finite set of elements called \emph{letters}.
We consider throughout an \emph{integer} alphabet $\Sigma=[0,\sigma)$.   
For a string $S=S[0\dd n-1]$ over alphabet $\Sigma$, we denote its length~$n$ by $|S|$ and its $i$-th letter by $S[i]$.
By $\Sigma^q$, for some integer $q>0$, we denote the set of length-$q$ strings over $\Sigma$; a \emph{$q$-gram} is a string from $\Sigma^q$.
A \emph{substring} of $S$ starting at position $i$ and ending at position $j$ of $S$ is denoted by $S[i\dd j]$.  
A \emph{prefix} (respectively, \emph{suffix}) of $S$ is a substring of the form $S[0 \dd j]$ ($S[i \dd n-1])$.
A substring $P$ of $S$ may have multiple occurrences in $S$. 
We thus characterize an \emph{occurrence} of $P$ in $S$ by its \emph{starting position} $i\in[0,n-1]$; i.e., $P=S[i\dd i+|P|-1]$. The set of occurrences of $P$ in~$S$ is denoted by $\occ_S(P)$ and its size $|\occ_S(P)|$ is called the \emph{frequency} of $P$ in $S$. For convenience, we assume that $S$ always ends with a terminating letter $\$$ that occurs only at the last position of $S$ and is the lexicographically smallest letter. 

\paragraph{Compacted Tries and Suffix Trees.}~A \emph{compacted trie} is a trie in which each maximal branchless path is replaced with a single edge whose label is a string equal to the concatenation of that path's edge labels. The dissolved nodes are called \emph{implicit} while the preserved nodes are called \emph{explicit}. A node with at least two children is called \emph{branching}. For a node $v$ in a compacted trie, $\str(v)$ is the concatenation of edge labels on the root-to-$v$ path.
We define the \emph{string depth} of a node $v$ as $\sd(v) = |\str(v)|$.
The \emph{locus} of a pattern $P$ is the node $v$ with the smallest string depth such that $P$ is a prefix of $\str(v)$.
The \emph{suffix tree} of a string $S$, 
denoted by $\ST(S)$, is the compacted trie of all suffixes of $S$~\cite{DBLP:conf/focs/Weiner73}. 

 \begin{example}
Let $S=\texttt{banana\$}$. The suffix tree $\ST(S)$ is in Fig.~\ref{fig:st}. The edge labeled \texttt{na} replaces two edges labeled \texttt{n} and~\texttt{a} in the (uncompacted) trie. For node $v$, $\str(v)=\texttt{ana}$, which is the concatenation of the edge labels $\texttt{a}$ and $\texttt{na}$ on the path from the root to $v$. The string depth of $v$ is $\sd(v)=|\texttt{ana}|=3$. The locus of pattern $P=\texttt{an}$ is $v$, since $v$ is the node with the smallest string depth such that $P$ is a prefix of $\str(v)=\texttt{ana}$.
\end{example}

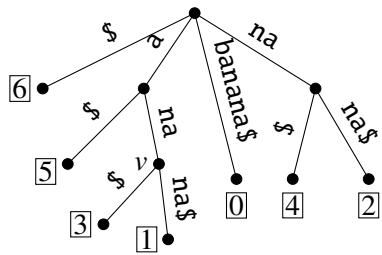
\begin{figure}[tbhp]
  \centering
\hspace{-0.1mm}
\scalebox{1.0}{
 \begin{tikzpicture}[- , level distance=1.5cm]
\filldraw[black] (0,0) circle (2pt); 
\filldraw[black] (-2,-1) circle (2pt);
\filldraw[black] (1.59,-1) circle (2pt);
\filldraw[black] (-1.67,-2) circle (2pt);
\filldraw[black] (-0.67,-1) circle (2pt);
\filldraw[black] (-0.47,-2) circle (2pt);
\filldraw[black] (-1.2,-2.8) circle (2pt);
\filldraw[black] (0.55,-2.2) circle (2pt);
\filldraw[black] (-0.35,-3) circle (2pt);
\filldraw[black] (1.29,-2.2) circle (2pt);
\filldraw[black] (2.29,-2.2) circle (2pt);
  \draw (0,0)   -- (-2,-1) node[midway,sloped,above] {$\texttt{\$}$};
  \draw (0,0)   -- (-0.67,-1) node[midway,sloped,above] {$\texttt{a}$};
  \draw (0,0)   -- (0.55,-2.2) node[midway,sloped,above] {$\texttt{banana\$}$};
  \draw (0,0) -- (1.59,-1)   node[midway,sloped,above] {$\texttt{na}$};
    \draw (-1.67,-2)   -- (-0.67,-1) node[midway,sloped,above] {$\texttt{\$}$};
        \draw (-0.47,-2)   -- (-0.67,-1) node[midway,sloped,above] {$\texttt{na}$};
           \draw (-0.47,-2)   -- (-1.2,-2.8) node[midway,sloped,above] {$\texttt{\$}$};
              \draw (-0.47,-2)   -- (-0.35,-3) node[midway,sloped,above] {$\texttt{na\$}$};
    \draw (1.59,-1) -- (1.29,-2.2)   node[midway,sloped,above] {$\texttt{\$}$};
    \draw (1.59,-1) -- (2.29,-2.2)   node[midway,sloped,above] {$\texttt{na\$}$};
    
\node[draw, rectangle, fill=white, inner sep=1pt] at (-1.96, -2.1) {5}; 
\node[draw, rectangle, fill=white, inner sep=1pt] at (-2.3, -1) {6};
\node[draw, rectangle, fill=white, inner sep=1pt] at (-1.5, -2.8) {3};
\node[draw, rectangle, fill=white, inner sep=1pt] at (-0.63, -3) {1};branching
\node[draw, rectangle, fill=white, inner sep=1pt] at (0.55, -2.55) {0};
\node[draw, rectangle, fill=white, inner sep=1pt] at (1.29, -2.55) {4};\node[draw, rectangle, fill=white, inner sep=1pt] at (2.29, -2.55) {2};

\node[] at (-0.7,-2) {$v$};

\end{tikzpicture}}
\caption{Suffix tree $\ST(S)$ for $S=\texttt{banana\$}$; the squares denote  starting positions in $S$.  
}\label{fig:st}
\end{figure}

\begin{lemma}[\cite{DBLP:conf/focs/Farach97}]\label{lem:ST}
Given a string $S$ of length $n$ over an integer alphabet of size $n^{\cO(1)}$, 
the suffix tree $\ST(S)$ of $S$ can be constructed in $\cO(n)$ time.
\end{lemma}

At each node of $\ST(S)$, we can store a hash table to access an edge based on the first letter of its label.
The hash tables can be constructed in $\cO(n)$ total time with high probability and support $\cO(1)$-time queries~\cite{DBLP:conf/sosa/BenderFKK24}. 
The total size of the hash tables is also bounded by $\cO(n)$~\cite{DBLP:conf/sosa/BenderFKK24}.
Spelling a pattern $P$ in a suffix tree $\ST(S)$ then takes $\cO(|P|)$ time~\cite{DBLP:conf/focs/Weiner73}: we start from the root, traverse down the tree edge by edge, matching as many letters as possible, until either the pattern ends, a mismatch occurs, or we land on a leaf. 

\section{Baselines for \SM}\label{sec:baselines}

We can address \SM by traversing $\mathcal{T}$ bottom-up and annotating each node with the accumulated color frequencies of its children. Then, for every node $v$ of $\mathcal{T}$, we find the 
mode $c^{\max}_v$ and its frequency $f^{\max}_v$.
We formalize this approach below.  

\begin{proposition}\label{thm:bl1}
Given a tree $\mathcal{T}$ on $N$ nodes, we can construct, in $\cO(N\Delta)$ time, a data structure that can answer any \SM query in $\cO(1)$ time.
In particular, for a given node $v$, 
the query algorithm returns both the mode $c^{\max}_v$ and its frequency~$f^{\max}_v$.
\end{proposition}
\begin{proof}
We initialize an integer array $A_v$ of size $\Delta$ for each node~$v$ of $\mathcal{T}$. 
In the base case ($v$ is a leaf of color $i$), we set 
$A_v[i]=1$ and $A_v[j]=0$, for each $j\neq i$.  
Using a bottom-up traversal of $\mathcal{T}$, we record the count $A_v[i]$ of each color~$i$: for an internal node $v$, with children $u_1,\ldots,u_{\ell}$, $A_v[i]:=\sum^{\ell}_{j=1}A_{u_j}[i]$.
Each value $A_v[i]$ is written once using the children of $v$ and read once by the parent of $v$.
Then, from each $A_v$, we derive $c^{\max}_v$ and $f^{\max}_v$ by reading each value once more. 
We have $\cO(N\Delta)$ values in total; the result follows.   
\end{proof}

We refer to the construction algorithm underlying \Cref{thm:bl1} as \textsf{Baseline 1}. 
The space used by \textsf{Baseline 1} is bounded by the construction time, which is $\cO(N\Delta)$.

Another way to solve \SM is by observing that it can be reduced to \RM by defining a single array that stores the color of each leaf of $\mathcal{T}$ from left to right, and then solving \RM on this array (i.e., preprocessing this array by constructing a data structure for \RM on it and then querying this data structure). This can be done because each subtree of $\mathcal{T}$ in \SM corresponds to a range of the array; and each range can be processed by a single query in the data structure for \RM. The drawback of this approach, however, is that constructing the data structure and answering each of the $N$ queries takes $\cO(N\sqrt{N})$ time, as shown below,  whereas \SM can be solved in the optimal $\cO(N)$ time by using our algorithm.  
We first recall a well-known result on \RM, and
then formalize this approach as follows. 

\begin{lemma}[\cite{DBLP:journals/corr/abs-1101-4068}]\label{lem:range-mode}
Given an array $\mathcal{A}$ on $N$ elements, for any $s\in [1,N]$, we can construct, 
in $\cO(sN)$ time, a data structure that can answer any \RM query in $\cO(N/s)$ time.
In particular, for a given range, 
the query algorithm returns both the most frequent color in the range and its frequency.
\end{lemma}

\begin{proposition}\label{thm:b2}
Given a tree $\mathcal{T}$ on $N$ nodes, we can construct, in $\cO(N\sqrt{N})$ time, a data structure that can answer any \SM query in $\cO(1)$ time.
In particular, for a given node $v$, 
the query algorithm returns both the mode $c^{\max}_v$ and its frequency $f^{\max}_v$.
\end{proposition}
\begin{proof}
The leaf nodes of $\mathcal{T}$ read in an in-order traversal induce an array of colors.
Similarly, every subtree of $\mathcal{T}$ induces a range on this array, as in the in-order traversal the colors of its leaves are stored consecutively in the array.   
The array and the ranges can be precomputed and stored via an in-order traversal on $\mathcal{T}$.
We apply \cref{lem:range-mode} with $s=\sqrt{N}$ and ask $N$ queries; a query corresponds to a range and takes $\cO(\sqrt{N})$ time. 
\end{proof}

We call \textsf{Baseline 2} the construction algorithm underlying 
\Cref{thm:b2}. 
Its benefit is that it does not depend on $\Delta$ and thus it is faster than \textsf{Baseline 1} for $\Delta = \omega( \sqrt{N})$. The space used by \textsf{Baseline 2} is bounded by the construction time, which is $\cO(N\sqrt{N})$.

\section{A Linear-Time Construction Algorithm for \SM}\label{sec:linear}

Problems like \SM, where statistics need to be computed for every subtree, 
are usually solved using the folklore
\emph{smaller-to-larger} technique (cf.~\cite{DBLP:books/daglib/0020103}) -- also known as \emph{disjoint set union} -- on trees.
Such an approach typically requires $\cO(N\log N)$ time (or slightly more depending on the used data structures).

We present our $\cO(N)$-time construction algorithm underlying \Cref{the:linear-time}
and show how it
can be modified to return the
anti-mode of a given node of $\mathcal{T}$ within the same complexities. 
Note that, when referring to an \emph{ancestor} (resp.,~\emph{descendant}) of a node $v$, we mean $v$ included 
unless stated otherwise, in which case this is referred to as \emph{strict} ancestor (resp.,~descendant).

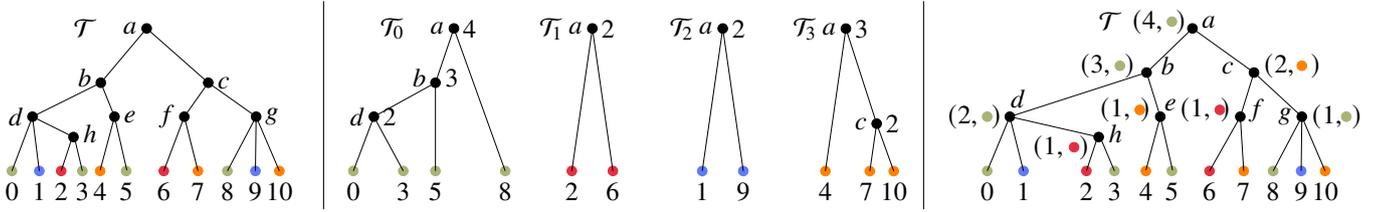
\begin{figure*}[t]
\begin{subfigure}[b]{0.23 \linewidth}
\scalebox{0.9}{
\begin{tikzpicture}[- , level distance=1.5cm]
\filldraw[black] (0,0) circle (2pt); 
\filldraw[black] (0.9,-0.8) circle (2pt);
\filldraw[black] (-1.67,-1.3) circle (2pt);  
\filldraw[black] (-0.67,-0.8) circle (2pt);
\filldraw[black] (-0.47,-1.3) circle (2pt);
\filldraw[black]  (0.55,-1.3) circle (2pt);   
\filldraw[black] (1.6,-1.3) circle (2pt);   
\filldraw[black] (-1.07,-1.6) circle (2pt); 

\filldraw[CWIContrGreen] (-1.97,-2.1) circle (2pt);  
\filldraw[CWIBlue] (-1.57,-2.1) circle (2pt); 
\filldraw[CWIRed] (-1.25,-2.1) circle (2pt);  
\filldraw[CWIContrGreen] (-0.94,-2.1) circle (2pt); 
\filldraw[orange] (-0.68,-2.1) circle (2pt);
\filldraw[CWIContrGreen] (-0.3,-2.1) circle (2pt);
\filldraw[CWIRed]  (0.25,-2.1) circle (2pt);   
\filldraw[orange]  (0.75,-2.1) circle (2pt);   
\filldraw[CWIContrGreen]  (1.18,-2.1) circle (2pt);   
\filldraw[CWIBlue]  (1.59,-2.1) circle (2pt);   
\filldraw[orange]  (1.94,-2.1) circle (2pt);   

\node at (-1.97,-2.4) {$0$};
\node at (-1.57,-2.4) {$1$};
\node at (-1.25,-2.4) {$2$};
\node at (-0.94,-2.4) {$3$};
\node at (-0.68,-2.4) {$4$};
\node at (-0.3,-2.4) {$5$};
\node at (0.25,-2.4) {$6$};
\node at (0.75,-2.4) {$7$};
\node at (1.18,-2.4) {$8$};
\node at (1.59,-2.4) {$9$};
\node at (1.94,-2.4) {$10$};

\node[left] at (-0.66,-0.73) {$b$};
\node[left] at (0,0) {$a$};
\node[right] at (0.9,-0.8) {$c$};
\node[left] at (-1.67,-1.3) {$d$};
\node[right] at (-0.47,-1.3) {$e$};
\node[right] at (-1.07,-1.55) {$h$};
\node[left] at (0.55,-1.3) {$f$};
\node[right] at (1.6,-1.3) {$g$};
\node at (-0.9,0) {$\mathcal{T}$};

\draw (0,0) -- (-0.67,-0.8);
\draw (0,0) -- (0.9,-0.8);
\draw (-1.67,-1.3) -- (-0.67,-0.8);
\draw (-0.47,-1.3) -- (-0.67,-0.8);
\draw (0.9,-0.8) -- (0.55,-1.3);
\draw (0.55,-1.3) -- (0.25,-2.1);
\draw (0.55,-1.3) -- (0.75,-2.1);
\draw (0.9,-0.8) -- (1.6,-1.3);
\draw (-0.47,-1.3) -- (-0.29,-2.09);
\draw (-0.47,-1.3) -- (-0.68,-2.09);
\draw (-1.67,-1.3) -- (-1.57,-2.1);
\draw (-1.67,-1.3) -- (-1.97,-2.1);
\draw (-1.67,-1.3) -- (-1.07,-1.6);
\draw (-1.07,-1.55) -- (-0.94, -2.1);
\draw (-1.07,-1.55) -- (-1.25, -2.1);
\draw (1.6, -1.3) -- (1.18,-2.1);
\draw (1.6, -1.3) -- (1.59,-2.1);
\draw (1.6, -1.3) -- (1.94,-2.1);
\end{tikzpicture}}
\end{subfigure}
\hspace{+0.5mm}
\rulesep 
\begin{subfigure}[b]{0.13\columnwidth}
\scalebox{0.9}{
\begin{tikzpicture}[- , level distance=1.5cm]
\filldraw[black] (0,0) circle (2pt); 
\filldraw[black] (-1.17,-1.3) circle (2pt);  
\filldraw[black] (-0.27,-0.8) circle (2pt);

\filldraw[CWIContrGreen] (-1.47,-2.1) circle (2pt);  
\filldraw[CWIContrGreen] (-0.74,-2.1) circle (2pt); 
\filldraw[CWIContrGreen] (-0.27,-2.1) circle (2pt);
\filldraw[CWIContrGreen]  (0.75,-2.1) circle (2pt);   

\node[left] at (-0.26,-0.73) {$b$};
\node[right] at (-0.26,-0.73) {$3$};
\node[left] at (0,0) {$a$};
\node[right] at (0,0) {$4$};
\node[left] at (-1.17,-1.3) {$d$};
\node[right] at (-1.17,-1.3) {$2$};
\draw (0,0) -- (-0.27,-0.8);
\draw (0,0) -- (0.75,-2.1);
\draw (-1.17,-1.3) -- (-0.27,-0.8);
\draw (-1.17,-1.3) -- (-0.74,-2.1);

\draw (-0.27,-0.8) -- (-0.27,-2.1);
\draw (-1.17,-1.3) -- (-1.47,-2.1);

\node at (-1.47,-2.4) {$0$};
\node at (-0.74, -2.4) {$3$};
\node at (-0.27,-2.4) {$5$};
\node at (0.75,-2.4) {$8$};
\node at (-0.9,0) {$\mathcal{T}_0$};

\end{tikzpicture}}
\end{subfigure}\hspace{+1mm}
\begin{subfigure}[b]{0.09\linewidth}
\scalebox{0.9}{
\begin{tikzpicture}[- , level distance=1.5cm]
\filldraw[black] (-0.3,0) circle (2pt); 
\filldraw[CWIRed] (-0.6,-2.1) circle (2pt);
\filldraw[CWIRed]  (0,-2.1) circle (2pt);   
\node[left] at (-0.3,0) {$a$};
\node[right] at (-0.3,0) {$2$};
\draw (-0.3,0) -- (-0.6,-2.1);
\draw (-0.3,0) -- (0,-2.1);
\node at (-0.6,-2.4) {$2$};
\node at (0,-2.4) {$6$};
\node at (-0.9,0) {$\mathcal{T}_1$};
\end{tikzpicture}}
\end{subfigure}
\begin{subfigure}[b]{0.09\linewidth}
\scalebox{0.9}{
\begin{tikzpicture}[- , level distance=1.5cm]
\filldraw[black] (-0.3,0) circle (2pt); 
\filldraw[CWIBlue] (-0.6,-2.1) circle (2pt);
\filldraw[CWIBlue]  (0,-2.1) circle (2pt);   

\node[left] at (-0.3,0) {$a$};
\node[right] at (-0.3,0) {$2$};
\draw (-0.3,0) -- (-0.6,-2.1);
\draw (-0.3,0) -- (0,-2.1);

\node at (-0.6,-2.4) {$1$};
\node at (0,-2.4) {$9$};
\node at (-0.9,0) {$\mathcal{T}_2$};
\end{tikzpicture}}
\end{subfigure}%
\begin{subfigure}[b]{0.09\linewidth}
\scalebox{0.9}{
\begin{tikzpicture}[- , level distance=1.5cm]
\filldraw[black] (-0.3,0) circle (2pt); 
\filldraw[black] (0.15,-1.4) circle (2pt); 
\filldraw[orange] (-0.6,-2.1) circle (2pt);
\filldraw[orange]  (0.4,-2.1) circle (2pt);   
\filldraw[orange]  (0.05,-2.1) circle (2pt);   

\node[left] at (-0.3,0) {$a$};
\node[right] at (-0.3,0) {$3$};
\node[left] at (0.15,-1.4) {$c$};
\node[right] at (0.15,-1.4) {$2$};

\draw (-0.3,0) -- (-0.6,-2.1);
\draw (-0.3,0) -- (0.15,-1.4);
\draw (0.15,-1.4) -- (0.4,-2.1);
\draw (0.15,-1.4) -- (0.05,-2.1);

\node at (-0.6,-2.4) {$4$};
\node at (0,-2.4) {$7$};
\node at (0.4,-2.4) {$10$};
\node at (-0.9,0) {$\mathcal{T}_3$};
\end{tikzpicture}}
\end{subfigure}
\hspace{+1mm}
\rulesep 
\begin{subfigure}[b]{0.24\linewidth}
\scalebox{0.9}{
\begin{tikzpicture}[- , level distance=1.5cm]
\filldraw[black] (0,0) circle (2pt); 
\filldraw[black] (0.9,-0.65) circle (2pt);
\filldraw[black] (-2.67,-1.3) circle (2pt);  
\filldraw[black] (-0.67,-0.65) circle (2pt);
\filldraw[black] (-0.47,-1.3) circle (2pt);
\filldraw[black]  (0.7,-1.3) circle (2pt);   
\filldraw[black] (1.6,-1.3) circle (2pt);   
\filldraw[black] (-1.37,-1.6) circle (2pt); 

\filldraw[CWIContrGreen] (-3,-2.1) circle (2pt);  
\filldraw[CWIBlue] (-2.47,-2.1) circle (2pt); 
\filldraw[CWIRed] (-1.55,-2.1) circle (2pt);  
\filldraw[CWIContrGreen] (-1.14,-2.1) circle (2pt); 
\filldraw[orange] (-0.68,-2.1) circle (2pt);
\filldraw[CWIContrGreen] (-0.3,-2.1) circle (2pt);
\filldraw[CWIRed]  (0.25,-2.1) circle (2pt);   
\filldraw[orange]  (0.75,-2.1) circle (2pt);   
\filldraw[CWIContrGreen]  (1.18,-2.1) circle (2pt);   
\filldraw[CWIBlue]  (1.59,-2.1) circle (2pt);   
\filldraw[orange]  (1.94,-2.1) circle (2pt);   

\node at (-3,-2.4) {$0$};
\node at (-2.47,-2.4) {$1$};
\node at (-1.55,-2.4) {$2$};
\node at (-1.14,-2.4) {$3$};
\node at (-0.68,-2.4) {$4$};
\node at (-0.3,-2.4) {$5$};
\node at (0.25,-2.4) {$6$};
\node at (0.75,-2.4) {$7$};
\node at (1.18,-2.4) {$8$};
\node at (1.59,-2.4) {$9$};
\node at (1.94,-2.4) {$10$};

\node[left] at (-1.06,-0.55) {$(3,$};
\node[right] at (-0.6,-0.55) {$b$};
\filldraw[CWIContrGreen,left]  (-1.06,-0.55) circle (2pt);
\node[left] at (-0.7,-0.55) {$)$};

\node[left] at (-0.3,0.1) {$(4,$};
\node[right] at (0,0.1) {$a$};
\filldraw[CWIContrGreen,left]  (-0.3,0.1) circle (2pt);
\node[left] at (0.05,0.1) {$)$};

\node[right] at (0.9,-0.55) {$(2,$};
\node[left] at (0.75,-0.58) {$c$};
\filldraw[orange,right]  (1.6,-0.55) circle (2pt);
\node[right] at (1.6,-0.55) {$)$};

\node[left] at (-3,-1.3) {$(2,$};
\node[right] at (-2.8,-1.03) {$d$};
\filldraw[CWIContrGreen,left]  (-3,-1.3) circle (2pt);
\node[left] at (-2.65,-1.3) {$)$};

\node[left] at (-0.77,-1.2) {$(1,$};
\node[right] at (-0.55,-1.12) {$e$};
\filldraw[orange,left]  (-0.77,-1.2) circle (2pt);
\node[left] at (-0.47,-1.2) {$)$};

\node[left] at (-1.75,-1.75) {$(1,$};
\node[right] at (-1.37,-1.55) {$h$};
\filldraw[CWIRed,left]  (-1.73,-1.75) circle (2pt);
\node[left] at (-1.37,-1.75) {$)$};

\node[left] at (0.4,-1.2) {$(1,$};
\node[right] at (0.68,-1.2) {$f$};
\filldraw[CWIRed,left]  (0.4,-1.2) circle (2pt);
\node[left] at (0.7,-1.2) {$)$};

\node[right] at (1.6,-1.3) {$(1,$};
\node[left] at (1.6,-1.3) {$g$};
\filldraw[CWIContrGreen,right]  (2.25,-1.3) circle (2pt);
\node[right] at (2.2,-1.3) {$)$};

\node at (-1.2,0.1) {$\mathcal{T}$};

\draw (0,0) -- (-0.67,-0.65);
\draw (0,0) -- (0.9,-0.65);
\draw (-2.67,-1.3) -- (-0.67,-0.65);
\draw (-0.47,-1.3) -- (-0.67,-0.65);
\draw (0.9,-0.65) -- (0.7,-1.3);
\draw (0.7,-1.3) -- (0.25,-2.1);
\draw (0.7,-1.3) -- (0.75,-2.1);
\draw (0.9,-0.65) -- (1.6,-1.3);
\draw (-0.47,-1.3) -- (-0.29,-2.09);
\draw (-0.47,-1.3) -- (-0.68,-2.09);
\draw (-2.67,-1.3) -- (-2.47,-2.1);
\draw (-2.67,-1.3) -- (-3,-2.1);
\draw (-2.67,-1.3) -- (-1.37,-1.6);
\draw (-1.37,-1.55) -- (-1.14, -2.1);
\draw (-1.37,-1.55) -- (-1.55, -2.1);
\draw (1.6, -1.3) -- (1.18,-2.1);
\draw (1.6, -1.3) -- (1.59,-2.1);
\draw (1.6, -1.3) -- (1.94,-2.1);
\end{tikzpicture}}
\end{subfigure}
\caption{In Step $1$ of the algorithm, the single-color trees $\mathcal{T}_0, \ldots, \mathcal{T}_3$ are created from $\mathcal{T}$. Note that each (internal) node $v$ of~$\mathcal{T}_i$, for all $i\in[0,4)$, is annotated with one node $\phi_i(v)$ of $\mathcal{T}$; e.g., $\phi_i(v)=a$ for all root nodes $v$ in $\mathcal{T}_i$. In Step~$2$, every internal node in 
$\mathcal{T}_0, \ldots, \mathcal{T}_3$ 
stores the count of its leaf descendants. In Step 3, the internal nodes of $\mathcal{T}$ store (frequency, color) pairs.}\label{fig:split}
\end{figure*}

\paragraph{High-level Overview.}~The construction algorithm underlying \Cref{the:linear-time} consists of
three main steps (see also Fig.~\ref{fig:split}):
\begin{enumerate}
    \item {\it Splitting the Tree.}~The tree $\mathcal{T}$ is split  into a forest of $\Delta$ single-color trees, denoted by $\mathcal{T}_0,\ldots,\mathcal{T}_{\Delta-1}$: all leaves of~$\mathcal{T}_i$, for each $i\in[0,\Delta)$, are colored $i$.
    For each $i\in[0,\Delta)$, each internal node $v$ of  $\mathcal{T}_i$ is associated with one node $\phi_i(v)$ of $\mathcal{T}$ such that a node $u$ is an ancestor of a node $u'$ in $\mathcal{T}_i$ if and only if $\phi_i(u)$ is an ancestor of $\phi_i(u')$. 
    The definition of the mapping $\phi_i$ will be provided later. 
    \item {\it Counting Colors.}~The count of leaf descendants of every internal node $v$ in $\mathcal{T}_i$, for each $i\in [0, \Delta)$, is computed in a bottom-up manner in a traversal of $\mathcal{T}_i$ and stored at $v$. 
    \item {\it Merging Counts.}~In a traversal of $\mathcal{T}$ in a bottom-up manner, the algorithm computes, for every internal node~$v$,
    a pair comprised of:
    (I)
    the maximum of
    the counts stored at the children of $v$
    and the counts stored at the internal nodes~$u$ in $\mathcal{T}_i$ with $\phi_i(u)=v$, for all $i\in[0,\Delta)$;
    and (II) a
    color corresponding to this maximum count.
\end{enumerate}

This algorithm, henceforth referred to as \SCM (for Splitting-Counting-Merging), outputs, for every node $v$ of $\mathcal{T}$, the mode of $v$ and its
frequency; see \Cref{alg:scm} for the  pseudocode. 

\begin{algorithm}
\caption{$\SCM(\mathcal{T}, \Delta)$}\label{alg:scm}
\begin{algorithmic}[1]\footnotesize 
  \ForAll{$i \in [0,\Delta)$}\Comment{Step 1}
    \State Construct single-color tree $\mathcal{T}_i$     
  \EndFor
  \ForAll{$i \in [0, \Delta)$}\Comment{Step 2}
  \ForAll{internal nodes $u$ in $\mathcal{T}_i$}
    \State $\textit{count}[u] \gets \text{number of leaf descendants of } u$
    \EndFor
\EndFor
  \State Initialize array $\textit{mode}$ with $N$ $\text{-1}$'s and array $ \textit{freq}$ with $N$  $0$'s \Comment{Step 3}
  \ForAll{leaf nodes $v$ of color $c$}
    \State $( \textit{freq}[v], \textit{mode}[v]) \gets (1, c)$ 
  \EndFor
  \ForAll{internal nodes $v$ in a bottom-up traversal of $\mathcal{T}$}
  \State $\mathcal{C}_v\gets $ set of counts comprised of:
  \INDSTATE $\textit{freq}[u]$ for each child $u$ of $v$ in $\mathcal{T}$ 
  \INDSTATE $\textit{count}[u]$ for each $\mathcal{T}_i$ and each  internal node $u$ in $\mathcal{T}_i$ s.t. $\phi_i(u) = v$
  \State $f^{\max}_v \gets \max(\mathcal{C}_v)$
  and $c^{\max}_v\gets$ a color with frequency $f^{\max}_v$ 
  \State $\textit{freq}[v] \gets f^{\max}_v$ and $\textit{mode}[v] \gets c^{\max}_{v}$
\EndFor
   \State \Return the arrays $\textit{mode}$ and $\textit{freq}$ 
\end{algorithmic}
\end{algorithm}

\paragraph{Preprocessing the Tree.}~An \emph{upward} path in a rooted tree $\mathcal{T}$ 
is a path from some node of $\mathcal{T}$ to one of its ancestors. 
We make the following simple observation:

\begin{observation}\label{obs:unary}
    Consider a rooted tree $\mathcal{T}$ and a node $v_0$ in~$\mathcal{T}$.
    For the maximal upward path $v_0,\ldots,v_\ell$, such that $v_0$ has more than one children (or it is a leaf node) and $v_1,\ldots,v_\ell$ have exactly one child, we have
    \[(f^{\min}_{v_0}, f^{\max}_{v_0})=(f^{\min}_{v_1}, f^{\max}_{v_1})=\ldots=(f^{\min}_{v_\ell}, f^{\max}_{v_\ell}),\]
    where $f_{v_i}^{\max}$ and $f_{v_i}^{\min}$ denote, respectively, the frequencies of the mode and the anti-mode of node $v_i$ for each $i\in[0,\ell]$. 
\end{observation}
\begin{proof}
  Since the nodes $v_0,\ldots,v_\ell$ have the same set of leaf descendants, the result follows directly.
\end{proof}

Based on \cref{obs:unary}, we henceforth assume that $\mathcal{T}$ has no node with one child (i.e., no unary path). If this is \emph{not} the case, we contract every edge in every maximal unary upward path $v_0,\ldots,v_\ell$, thus dissolving $v_1,\ldots,v_\ell$.
Performing these contractions is necessary for our algorithm, as we explain in Step 1. 
This preprocessing step can be performed in-place in $\cO(N)$ time using a traversal of $\mathcal{T}$.
Afterwards, the mode and frequency for the removed nodes can easily be recovered from the mode of the surviving nodes.

\paragraph{Step 1: Splitting the Tree.}~Step $1$ takes a tree $\mathcal{T}$ on $N$ nodes, $N_L<N$ leaves, and $\Delta\leq N_L$ colors as input. 
It outputs $\Delta$ trees, $\mathcal{T}_0,\ldots,\mathcal{T}_{\Delta-1}$, each associated with an injective mapping
$\phi_i: V(\mathcal{T}_i) \rightarrow V(\mathcal{T})$, where $V(\mathcal{G})$ is the set of nodes of graph~$\mathcal{G}$. 
For each $i \in [0,\Delta)$, we define $\mathcal{T}_i$ and $\phi_i$ as follows:
$\mathcal{T}_i$ is the tree obtained from $\mathcal{T}$ by deleting each node that does not have a descendant colored $i$ and then dissolving any node with  
one child; and $\phi_i$ maps each node of $\mathcal{T}_i$ to its origin in $\mathcal{T}$. 
Note, for each $i$, the number of leaf descendants with color~$i$ of each node of~$\mathcal{T}$ that gets dissolved in the construction of~$\mathcal{T}_i$ is equal to that of its single child; e.g., nodes $h, e, f$, and $g$  in \cref{fig:split} are dissolved in the construction of each  
tree $\mathcal{T}_0, \ldots, \mathcal{T}_3$;
the frequencies of these nodes (per color) can be deduced by those of their children. The following properties hold:
\begin{enumerate}
    \item For each $i\in[0, \Delta)$, all the leaves of $\mathcal{T}_i$ are colored $i$.
    \item The leaf nodes of $\mathcal{T}$ are precisely the elements of
    \[\bigcup_{i \in [0,\Delta)} \bigcup_{\text{leaf } u \text{ of } \mathcal{T}_i} \phi_i(u).\]
    \item For each pair $u,v\in V(\mathcal{T}_i)$, $i\in[0, \Delta)$, $u$ is an ancestor of $v$ if and only if $\phi_i(u)$ is an ancestor of $\phi_i(v)$ in $\mathcal{T}$.
\end{enumerate}
The construction of the single-color trees is  performed in two phases which we detail below. 

\paragraph{Leaf Lists}
Let $O_{\mathcal{T}}=v_1,\ldots,v_{N_L}$ be the list of the leaf nodes of $\mathcal{T}$ in the order in which they are visited in an in-order traversal of $\mathcal{T}$. ($O_{\mathcal{T}}$ can be constructed in $\cO(N)$ time.)
For each $i\in[0, \Delta)$,
we create an (initially empty) leaf list $\mathcal{L}_i$ that will eventually store all leaf nodes of $\mathcal{T}$ colored $i$.
We construct the leaf lists in $\cO(N_L)$ total time by 
scanning $O_{\mathcal{T}}$ from left to right and, for each element $v$ of $O_{\mathcal{T}}$ with color $i$, appending a leaf $u$ with $\phi_i(u) := v$ to~$\mathcal{L}_i$. 

\paragraph{Trees}
For each leaf list $\mathcal{L}_i$, we construct a leaf-colored tree $\mathcal{T}_i$ using a single color $i$.
In particular, $\mathcal{T}_i$ is a tree with the elements of $\mathcal{L}_i$ as leaves and internal nodes being in one-to-one correspondence with the elements of the set $\{\LCA_{\mathcal{T}}(\phi_i(u),\phi_i(v)) : u,v \in \mathcal{L}_i\}$, where $\LCA_T(x,y)$ denotes the lowest common ancestor of two nodes in a tree~$T$.
Our goal is to construct $\mathcal{T}_i$ in $\cO(|\mathcal{L}_i|)$ time.
The tree $\mathcal{T}_i$ can be constructed in $\cO(|\mathcal{L}_i|)$ time using the algorithm of Kasai et al.~\cite[Section 5.2]{DBLP:conf/cpm/KasaiLAAP01}.
This algorithm simulates a traversal of any rooted tree $\mathcal{T}'$ with no unary paths, if one has:
(I) the leaf list of $\mathcal{T}'$ (from left to right);
(II) the \LCA{}'s of adjacent leaves in the list; and
(III) access to the partial order of these \LCA{} nodes (each specified by two leaves $x,y$ such that the node is $\LCA_\mathcal{T'}(x,y)$) defined as $u<v$ if $u$ is a strict ancestor of $v$.

We have already constructed the list $\mathcal{L}_i$ of leaves of $\mathcal{T}_i$.
To find the \LCA{}'s of the leaves in $\mathcal{L}_i$, in a preprocessing step, we construct a data structure for answering \LCA queries on~$\mathcal{T}$.
The construction takes $\cO(N)$ time~\cite{DBLP:conf/latin/BenderF00}.
Given any two nodes $u$ and~$v$ in $\mathcal{T}$, the data structure returns node $w=\LCA_{\mathcal{T}}(u,v)$ in $\cO(1)$ time.
We scan $\mathcal{L}_i$, from left to right, and ask $(|\mathcal{L}_i|-1)$ $\LCA$ queries in $\mathcal{T}$, between the nodes associated with each pair of successive nodes of $\mathcal{L}_i$.
This way, we compute at most $|\mathcal{L}_i|-1$ distinct internal nodes of $\mathcal{T}$; for each such node $w$, we create a node $v$ in $\mathcal{T}_i$ and set $\phi_i(v) := w$.
Given two distinct nodes $u_1$ and $u_2$ of $\mathcal{T}_i$, each specified by a pair of successive leaves of which it is the \LCA, we have that $u_1<u_2$ if and only if $\phi_i(u_1) \neq \phi_i(u_2)$ and $\LCA_\mathcal{T}(\phi_i(u_1),\phi_i(u_2)) = \phi_i(u_1)$; these conditions can be checked in $\cO(1)$ time.
Thus, using Kasai et al.'s algorithm, we construct $\mathcal{T}_i$ and $\phi_i$ in $\cO(|\mathcal{L}_i|)$ time. We have thus proved the following lemma.

\begin{lemma}\label{lem:trees}
 After $\cO(N)$-time preprocessing of $\mathcal{T}$, the trees $\mathcal{T}_0,\ldots,\mathcal{T}_{\Delta-1}$ and the mappings $\phi_0,\ldots,\phi_{\Delta-1}$ can be constructed in $\cO(\sum_{i\in[0,\Delta)}|\mathcal{L}_i|)$ total time. 
\end{lemma}

We also make the following simple observation.

\begin{observation}\label{obs:size}
The total size $\sum_{i\in[0,\Delta)}|\mathcal{T}_i|$ of the trees $\mathcal{T}_0,\ldots,\mathcal{T}_{\Delta-1}$
is $\sum_{i\in[0,\Delta)}|\mathcal{L}_i|=\cO(N)$.
\end{observation}
\begin{proof}
By the construction of the single-color trees, the trees $\mathcal{T}_0,\ldots,\mathcal{T}_{\Delta-1}$ have only branching and leaf nodes.
The total number of leaves in $\mathcal{T}_0,\ldots,\mathcal{T}_{\Delta-1}$ is $N_L$.
 Thus, the total size of $\mathcal{T}_0,\ldots,\mathcal{T}_{\Delta-1}$ is less than $2\cdot N_L=\cO(N)$.
\end{proof}

By \cref{lem:trees} and \cref{obs:size}, we obtain the following.

\begin{lemma}\label{lem:step1}
    Step 1 of the \SCM algorithm takes $\cO(N)$ time.
\end{lemma}

\paragraph{Step 2: Counting Colors.}~In Step 2, we count the leaf descendants of  every internal node $v$ in $\mathcal{T}_i$, for all $i\in[0, \Delta)$, and store the count at $v$.   We achieve this using a separate bottom-up traversal for every $\mathcal{T}_i$, $i\in [0, \Delta)$. 
As any bottom-up traversal can be implemented in linear time in the tree size and $\sum_{i\in[0,\Delta)} |\mathcal{T}_i|=\cO(N)$,
we obtain: 

\begin{lemma}\label{lem:step2}
    Step 2 of the \SCM algorithm  takes $\cO(N)$ time.
\end{lemma}

With \cref{lem:step1,lem:step2}, it is easy to obtain an $\cO(N\log \Delta)$-time construction algorithm. 
This is slower than \SCM, so 
we just provide the intuition:  
By performing the inverse operation of splitting, we can merge \emph{two trees} in linear time, and then employ \cref{thm:bl1} on each merged tree, which takes linear time because each merged tree consists of \emph{two colors}.
If we do this iteratively in $\log \Delta$ levels, with appropriate color renaming to ensure that each merged tree has leaves of two colors, the whole algorithm takes $\cO(N\log \Delta)$ total time and $\cO(N)$ space. 

\paragraph{Step 3: Merging Counts.}~It seems difficult to improve on the above-mentioned $\cO(N\log \Delta)$-time algorithm if we insist on the technique that merges two trees at a time.
The \emph{crucial observation} we make is that, instead of merging whole trees, 
we can merge the counts of the (at most) $\Delta$ nodes
in $\mathcal{T}_0,\ldots,\mathcal{T}_{\Delta-1}$ that are associated with the same node in $\mathcal{T}$ \emph{at once}.
More specifically, we make a \emph{single}  traversal of $\mathcal{T}$, processing nodes in a bottom-up manner. 
At every internal node $v$ of $\mathcal{T}$, we compute and store the maximum among
the counts stored by the children of $v$ and the counts stored at all internal nodes $u$ of $\mathcal{T}_i$ with $\phi_i(u)=v$, for $i\in[0,\Delta)$ (recall that the latter counts were computed in Step 2).
The base case is at the leaves, where we initialize the count to $1$ for the color of the leaf. We also store the corresponding most frequent color per node, breaking ties arbitrarily. 

Even if for a single node $v$ in $\mathcal{T}$, we have $\Delta$ nodes $u$
in $\mathcal{T}_0,\ldots,\mathcal{T}_{\Delta-1}$ with $\phi_i(u)=v$, 
the running time of Step 3 amortizes to $\cO(N)$ because: (1) accessing the associated node $u \in \mathcal{T}_i, i \in [0, \Delta),$ of a node $v=\phi_i(u) \in \mathcal{T}$ takes $\cO(1)$ time (e.g., when $\phi_i()$ is implemented by a  pointer from $u$ to $v$); (2) each node of any $\mathcal{T}_i, i\in [0,\Delta),$ is accessed \emph{once}, by the definition of $\phi_i()$; and (3) 
the total size of $\mathcal{T},\mathcal{T}_0,\ldots,\mathcal{T}_{\Delta-1}$ is $\cO(N)$, as explained above.  
Thus, we obtain the following:

\begin{lemma}\label{lem:step3}
    Step 3 of the \SCM algorithm  takes $\cO(N)$ time.
\end{lemma}

\begin{example}
Consider the tree $\mathcal{T}$ in Fig.~\ref{fig:split}. 
The first visited internal node of $\mathcal{T}$ is $h$. As there is no node $u$ in any $\mathcal{T}_i$ with $\phi_i(u)=h$, the two children of $h$ both store counts of $1$. Thus, by breaking ties arbitrarily, $h$ 
stores a pair (frequency, color) set to $(1,\textcolor{CWIRed}{\bullet})$, as shown in the right part of Fig.~\ref{fig:split}. 
The next visited internal node of $\mathcal{T}$ is $d$, so we choose the maximum among the counts $1, 1, 1$ coming from its three children in $\mathcal{T}$ (the $1$ for $h$ was computed above), and $2$ coming from the internal node $u$ with $\phi_0(u)=d$. Since $2$ is the maximum, we store $(2,\textcolor{CWIContrGreen}{\bullet})$ at $d$ in $\mathcal{T}$. 

The next visited internal node is $e$. 
By breaking ties arbitrarily, $e$ in~$\mathcal{T}$ stores $(1,\textcolor{orange}{\bullet})$. 
The next visited internal node is $b$, so we choose the maximum among $2$, $1$, and $3$;  
the first two of these counts come from its two children in $\mathcal{T}$,
and the last 
from the internal node $u$ with $\phi_0(u)=b$. 
Since $3$ is the maximum and its color is $\textcolor{CWIContrGreen}{\bullet}$, $b$ in $\mathcal{T}$ stores $(3,\textcolor{CWIContrGreen}{\bullet})$. 

The next visited internal node is $f$. 
By breaking ties arbitrarily, $f$ in $\mathcal{T}$ stores $(1,\textcolor{CWIRed}{\bullet})$. 
The next visited internal node is $g$. 
By breaking ties arbitrarily, $g$ in $\mathcal{T}$ stores $(1,\textcolor{CWIContrGreen}{\bullet})$. 
The next visited internal node is $c$, so we choose the maximum 
among $1$, $1$, and $2$; the first two counts come from the two children of $c$ in $\mathcal{T}$ and the last from 
the internal node $u$ in $\phi_3(u)=c$.  
Hence $c$ in~$\mathcal{T}$ stores $(2,\textcolor{orange}{\bullet})$. 
The next visited internal node is $a$ (the root), so we choose the maximum among $3, 2, 4, 2, 2$, and~$3$; the first two counts come from the two children of $a$ in $\mathcal{T}$ and the rest from the root nodes $u$ with $\phi_i(u)=a$, for $i\in[0,4)$.  
Hence $a$ in $\mathcal{T}$ stores $(4,\textcolor{CWIContrGreen}{\bullet})$.
At this point, the bottom-up traversal of $\mathcal{T}$ is completed, and 
every node $v$ in $\mathcal{T}$ stores $(f^{\max}_v,c^{\max}_v)$.
\end{example}

\paragraph{Correctness and Wrapping-up.}~We prove the correctness of \SCM 
and that it takes $\cO(N)$ time and space.

\begin{lemma}\label{lem:correct}
The \SCM algorithm solves $\SM$ correctly.    
\end{lemma}
\begin{proof}

It suffices to show that, for every internal node $v$ of~$\mathcal{T}$,
it is correct to take the maximum of the counts stored at the children of $v$ in $\mathcal{T}$ and the counts stored at all internal nodes~$u$ of $\mathcal{T}_0,\ldots,\mathcal{T}_{\Delta-1}$ with $\phi_i(u)=v$.
Fix a color~$i$.
If $v$ has at most one child with leaf descendants colored $i$, then
the maximum among the counts stored by the children of $v$ covers the case when $i$ is the mode.
Otherwise, $v$ is the \LCA of at least two distinct leaf descendants of $v$ colored $i$.
In this case, there is, by definition, an internal node $u$ in $\mathcal{T}_i$ with $\phi_i(u)=v$, which stores the number of leaf descendants of~$v$ colored $i$. 
\end{proof}

\begin{lemma}\label{lem:combinedtime}
    The \SCM algorithm takes $\cO(N)$ time and space.
\end{lemma}

\begin{proof} 
Each step of \SCM is executed once (see also \Cref{alg:scm}). 
Thus, the time follows from \Cref{lem:step1,lem:step2,lem:step3}. 
The space used by \SCM is bounded by the time, which is $\cO(N)$.
In particular, Steps 1 and 2 take $\cO(N)$ space by 
\cref{obs:size} (for the leaf lists and trees) and by~\cite{DBLP:conf/latin/BenderF00} (for the \LCA data structure); and Step 3 takes $\cO(N)$ space since the total size of $\mathcal{T}, \mathcal{T}_0, \ldots, \mathcal{T}_{\Delta-1}$ is $\cO(N)$ and we store $\cO(N)$ counts. 
\end{proof}

\Cref{lem:correct,lem:combinedtime} imply \Cref{the:linear-time}.

\paragraph{Anti-mode.}~In a preprocessing step, we compute, for every node $v$ of $\mathcal{T}$, 
the number $c(v)$ of distinct colors in the leaves of the subtree rooted at $v$, in $\cO(N)$ time using a well-known algorithm~\cite{DBLP:conf/cpm/Hui92}. 
Then, using a bottom-up traversal on $\mathcal{T}$, we find the nodes $u$ of $\mathcal{T}$, such that $c(u)<\Delta$ and $c(\textsf{parent}(u))=\Delta$. We call these the \emph{border} nodes of $\mathcal{T}$.
Two observations are immediate. First, every border node and all its strict descendants have an anti-mode with a \emph{zero frequency}. Second, the subtrees rooted at border nodes are pairwise disjoint. Using a bottom-up traversal per subtree, we compute a representative absent color per border node (and propagate it to all its descendants). 
Specifically, for each such subtree, we (re-)use a bit array $B$ of size $\Delta$ to record all colors encountered in its traversal and report the smallest color $i$ with $B[i]=0$. As we report one color per subtree, this takes time linear in the subtree size.
(We re-set all array entries to $0$ once the subtree is processed, and process the next subtree.) The total time is thus $\cO(N)$.
We then focus on the \emph{top} part of $\mathcal{T}$: the set of strict ancestors of the border nodes; every such ancestor has an anti-mode with a \emph{non-zero frequency}.

 To this end, we modify \SCM as follows. 
 The first two steps are identical.
 We change Step~3 from a bottom-up to a top-down traversal.
 The base case is now at the root, where we simply take as the answer, the pair $(f,i)$, where $i$ is an anti-mode of the root of $\mathcal{T}$ and $f$ its frequency 
 (breaking ties arbitrarily). 
 For an arbitrary node~$v$ of $\mathcal{T}$, we perform the following: 
 (I) Let $(f,i)$ be the pair stored at $v$.
 If $\mathcal{T}_i$ does not have a node $u$ with $\phi_i(u)=v$,
 we push $(f,i)$ downwards to its corresponding child (i.e., the child whose subtree has $f$ leaves colored $i$). 
 (II) If any~$\mathcal{T}_i$, with $i\in [0,\Delta)$,
 has a node~$u$ with $\phi_i(u)=v$, for each child $w$ of $u$, with count $f_w$, 
 we push the pair $(f_w,i)$ to the child of $v$ that is an ancestor of $\phi(w)$ 
 in $\mathcal{T}$---this node can be computed in constant time after a linear-time preprocessing of $\mathcal{T}$ for level ancestor queries~\cite{DBLP:journals/jcss/BerkmanV94}. 
 When all pairs are pushed downwards from $v$, each child of $v$ (lying on the top part of~$\mathcal{T}$)
 selects the pair $(f,i)$ with minimum $f$ among its list of pairs (breaking ties arbitrarily) and discards the rest.
 The total number of pairs pushed downwards is asymptotically linear in the total number of edges in the $\mathcal{T}_i$'s and $\mathcal{T}$, and thus linear in $N$.
 As noted, each such push can be performed in constant time and hence the total running time of the algorithm is $\cO(N)$.
 
 To see that the algorithm is correct, fix an edge $(\textsf{parent}(v),v)$ in $\mathcal{T}$ and a color $i$.
 We have three cases when descending from $\textsf{parent}(v)$ to $v$: 
 either the frequency of $i$ is non-zero in $\textsf{parent}(v)$ and zero in $v$;
 or the frequency of $i$ is the same in $\textsf{parent}(v)$ and $v$;
 or the frequency of $i$ is reduced from $\textsf{parent}(v)$ to $v$ to a non-zero value.
 The first case is handled by the preprocessing.
 The second case is handled by Step 3(I).
 In the third case, there exists a node $u$ in $\mathcal{T}_i$ with $\phi_i(u) = v$.
 This case is handled by Step 3(II).
 To conclude, if $i$ is the anti-mode of $v$, 
 the algorithm will consider $i$ as part of its minimum computation.
 We thus obtain the following result. 

\begin{theorem}\label{the:least}
Given a tree $\mathcal{T}$ on $N$ nodes, we can construct, in $\cO(N)$ time, 
a data structure that, given a node $v$ of $\mathcal{T}$ as a query, it returns the anti-mode $c^{\min}_v$ of $v$ in $\cO(1)$ time. In particular, for a given node $v$, 
the query algorithm returns both the anti-mode $c^{\min}_v$ and
its frequency $f^{\min}_v$.
\end{theorem}

We illustrate this result in an example.

\begin{example}
 We will compute the anti-mode for each node of the tree $\mathcal{T}$ in Fig.~\ref{fig:split}, in which $\Delta=4$ and $i=0, 1, 2, 3$ corresponds to $\textcolor{CWIContrGreen}{\bullet},  \textcolor{CWIRed}{\bullet}, \textcolor{CWIBlue}{\bullet}, \textcolor{orange}{\bullet}$. 
 During preprocessing,
 the border nodes are $d,e,f$ and $g$.
 For $d$ (and its descendants) we report $(0,\textcolor{orange}{\bullet})$; $\textcolor{orange}{\bullet}$ is the smallest color that is absent from the  subtree of $d$, as during the traversal we encounter $\textcolor{CWIContrGreen}{\bullet}$, $\textcolor{CWIRed}{\bullet}$, and $\textcolor{CWIBlue}{\bullet}$. 
 For $e$ (and its descendants), we report $(0,\textcolor{CWIRed}{\bullet})$.
 For $f$ (and its descendants), we report $(0,\textcolor{CWIContrGreen}{\bullet})$.
 For $g$ (and its descendants), we report $(0,\textcolor{CWIRed}{\bullet})$.
 The top part of $\mathcal{T}$ contains $a$, $b$, and $c$, the strict ancestors of the border nodes. 
 For the root (node $a$), we take as the answer $(2, \textcolor{CWIRed}{\bullet})$ (breaking ties arbitrarily). Since, for all $i\in [0,4)$, $\mathcal{T}_i$ has a node $u$ with $\phi_i(u)=a$, we proceed to 
 Step 3(II). For $\mathcal{T}_0$, 
 we push $(3,\textcolor{CWIContrGreen}{\bullet})$ to $b$, 
 as the child $w$ of $u$ with $\phi_0(w)=b$ has count $3$, and $b$ is a child of $a$ and ancestor of itself in $\mathcal{T}$. Similarly, we push $(1, \textcolor{CWIContrGreen}{\bullet})$ to $c$, as the child $w$ of $u$ with $\phi_0(w)=8$ has count $1$, and $c$ is the child of $a$ that is an  ancestor of $8$ in $\mathcal{T}$. 
 For $\mathcal{T}_1$, we push $(1,\textcolor{CWIRed}{\bullet})$ to $b$    and $(1,\textcolor{CWIRed}{\bullet})$ to $c$.
 For $\mathcal{T}_2$, we push $(1,\textcolor{CWIBlue}{\bullet})$ to $b$ and $(1,\textcolor{CWIBlue}{\bullet})$ to $c$, and for $\mathcal{T}_3$, 
 we push $(1,\textcolor{orange}{\bullet})$ to $b$ and $(2,\textcolor{orange}{\bullet})$ to $c$. 
 Thus, $b$ selects $(1,\textcolor{CWIRed}{\bullet})$ and $c$
  $(1,\textcolor{CWIRed}{\bullet})$ as these pairs have the minimum frequency  (breaking ties arbitrarily). Fig.~\ref{fig:anti-modecomp} shows the result when the algorithm terminates.

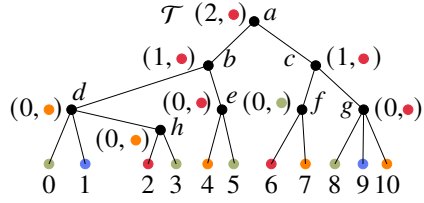
\begin{figure}[!ht]
\begin{center}
\scalebox{0.9}{
\begin{tikzpicture}[- , level distance=1.5cm]
\filldraw[black] (0,0) circle (2pt); 
\filldraw[black] (0.9,-0.65) circle (2pt);
\filldraw[black] (-2.67,-1.3) circle (2pt);  
\filldraw[black] (-0.67,-0.65) circle (2pt);
\filldraw[black] (-0.47,-1.3) circle (2pt);
\filldraw[black]  (0.7,-1.3) circle (2pt);   
\filldraw[black] (1.6,-1.3) circle (2pt);   
\filldraw[black] (-1.37,-1.6) circle (2pt); 

\filldraw[CWIContrGreen] (-3,-2.1) circle (2pt);  
\filldraw[CWIBlue] (-2.47,-2.1) circle (2pt); 
\filldraw[CWIRed] (-1.55,-2.1) circle (2pt);  
\filldraw[CWIContrGreen] (-1.14,-2.1) circle (2pt); 
\filldraw[orange] (-0.68,-2.1) circle (2pt);
\filldraw[CWIContrGreen] (-0.3,-2.1) circle (2pt);
\filldraw[CWIRed]  (0.25,-2.1) circle (2pt);   
\filldraw[orange]  (0.75,-2.1) circle (2pt);   
\filldraw[CWIContrGreen]  (1.18,-2.1) circle (2pt);   
\filldraw[CWIBlue]  (1.59,-2.1) circle (2pt);   
\filldraw[orange]  (1.94,-2.1) circle (2pt);   

\node at (-3,-2.4) {$0$};
\node at (-2.47,-2.4) {$1$};
\node at (-1.55,-2.4) {$2$};
\node at (-1.14,-2.4) {$3$};
\node at (-0.68,-2.4) {$4$};
\node at (-0.3,-2.4) {$5$};
\node at (0.25,-2.4) {$6$};
\node at (0.75,-2.4) {$7$};
\node at (1.18,-2.4) {$8$};
\node at (1.59,-2.4) {$9$};
\node at (1.94,-2.4) {$10$};

\node[left] at (-1.06,-0.55) {$(1,$};
\node[right] at (-0.6,-0.55) {$b$};
\filldraw[CWIRed,left]  (-1.06,-0.55) circle (2pt);
\node[left] at (-0.7,-0.55) {$)$};

\node[left] at (-0.3,0.1) {$(2,$};
\node[right] at (0,0.1) {$a$};
\filldraw[CWIRed,left]  (-0.3,0.1) circle (2pt);
\node[left] at (0.05,0.1) {$)$};

\node[right] at (0.9,-0.55) {$(1,$};
\node[left] at (0.75,-0.58) {$c$};
\filldraw[CWIRed,right]  (1.6,-0.55) circle (2pt);
\node[right] at (1.6,-0.55) {$)$};

\node[left] at (-3,-1.3) {$(0,$};
\node[right] at (-2.8,-1.03) {$d$};
\filldraw[orange,left]  (-3,-1.3) circle (2pt);
\node[left] at (-2.65,-1.3) {$)$};

\node[left] at (-0.77,-1.2) {$(0,$};
\node[right] at (-0.55,-1.12) {$e$};
\filldraw[CWIRed,left]  (-0.77,-1.2) circle (2pt);
\node[left] at (-0.47,-1.2) {$)$};

\node[left] at (-1.75,-1.75) {$(0,$};
\node[right] at (-1.37,-1.55) {$h$};
\filldraw[orange,left]  (-1.73,-1.75) circle (2pt);
\node[left] at (-1.37,-1.75) {$)$};

\node[left] at (0.4,-1.2) {$(0,$};
\node[right] at (0.68,-1.2) {$f$};
\filldraw[CWIContrGreen,left]  (0.4,-1.2) circle (2pt);
\node[left] at (0.7,-1.2) {$)$};

\node[right] at (1.6,-1.3) {$(0,$};
\node[left] at (1.6,-1.3) {$g$};
\filldraw[CWIRed,right]  (2.25,-1.3) circle (2pt);
\node[right] at (2.2,-1.3) {$)$};

\node at (-1.2,0.1) {$\mathcal{T}$};

\draw (0,0) -- (-0.67,-0.65);
\draw (0,0) -- (0.9,-0.65);
\draw (-2.67,-1.3) -- (-0.67,-0.65);
\draw (-0.47,-1.3) -- (-0.67,-0.65);
\draw (0.9,-0.65) -- (0.7,-1.3);
\draw (0.7,-1.3) -- (0.25,-2.1);
\draw (0.7,-1.3) -- (0.75,-2.1);
\draw (0.9,-0.65) -- (1.6,-1.3);
\draw (-0.47,-1.3) -- (-0.29,-2.09);
\draw (-0.47,-1.3) -- (-0.68,-2.09);
\draw (-2.67,-1.3) -- (-2.47,-2.1);
\draw (-2.67,-1.3) -- (-3,-2.1);
\draw (-2.67,-1.3) -- (-1.37,-1.6);
\draw (-1.37,-1.55) -- (-1.14, -2.1);
\draw (-1.37,-1.55) -- (-1.55, -2.1);
\draw (1.6, -1.3) -- (1.18,-2.1);
\draw (1.6, -1.3) -- (1.59,-2.1);
\draw (1.6, -1.3) -- (1.94,-2.1);
\end{tikzpicture}}
\end{center}
\caption{The tree $\mathcal{T}$ of Fig.~\ref{fig:split} annotated with (anti-mode, color) pairs; those of the leaves are the same as their parents.}\label{fig:anti-modecomp}
\end{figure}
\end{example}

\section{String-Processing Applications}\label{sec:applications}
In this section, we discuss the problems underlying  the string-processing applications of \SM and their solutions.

\paragraph{Top-1 Document Retrieval.}~
We formally define the \DRQ problem below.  

\defDSproblem{Top-$1$ Document Retrieval (\DRQ)} 
{A collection $\mathcal{S}$ of strings.}{Given a string $P$, output the string in $\mathcal{S}$ with the maximum number $\ell >0$ of occurrences of $P$ (breaking ties arbitrarily) or $-1$ if $P$ does not occur in any string in $\mathcal{S}$.} 

\begin{theorem}\label{thm:tb-1}
Let $\mathcal{S}$ be a collection of strings of total length $N$ over an integer alphabet $\Sigma$ of size $N^{\cO(1)}$. 
We can construct, in $\cO(N)$ time with high probability,  
a data structure that answers any \DRQ query $P\in\Sigma^m$ for $\mathcal{S}$ in $\cO(m)$ time.
\end{theorem}

\begin{proof}
Let the strings in $\mathcal{S}$ be $S_0, \ldots, S_{\Delta-1}$.
During preprocessing, we construct the suffix tree $\ST(S)$ for $S:=S_0\$_0\ldots S_{\Delta-1}\$_{\Delta-1}$, where each $\$_i\notin\Sigma$ is a unique delimiter, using \cref{lem:ST} (with hash tables), and color every leaf whose path-label starts at a position in $S_i\$_i$ with color $i$.
We conclude our preprocessing with an application of \cref{the:linear-time} on $\ST(S)$. This takes $\cO(N)$ time.
 Given a query pattern $P$, we spell $P$ in $\ST(S)$ in $\cO(m)$ time.
Say that we have
arrived at the explicit node $v$ (if we arrive at an implicit node, we 
take its nearest explicit descendant as $v$). 
If $v$ is branching, 
we return $c^{\max}_P:=c^{\max}_v$; else, it is a leaf colored $c_v$, in which case we return $c^{\max}_P:=c_v$.
If $P$ does not occur in~$S$, we return $-1$.
\end{proof}

Using \cref{the:least} on $\ST(S)$, we can solve the \textsf{Bottom-$1$ DR} problem, which  analogously to \textsf{Bottom} \kDRQ~\cite{sharma}, asks for the string from $\mathcal{S}$ that contains the least number of occurrences of $P$, within the complexities of \cref{thm:tb-1}.

\paragraph{Uniform Pattern Mining.}~The \CPM problem asks for all substrings (patterns) whose frequencies in
each pair of strings in $\mathcal{S}$ differ by at most $\epsilon$. 
We next define the notion of an $\epsilon$-uniform pattern in $\mathcal{S}$ and the \CPM problem: 

\begin{definition}
A string $P$ is \emph{$\epsilon$-uniform} in a collection of strings~$\mathcal{S}$, for an integer $\epsilon\geq 0$, if, for each pair of strings $S,S'\in \mathcal{S}$, it holds that  
$\left||\occ_S(P)|-|\occ_{S'}(P)|\right|\leq \epsilon$. \label{def:fairness-pairwise}
\end{definition}

\defproblem{Uniform Pattern Mining (\CPM)}
{A collection $\mathcal{S}$ of strings and an integer $\epsilon\geq 0$.}
{All $\epsilon$-uniform patterns in $\mathcal{S}$.}

\begin{observation}\label{obs:minmax}
A string $P$ is $\epsilon$-uniform in $\mathcal{S}$ if and only if 
$\max_{S\in \mathcal{S}}|\occ_S(P)|-\min_{S\in \mathcal{S}}|\occ_S(P)|\leq \epsilon$.  
\end{observation} 

\begin{theorem}\label{thm:fpm}
Consider an integer $\epsilon \geq 0$ and a collection~$\mathcal{S}$ of strings of total length $N$ over an integer alphabet $\Sigma$ of size~$N^{\cO(1)}$.
The 
\CPM problem 
can be solved in $\cO(N + \Output)$ time using $\cO(N)$ space, where $\Output$ is the output size.
\end{theorem}
\begin{proof}
Let the strings in $\mathcal{S}$ be $S_0, \ldots, S_{\Delta-1}$.
During preprocessing, we construct $\ST(S)$ for $S:=S_0\$_0\ldots S_{\Delta-1}\$_{\Delta-1}$, where each $\$_i\notin\Sigma$ is a unique delimiter, using \cref{lem:ST}, and color every leaf whose path-label starts at a position in $S_i\$_i$ with color $i$. 
We apply \cref{the:linear-time,the:least} on $\ST(S)$ in $\cO(N)$ time.
Then, we mark every explicit node~$v$ of $\ST(S)$, such that $f^{\max}_v - f^{\min}_v\leq \epsilon$. 
Then, we output, for each marked node $v$, each prefix of $\str(v)$ with length in $[|\str(\textrm{parent}(v))|+1,|\str(v)|]$.
These are precisely the $\epsilon$-uniform patterns in $\mathcal{S}$ due to \cref{obs:minmax} and the fact that, for any explicit node $v$, all implicit nodes along the edge from $\textrm{parent}(v)$ to $v$ have the same mode and anti-mode frequencies as $v$ since they have the same leaf descendants as $v$.
\end{proof}

\paragraph{Consistent Query String.}~The \CSQ  problem asks to 
count the distinct $q$-grams $Q\in\Sigma^q$, 
whose frequency in $P$ is in the interval $[{\min_{S\in \mathcal{S}}|\occ_S(Q)|-\epsilon}, \max_{S\in \mathcal{S}}|\occ_S(Q)|+\epsilon]$, for a string collection $\mathcal{S}$ and an integer $\epsilon\geq 0$. We call such a $q$-gram of $P$ \emph{$\epsilon$-consistent} with~$\mathcal{S}$ -- we may drop ``with $\mathcal{S}$'' when $\mathcal{S}$ is clear from the context. 

\defDSproblem{Consistent Query String (\CSQ)}
{A collection $\mathcal{S}$ of strings.}
{Given a string $P$ and integers $q>0$, $\epsilon\geq 0$, output the number of distinct $q$-grams of $P$ that are $\epsilon$-consistent.} 

\begin{theorem}\label{thm:q-cq}
Let $\mathcal{S}$ be a collection of strings with total length $N$ over an integer alphabet $\Sigma$ of size $N^{\cO(1)}$.
We can construct, in $\cO(N)$ time with high probability, 
a data structure that answers any \CSQ query with $P\in \Sigma^m$ in $\cO(m)$ time.
\end{theorem}

\begin{proof}
Let the strings in $\mathcal{S}$ be $S_0, \ldots, S_{\Delta-1}$.
During preprocessing, we construct the suffix tree $\ST(S)$ for $S:=S_0\$_0\ldots S_{\Delta-1}\$_{\Delta-1}$, where each $\$_i\notin\Sigma$ is a unique delimiter, using \cref{lem:ST} (with hash tables), and color every leaf whose path-label starts at a position in $S_i\$_i$ with color $i$.
We then apply \cref{the:linear-time} on $\ST(S)$ in $\cO(N)$ time.

Given a query with $P \in \Sigma^m$, we construct the suffix tree $\ST(P)$ using \cref{lem:ST}, and compute $|\occ_P(Q)|$, for every string $Q$ of length $q$ that occurs in $P$ in $\cO(m)$ total time.
We also use $\ST(P)$ to mark every position $i'$ of $P$ such that $P[i\dd i+q)=P[i'\dd i'+q)$, for some $i<i'$, to avoid double-counting.
We then spell~$P$ in $\ST(S)$ using suffix links~\cite{DBLP:books/cu/Gusfield1997}, to iterate, for increasing $i$, over the loci (in $\ST(S)$) of the longest prefixes of $P[i\dd i+q)$ that occur in $S$.
We maintain a counter, which is initialized as zero, as follows. If $Q:=P[i\dd i+q)$ occurs in $\mathcal{S}$, $i$ is not marked, and 
${|\occ_P(Q)|\in [f^{\min}_v-\epsilon,f^{\max}_v+\epsilon]}$, where $v$ is the node of $\ST(S)$ with path-label $Q$ (if this is an implicit node, we 
take its nearest explicit descendant as $v$), 
we increase the counter by one. 
The value of the counter is output after processing all of $P$.
Spelling $P$ takes $\cO(m)$ time~\cite{DBLP:books/cu/Gusfield1997};
for every $q$-gram $Q$, we make $\cO(1)$ elementary $\cO(1)$-time operations, and hence the query time follows. 
\end{proof}

\section{Generalizations of \SM}\label{sec:generalizations}

We describe two natural generalizations of \SM that are solved by straightforward modifications to our \SCM algorithm.

\paragraph{Node-colored Trees.}~One may wonder whether the algorithms for \SM
rely strictly on the fact that $\mathcal{T}$ has colors only on the leaves. This is \emph{not} the case. We call \NCSMQ the generalization of \SM in which \emph{all nodes} of $\mathcal{T}$ are colored.
The corollary below shows a simple linear-time reduction in which an algorithm for \NCSMQ can be gleamed 
from any algorithm for \SM.

\begin{corollary}
Given a tree $\mathcal{T}$ on $N$ nodes, we can construct, in $\cO(N)$ time, a data structure that can answer any \NCSMQ query in $\cO(1)$ time.
In particular, for a given node $v$,  
the query algorithm returns both the most frequent color in the subtree
rooted at $v$ and its frequency.
\end{corollary}
\begin{proof}
In $\cO(N)$ time, we transform the given instance of \NCSMQ to an instance of \SM on $\cO(N)$ nodes and apply \cref{the:linear-time}.
We construct a tree~$\mathcal{T}'$ from~$\mathcal{T}$ by: 
(I) attaching a new leaf child to each internal node and coloring it with that internal node's color, and
(II) removing all colors from internal nodes.
A direct application of \cref{the:linear-time} to $\mathcal{T}'$ then yields, in $\cO(N)$ time, a data structure for $\NCSMQ$ with query time $\cO(1)$.
\end{proof}

\paragraph{\texorpdfstring{$k$}{k} Most Frequent Colors.}~The $k$-\RM problem asks for the $k$ most frequent colors in a range of an array~\cite{DBLP:journals/corr/abs-1101-4068}. 
Analogously, we define the $k$-\SM problem asking for the $k$ most frequent colors in the leaves of the subtree rooted at a node of~$\mathcal{T}$ (breaking ties arbitrarily).
\cref{cor:k-most-frequent} is obtained by slightly modifying the \SCM algorithm from \cref{sec:linear}.

\begin{corollary}\label{cor:k-most-frequent}
Given a tree $\mathcal{T}$ on $N$ nodes, for any $k\leq \Delta$, 
we can construct, in $\cO(kN)$ time, a data structure that can answer any $k$-\SM query in $\cO(k)$ time.
In particular, for a given node $v$,   
the query algorithm returns both the $k$ most frequent colors in the leaves of the subtree
rooted at $v$ and their frequencies in sorted order.
\end{corollary}

\begin{proof}
    The \SCM algorithm 
    naturally extends for $k$-\SM.
    For every node $v$ of $\mathcal{T}$, in Step 3, we now store the $k$ most frequent colors (breaking ties arbitrarily) and their frequencies as a list $(c^1_v, f^1_v), \dots, (c^k_v, f^k_v)$.
    (If, for some node, we have fewer than $k$ colors, we leave some entries undefined.)
    We thus only need to verify that given $n>k$ integers we can select the $k$ largest ones 
    in $\cO(n)$ time to obtain the claimed generalization.
    We do this using the classic linear-time selection algorithm~\cite{DBLP:journals/jcss/BlumFPRT73}. 
    After computing one list $(c_1, f_1), \dots, (c_k, f_k)$ per node, we sort the lists by frequency using a single global radix sort.
    This takes $\cO(kN)$ time since each frequency is at most $N$.
\end{proof}

\Cref{cor:k-most-frequent} implies a linear-space data structure for any $k=\cO(1)$ for \kDRQ, which answers queries in optimal time \emph{and} can be constructed in linear time. In particular, we make the same reduction as in \Cref{thm:tb-1}, but instead of using \Cref{the:linear-time}, we use \Cref{cor:k-most-frequent}. We obtain the following result.

\begin{theorem}\label{thm:tb-k}
Let $\mathcal{S}$ be a collection of strings of total length $N$ over an integer alphabet $\Sigma$ of size $N^{\cO(1)}$.
For any $k=\cO(1)$, we can construct, in $\cO(N)$ time with high probability, a data structure that can answer any \kDRQ query $P\in\Sigma^m$ for $\mathcal{S}$ in $\cO(m)$ time.
\end{theorem}

\section{Lower Bounds for Descendant-Mode in DAGs}\label{sec:lb}

As \SM can be solved in linear time and it applies to a rooted tree, it is reasonable to ask whether the analogous problem on a
directed acyclic graph (DAG) defined below can also be solved fast (e.g., as fast as \SM or faster than \RM).

\defDSproblem{DAG-Descendant Mode (\DM)}
{A DAG $\mathcal{D}$ on $N$ nodes with every sink colored from a set $\{0, \ldots, \Delta-1\}$ of colors (integers).}
{Given a node $u$ of $\mathcal{D}$, output the most frequent color $c^{\max}_{u}$ in the sink descendants of $u$.}

We prove the following hardness result.

\begin{theorem}\label{the:subtree-union-bmm}
    If there exists a data structure for \DM with
    construction time $p(N)$ and query time $q(N)$, then the Boolean matrix multiplication problem for two $n \times n$ matrices admits an $\cO(p(n^2) + n^2 \cdot q(n^2))$-time solution.
\end{theorem}
\begin{proof}
Let $A,B$ be Boolean $n\times n$ matrices.  We build a sink-colored DAG $\mathcal D$ with $N=\Theta(n^2)$ nodes and $\Delta\leq n$ colors so that each entry of $AB$ is revealed by a specific \DM query.  

For each row index $i\in[0,n)$, we create an internal node $r(a_i)$ and attach, for every column index $k$ with $A_{i,k}=1$, a sink child of $r(a_i)$ colored $k$.  
For each column index $j\in[0,n)$, we create an internal node $r(b_j)$ and attach, for every row index~$k$ with $B_{k,j}=1$, a sink child of $r(b_j)$ colored $k$. 
For every pair $(j,i)\in[0,n)^2$, we create a node $y_{j,i}$ and add the following directed edges:
\[
y_{j,i}\to r(a_i)\qquad\text{and}\qquad y_{j,i}\to r(b_j).
\]
The graph $\mathcal D$ is a DAG because we can partition its nodes into three sets $Y = \{y_{j,i} : (j,i)\in[0,n)^2$\}, $R=\{r(x_i) : x\in \{a,b\}, i \in [0,n)\}$, and $S = V(\mathcal{D}) \setminus (Y \cup R)$, such that any edge is either from a node of $Y$ to a node of $R$ or from a node of $R$ to a node of $S$.
Further, the size of $\mathcal{D}$ is $\Theta(n^2)$.

The sink descendants of $y_{j,i}$ are exactly the multiset union of the sinks of $r(a_i)$ and $r(b_j)$. Hence any color $k$ with $A_{i,k}=B_{k,j}=1$ appears twice among those descendants, while any color that appears on only one side  appears at most once. Therefore a \DM query at $y_{j,i}$ returns some color $k$ with $A_{i,k}=B_{k,j}=1$ whenever such a $k$ exists; if no such $k$ exists every color multiplicity is at most $1$.

Thus, after constructing a \DM data structure on $\mathcal D$ in $p(N)$ time, we determine each entry $(AB)_{i,j}$ by querying \DM at $y_{j,i}$ to obtain a candidate color $k$ and then checking in $\cO(1)$ time whether $A_{i,k}=B_{k,j}=1$. There are $n^2$ queries, so the total time is $\cO(p(n^2)+n^2\cdot q(n^2))$, as claimed.
\end{proof}

\begin{example}
    Consider the following $2\times 2$ Boolean matrices:
\[
A = \begin{pmatrix} 1 & 0 \\ 0 & 1 \end{pmatrix},\qquad
B = \begin{pmatrix} 0 & 1 \\ 1 & 0 \end{pmatrix}.
\]

From matrices $A$ and $B$, we construct the DAG $\mathcal D$ in Fig.~\ref{fig:dag}. For instance, $(AB)_{i,j}=(AB)_{1,0}=1$ is obtained by querying
$y_{j,i}=y_{0,1}$, which returns the mode $k=1$. 
Indeed, $A_{i,k}=A_{1,1}=B_{k,j}=B_{1,0}=1$.
\end{example}

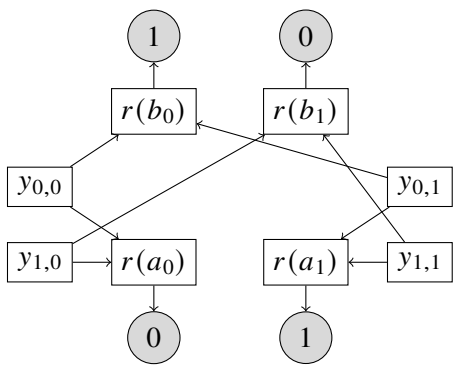
\begin{figure}[tbph]
\centering
\scalebox{1.0}{%
\hspace{-0.1cm}
\begin{tikzpicture}[
    node distance=1cm and 1cm,
    every node/.style={draw, rectangle, minimum width=0.7cm, minimum height=0.4cm},
    sink/.style={draw, circle, fill=gray!30, minimum size=0.4cm}
]

\node (a0) at (-0.5,0) {$r(a_0)$};
\node (a1) at (1.5,0) {$r(a_1)$};

\node[sink] (a0s0) at (-0.5,-1) {0};
\node[sink] (a1s1) at (1.5,-1) {1};

\node (b0) at (-0.5,2) {$r(b_0)$};
\node (b1) at (1.5,2) {$r(b_1)$};

\node[sink] (b0s1) at (-0.5,3) {1};
\node[sink] (b1s0) at (1.5,3) {0};

\node (y00) at (-2,1) {$y_{0,0}$};
\node (y01) at (3,1) {$y_{0,1}$};
\node (y10) at (-2,0) {$y_{1,0}$};
\node (y11) at (3,0) {$y_{1,1}$};

\draw[->] (y00) -- (a0);
\draw[->] (y00) -- (b0);

\draw[->] (y01) -- (a1);
\draw[->] (y01) -- (b0);

\draw[->] (y10) -- (a0);
\draw[->] (y10) -- (b1);

\draw[->] (y11) -- (a1);
\draw[->] (y11) -- (b1);

\draw[->] (a0) -- (a0s0);
\draw[->] (a1) -- (a1s1);

\draw[->] (b0) -- (b0s1);
\draw[->] (b1) -- (b1s0);

\end{tikzpicture}}
\caption{DAG $\mathcal D$ for $A$ and $B$.}\label{fig:dag}
\end{figure}

Theorem~\ref{the:subtree-union-bmm} implies that an algorithm for \DM that is similarly fast to our \SCM algorithm is highly unlikely, as $p(N)=\cO(N)$ and $q(N)=\cO(1)$ for \SM, and this would contradict the combinatorial BMM conjecture~\cite{AbboudFKLM24}. 
More generally, for any $\varepsilon >0$, there is no data structure for \DM with construction time $\cO(N^{\omega/2-\varepsilon})$ and sub-polynomial query time (unless the square matrix multiplication exponent~$\omega$ is $2$). 
Under the combinatorial BMM conjecture, for any $\varepsilon_1, \varepsilon_2\geq 0$, there is no combinatorial data structure for \DM with construction time $\cO(N^{3/2-\varepsilon_1})$ and query time $\cO(N^{1/2-\varepsilon_2})$. 
These conditional lower bounds are identical to the ones for \RM proved by Chan et al.~\cite{DBLP:journals/corr/abs-1101-4068}.

\section{Related Work}\label{sec:related} 
\paragraph{Range Mode (\RM).}~
Krizanc et al.~\cite{DBLP:journals/njc/KrizancMS05} proposed a data structure for \RM with $\cO(1)$-time queries and $\cO(N^2 \log \log N / \log N)$ space, and another with $\cO(\sqrt{N}\log N)$-time queries and $\cO(N)$ space. Chan et al.~\cite{DBLP:journals/corr/abs-1101-4068} improved the time/space trade-off of the data structures of~\cite{DBLP:journals/njc/KrizancMS05}
by designing an $\cO(N + s^2/w)$-space data structure with $\cO(N/s)$ query time, for any $s\in [1,N]$, where $w=\Omega(\log N)$ is the machine word size~\cite[Section~6]{DBLP:journals/corr/abs-1101-4068}. 
By setting $s:=\lceil\sqrt{Nw}\rceil$, said data structure takes $\cO(N)$ space and has $\cO(\sqrt{N/w})$-time queries.
In \textsf{Baseline~2}, we employ the data structure in~\cite[Section 3]{DBLP:journals/corr/abs-1101-4068}, as it is less expensive to construct~\cite{DBLP:journals/corr/abs-1101-4068,DBLP:journals/information/KarrasTKK24} and thus leads to a faster baseline for \SM (\Cref{thm:b2}).
Greve et al.~\cite{DBLP:conf/icalp/GreveJLT10} showed a lower bound for \RM: any data
structure that uses $N \log^{\cO(1)}N$ space needs $\Omega(\log N / \log \log N)$ time to answer
an \RM query, and any data structure that supports \RM  
queries in $\cO(1)$ time needs $N^{1+\Omega(1)}$  space. The current best conditional lower bound~\cite{DBLP:journals/corr/abs-1101-4068} indicates that
answering $N$ \RM queries on an array of size $\cO(N)$ cannot be performed in 
$\cO(N^{\omega/2-\epsilon})$ time for any
$\epsilon > 0$, where $\omega < 2.3716$~\cite{DBLP:conf/soda/WilliamsXXZ24}  
is the square matrix multiplication exponent. 
Dynamic and multidimensional versions of \RM were studied in~\cite{DBLP:conf/icalp/SandlundX20,DBLP:conf/esa/El-Zein0MS18} and~\cite{DBLP:journals/tcs/DurocherEMT15}, respectively. \RM on paths of a tree (instead of subtrees as in \SM) was  studied in~\cite{DBLP:journals/algorithmica/DurocherSST16,DBLP:conf/esa/Gao022}. 

\paragraph{Document Retrieval.}~\SM can be used to construct a data structure for the \DRQ problem; see Section~\ref{sec:applications}. The \kDRQ problem has been studied in theoretical computer science; \cite{DBLP:conf/soda/NavarroN12,DBLP:journals/jacm/HonSTV14, DBLP:journals/siamcomp/NavarroN17,DBLP:conf/soda/0001N25} proposed linear-space data structures with optimal or near-optimal query time but no efficient algorithms to construct them. Thus, these works are of theoretical interest, and the data structure we propose is the first practical approach for \kDRQ.
\kDRQ differs from Top-$k$ document retrieval based on \emph{known patterns} (keywords), for which we refer to~\cite{DBLP:conf/sigir/KhattabHE20, DBLP:conf/sigir/GouMLSS25}.

\paragraph{Pattern Mining.}~The \CPM problem in   Section~\ref{sec:applications} falls into the area of pattern mining~\cite{DBLP:books/sp/fpm14/Aggarwal14}. Specifically, it is somewhat related to \emph{discriminative} (a.k.a \emph{emerging} or \emph{contrast sets}) pattern mining~\cite{DBLP:conf/kdd/DongL99,DBLP:journals/kais/DongL05,DBLP:journals/datamine/BayP01}. Given a collection of records (sequences~\cite{DBLP:journals/bib/LiuWGWH15,DBLP:journals/kais/MathonatNBK21,DBLP:conf/dasfaa/ChanKYT03}, transactions~\cite{DBLP:journals/kais/DongL05}, relational tuples~\cite{DBLP:journals/kais/DongL05}, or vectors~\cite{DBLP:journals/datamine/BayP01}), the latter  problem asks for mining all patterns (subsequences~\cite{DBLP:journals/bib/LiuWGWH15,DBLP:journals/kais/MathonatNBK21}, 
substrings~\cite{DBLP:conf/dasfaa/ChanKYT03}, itemsets~\cite{DBLP:journals/kais/DongL05}, sets of relational attribute values~\cite{DBLP:journals/kais/DongL05}, or sets of attribute/value pairs~\cite{DBLP:journals/datamine/BayP01}) which occur ``disproportionately'' in two or more parts of the collection that have different class labels. The disproportionality is measured based on difference in frequency~\cite{DBLP:journals/datamine/BayP01}, \emph{growth rate}~\cite{DBLP:conf/kdd/DongL99,DBLP:journals/kais/DongL05}, \emph{weighted relative accuracy}~\cite{DBLP:journals/kais/MathonatNBK21}, or other measures~\cite{DBLP:journals/bib/LiuWGWH15}. The \CPM problem is also relevant to \emph{fair} pattern mining~
\cite{DBLP:conf/sac/HajianMPDG14}. The latter problem  asks for mining itemsets in a transaction dataset that do not produce association rules which are unprotected according to legally-grounded fairness measures~\cite{DBLP:conf/sac/HajianMPDG14}. None of the aforementioned approaches can solve  \CPM.  

\paragraph{$\bm{q}$-grams.}~Sequence comparison by means of $q$-grams is ubiquitous in bioinformatics~\cite{DBLP:journals/tcs/Ukkonen92}. It offers a faster alternative to using more expensive string measures such as edit distance. For instance, BLAST~\cite{Altschul1990} and FASTA~\cite{Pearson1988}, two of the most widely-used tools for sequence-to-database search, are based on the notion of $q$-grams to report the best hits for a query.

\paragraph{Problems on Colored Trees.}~There are many other problems on leaf-colored~\cite{DBLP:conf/cpm/Hui92,steel,DBLP:conf/soda/Muthukrishnan02} and node-colored~\cite{pawel,shlomo,golovach} trees.

\section{Experimental Evaluation}\label{sec:experiments}

\paragraph{Data and Setup.}~We used $4$ benchmark datasets (see  Table~\ref{dataset2}): (1) \WebKB~\cite{craven1998learning}, which is a collection of webpages of computer science departments of various universities;
 (2) \Genes, which is a collection of DNA sequences between two markers flanking the human X chromosome centromere~\cite{sayers2023database}; (3) \News~\cite{zhou2015pattern}, which is a collection of documents from the Newsgroups dataset, and (4) \VIR~\cite{ncbi_genome_2025}, which is a collection of viral genomes. As the  baselines could not run on large  datasets, we applied them to samples of the first two datasets, constructed by selecting strings that have the same length uniformly at random; see Table~\ref{tab:sampledata}.

\begin{table}[!ht] 
\centering
\caption{Datasets characteristics}
\begin{tabular}{ccccccc}
\toprule
\textbf{Dataset}      & \textbf{Domain} & {\bf Alphabet} & {\bf No. of colors}  & 
\textbf{Mean string} & \textbf{Total length} & {\bf Nodes in} \\ 
~ & ~ & {\bf size} $\sigma$ & \textbf{$\Delta$} & {\bf length} & $N$ & $\mathcal{T}$ \\
\midrule
\WebKB~\cite{craven1998learning}  & Web    & 26 & 790,340         & 
32.98              & 26,068,175  & 38,491,476 \\ \midrule
\Genes~\cite{ncbi_handbook_chapter6} & Biology & 4 & 800,000      & 
1,250.01           & 1,000,007,888  & 2,030,921,400 \\ \midrule
\News~\cite{zhou2015pattern}         & News   & 26 & 4,781           & 
725.556            & 3,768,883   & 5,293,347 
\\ \midrule
\VIR~\cite{ncbi_genome_2025} & Biology & 4 & 143,588 & 29,085 & 4,176,208,246 & 7,313,326,212 
\\\bottomrule
\end{tabular}
\label{dataset2}
\vspace{+1mm}
\centering
\caption{Sample datasets characteristics}\label{tab:sampledata}
\begin{tabular}{ccccccc}
\toprule
\textbf{Dataset}      & \textbf{Domain} & {\bf Alphabet} & \textbf{No. of colors}  &
\textbf{String length} & \textbf{Total length} & {\bf Nodes in} \\ 
~ & ~ & \textbf{size $\sigma$} & $\Delta$ & ~ & $N$ & $\mathcal{T}$ \\\midrule
\WebKBsample~\cite{craven1998learning}       & Web    & 26 & 31,030          & 
100                & 3,103,000   & 4,876,306 \\ \midrule
\Genessample~\cite{ncbi_handbook_chapter6} & Biology   & 4 & 10,000          
& 200                & 2,000,000   & 3,753,134\\ \bottomrule
\end{tabular}
\label{dataset1}
\end{table}

Each of the used datasets is a collection of $\Delta$ strings $S_0, \ldots, S_{\Delta-1}$. Therefore, the tree $\mathcal{T}$ in the \SM problem is the suffix tree for $S_0\$_0\ldots S_{\Delta-1}\$_{\Delta-1}$ and the leaves with color $i\in [0,\Delta)$ correspond to the suffixes starting in $S_i\$_i$ (see Section~\ref{sec:applications}). We used string datasets because they are represented using large and complex trees that stress-test our algorithms (e.g.,\,note in Table~\ref{dataset2} that the suffix tree for \VIR has over $7.3$ billion nodes). On the other hand, phylogenetic trees or trees modeling file structures are generally much smaller.

Since there is no existing algorithm for \SM, we compared our \SCM algorithm (see Section~\ref{sec:linear}) to: (I) \textsf{Baseline~1} (\BAI) and \textsf{Baseline 2} (\BAII) (see Section~\ref{sec:baselines}); and (II) the  fastest $\cO(N\log \Delta)$-time baseline (\BAIII) (see Section~\ref{sec:linear}). 
Recall that all construction algorithms pre-compute and store all possible $N$ query answers. 
As a consequence, they have the same query time, and thus we do not evaluate this in our experiments.

We examined the impact of the two problem parameters, $N$ and $\Delta$, on runtime and space. We also showcase the benefit of our approach in \CPM, \CSQ, \kDRQ, and phylogenetic tree annotation. In these applications, we used four real datasets. 

Our experiments were conducted on a server equipped with an AMD EPYC 7702 64-Core Processor @ 2.00 GHz, 1 TB RAM, and Ubuntu 22.04.5 LTS. We implemented all algorithms in \textsf{C++}; see~\url{https://github.com/JialongZhou666/subtree-mode-mining}
for our code and datasets. 

\paragraph{Efficiency on Small Datasets.}~We present results showing that our \SCM algorithm substantially outperforms all three baselines in terms of runtime and memory consumption. 

\subparagraph{Impact of $N$.} Figs.~\ref{fig:webkb_small_time_n} to~\ref{fig:gene_small_time_n} show the runtime for varying $N$ and fixed $\Delta$. Our \SCM algorithm was 
\emph{faster than the fastest baseline, \BAIII, by at least one order of magnitude
and $24$ times on average}. \BAIII in turn was faster than both \BAI and \BAII by at least two orders of magnitude on average. All algorithms scaled linearly with $N$, in line with  their time complexities, except 
\BAII whose time complexity is $\cO(N\sqrt{N})$.
In fact, since the space complexity of \BAII is also
$\cO(N\sqrt{N})$, it 
did not terminate for strings with more than 1 million letters; it needed more than the 1TB of memory that was available. Figs.~\ref{fig:webkb_small_memory_n} to~\ref{fig:gene_small_memory_n} show the  peak memory consumption for the experiments of Figs.~\ref{fig:webkb_small_time_n} to~\ref{fig:gene_small_time_n}. \SCM needed 
\emph{on average $26\%$ and up to $58\%$ less memory compared to the best baseline}, \BAIII, as storing counts for the single-color trees needs less memory than merging trees. \BAIII in turn needed two orders of magnitude less memory on average compared to both \BAI and \BAII, as its space complexity is $\cO(N)$,
while that of \BAI and \BAII is $\cO(N\Delta)$ and $\cO(N\sqrt{N})$, respectively.

\begin{figure}[htbp]
\hspace{-2mm}
\begin{subfigure}[b]{0.31\columnwidth}
    \includegraphics[width=1.06\textwidth]{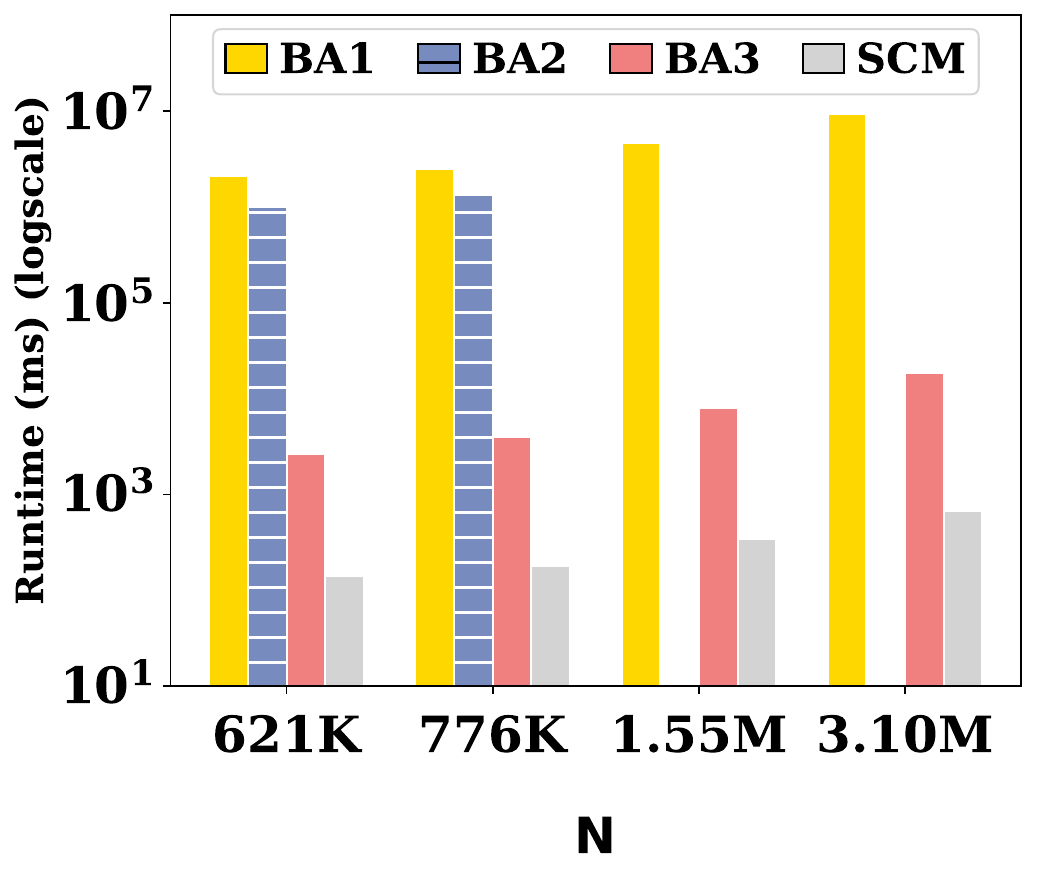}
    \caption{\WebKBsample}
    \label{fig:webkb_small_time_n}
\end{subfigure}\hspace{+1mm}
\begin{subfigure}[b]{0.31\columnwidth}
    \includegraphics[width=1.06\textwidth]{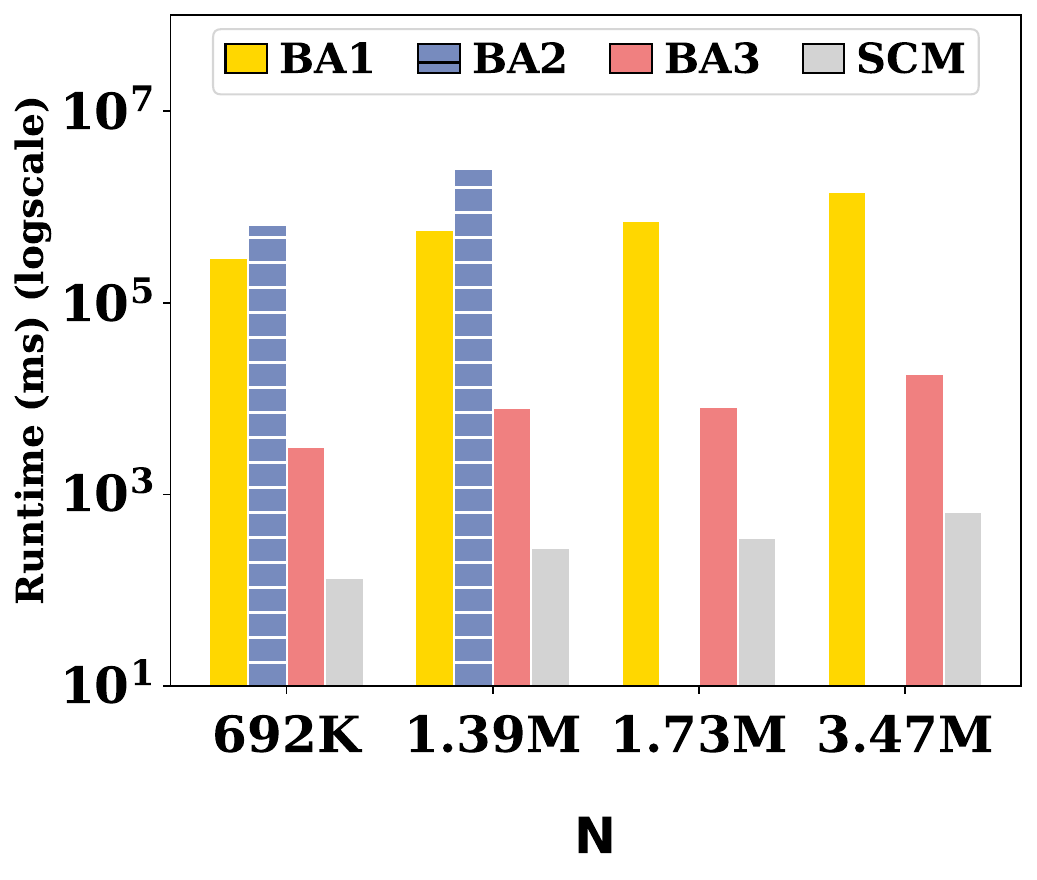}
    \caption{\News}
    \label{fig:news_time_n}
\end{subfigure}\hspace{+1mm}
\begin{subfigure}[b]{0.31\columnwidth}
    \includegraphics[width=1.06\textwidth]{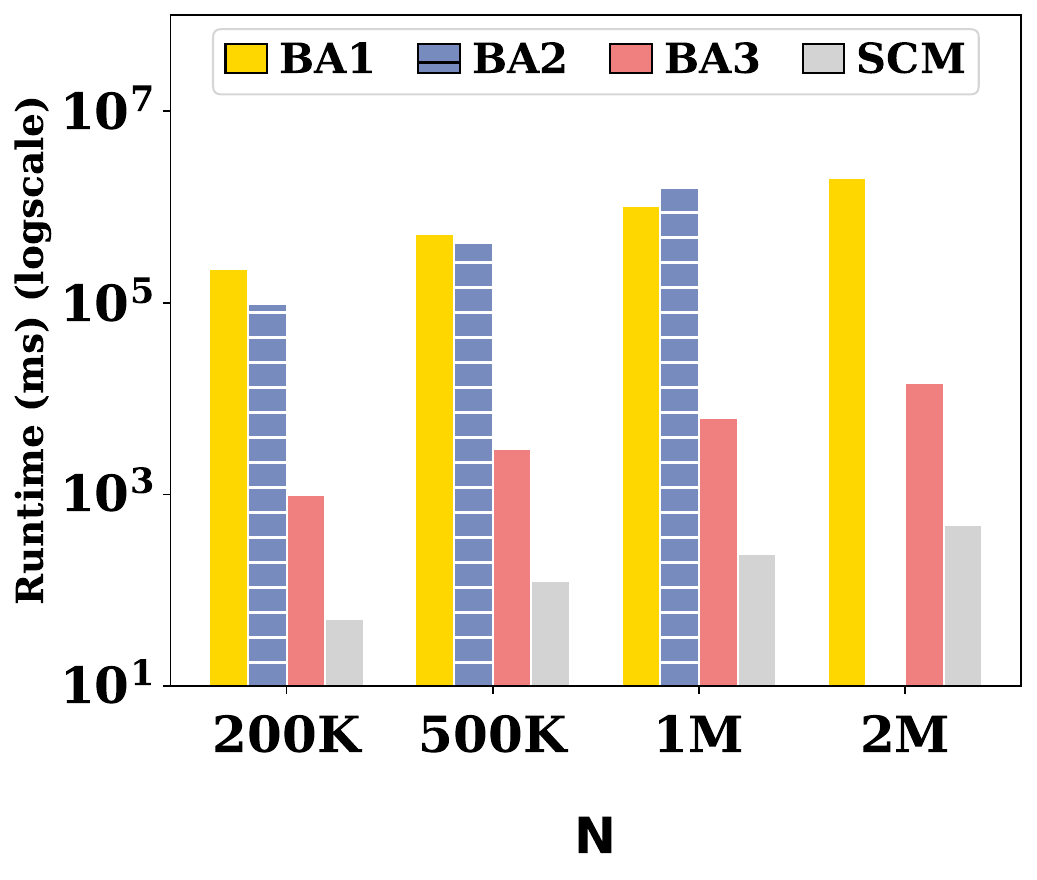}
    \caption{\Genessample}
    \label{fig:gene_small_time_n}
\end{subfigure}\\
\begin{subfigure}[b]{0.31\columnwidth}
    \includegraphics[width=1.06\textwidth]{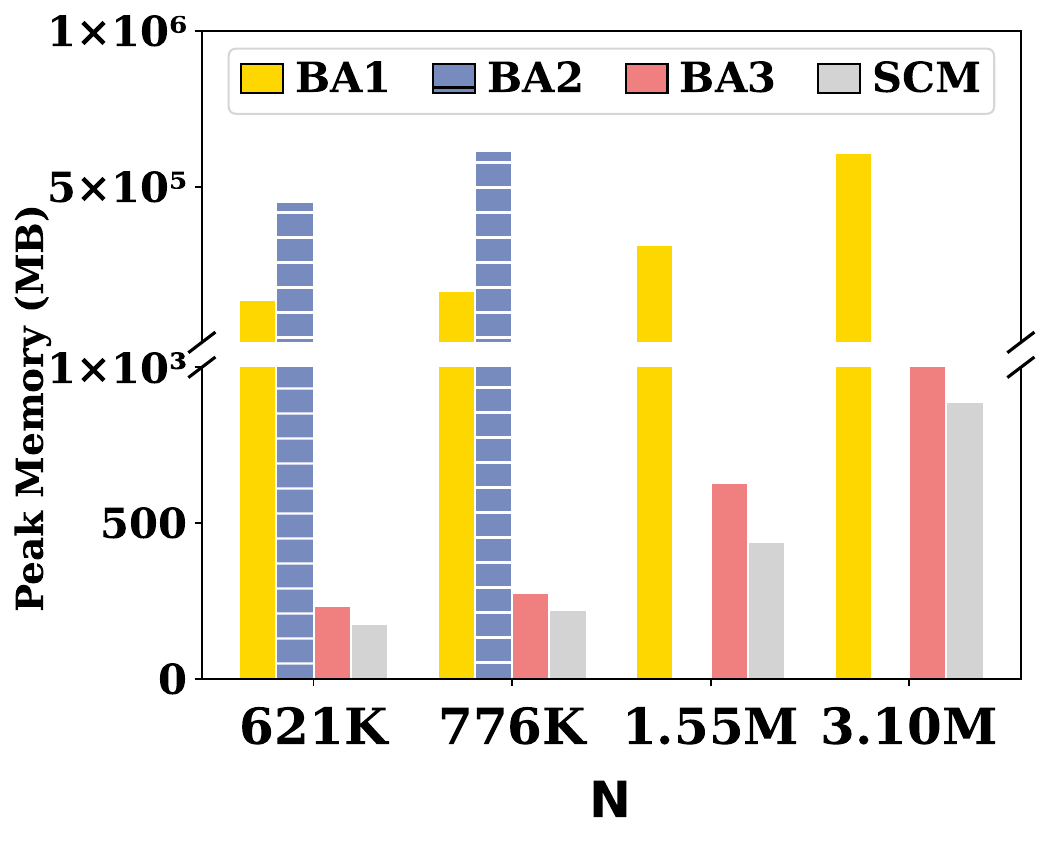}
    \caption{\WebKBsample}
    \label{fig:webkb_small_memory_n}
\end{subfigure}\hspace{+1mm}
\begin{subfigure}[b]{0.31\columnwidth}
    \includegraphics[width=1.06\textwidth]{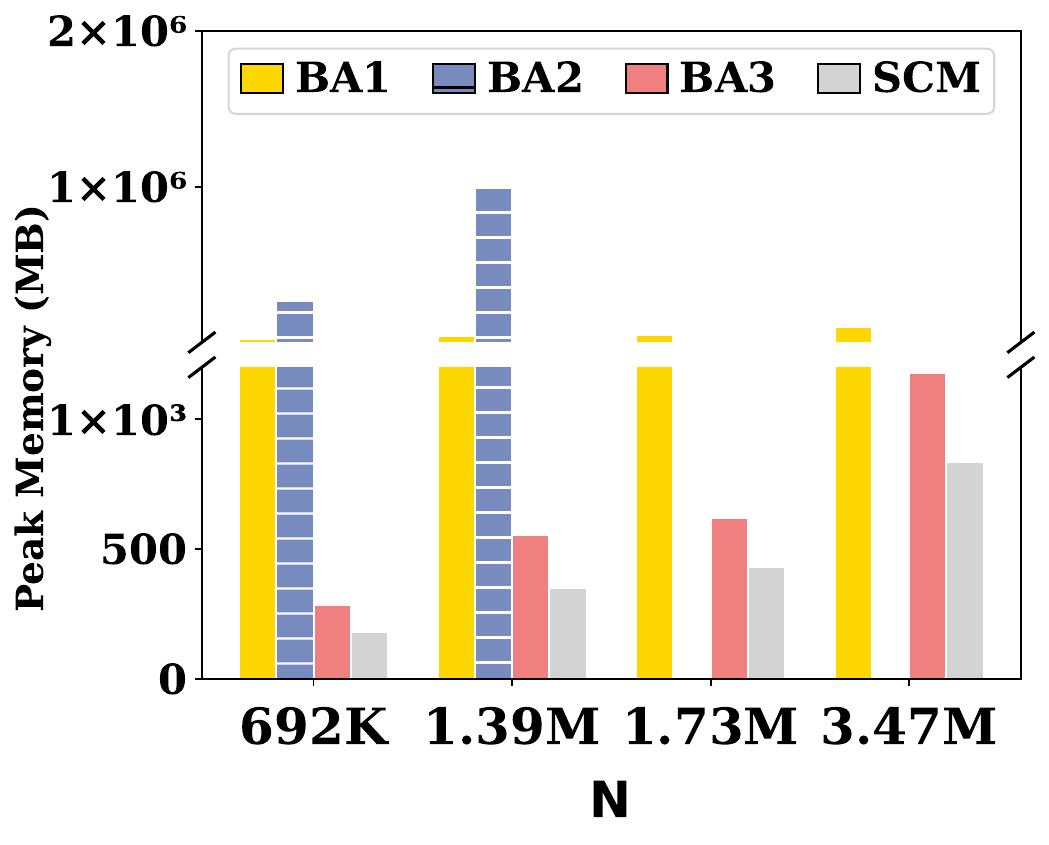}
    \caption{\News}
    \label{fig:news_memory_n}
\end{subfigure}\hspace{+1mm}
\begin{subfigure}[b]{0.31\columnwidth}
    \includegraphics[width=1.06\textwidth]{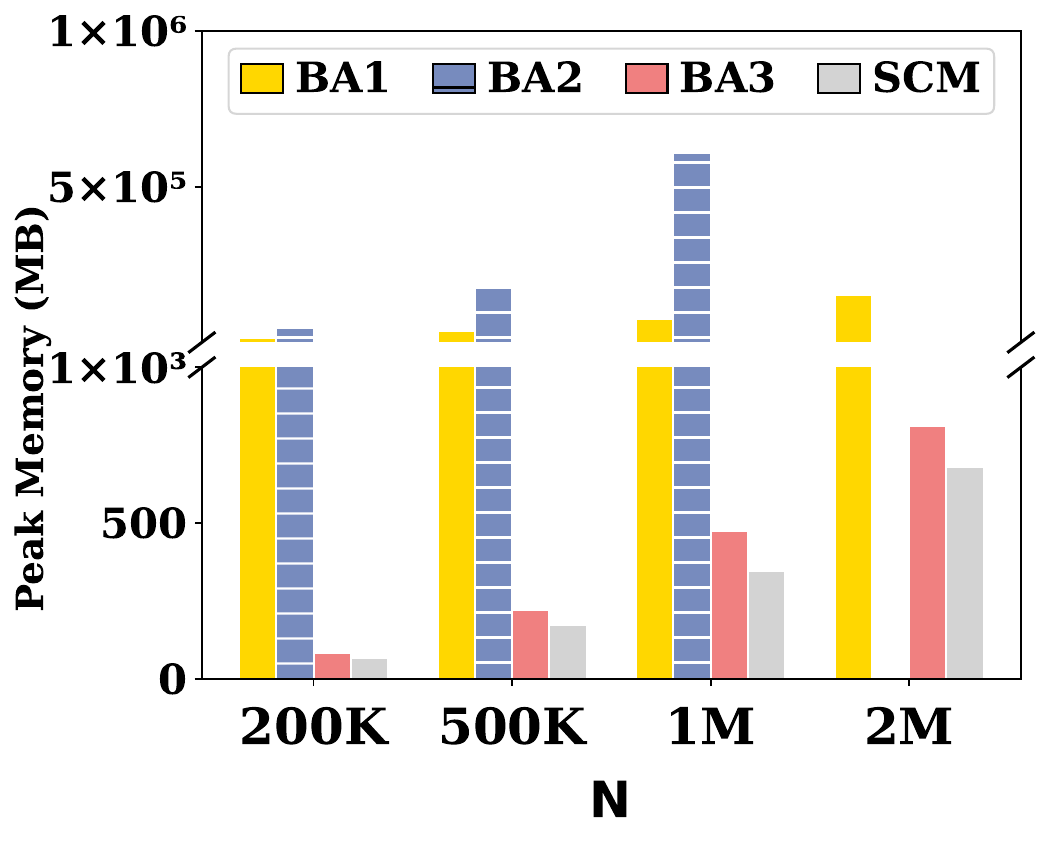}
    \caption{\Genessample}
    \label{fig:gene_small_memory_n}
\end{subfigure}

\caption{(a-c) Runtime and (d-f) memory vs. $N$. Missing bars indicate that a method needed more than 1TB of memory.}
\label{fig:runtime-mem-small-N}
\end{figure}

\subparagraph{Impact of $\Delta$.} Figs.~\ref{fig:webkb_small_time_delta} to~\ref{fig:gene_small_time_delta} show the runtime for varying~$\Delta$ and fixed $N$; we fixed $N=300,000$ for the \WebKBsample dataset, $N=162,000$ for \News, and $N=200,000$ for \Genessample by taking the prefix of length $\lfloor N/\Delta \rfloor$ of each string. 
Our \SCM algorithm was 
\emph{faster than the fastest baseline, \BAIII, by at least an order of magnitude and $21$ times on average} and, as expected by its time complexity, its runtime was not affected by $\Delta$. \BAIII in turn was faster than both \BAI and \BAII by two orders of magnitude on average, which is in line with the time complexities of these algorithms.  Figs.~\ref{fig:webkb_small_memory_delta} to~\ref{fig:gene_small_memory_delta} show the peak memory consumption for the experiments of Figs.~\ref{fig:webkb_small_time_delta} to~\ref{fig:gene_small_time_delta}. \SCM  needed \emph{at least $20\%$ and up to $50\%$ less memory compared to the best baseline}, \BAIII, for the same reason as in the experiments of Figs.~\ref{fig:webkb_small_memory_n} to~\ref{fig:gene_small_memory_n}. \BAIII in turn was more space-efficient than both \BAI and \BAII by more than two orders of magnitude on average, which is in line with the space complexities of these algorithms. Also, \BAII uses more memory as $\Delta$ increases, as the string index it uses gets  larger.

\begin{figure}[htbp]
\begin{subfigure}[b]{0.32\columnwidth}
    \includegraphics[width=1.05\textwidth]{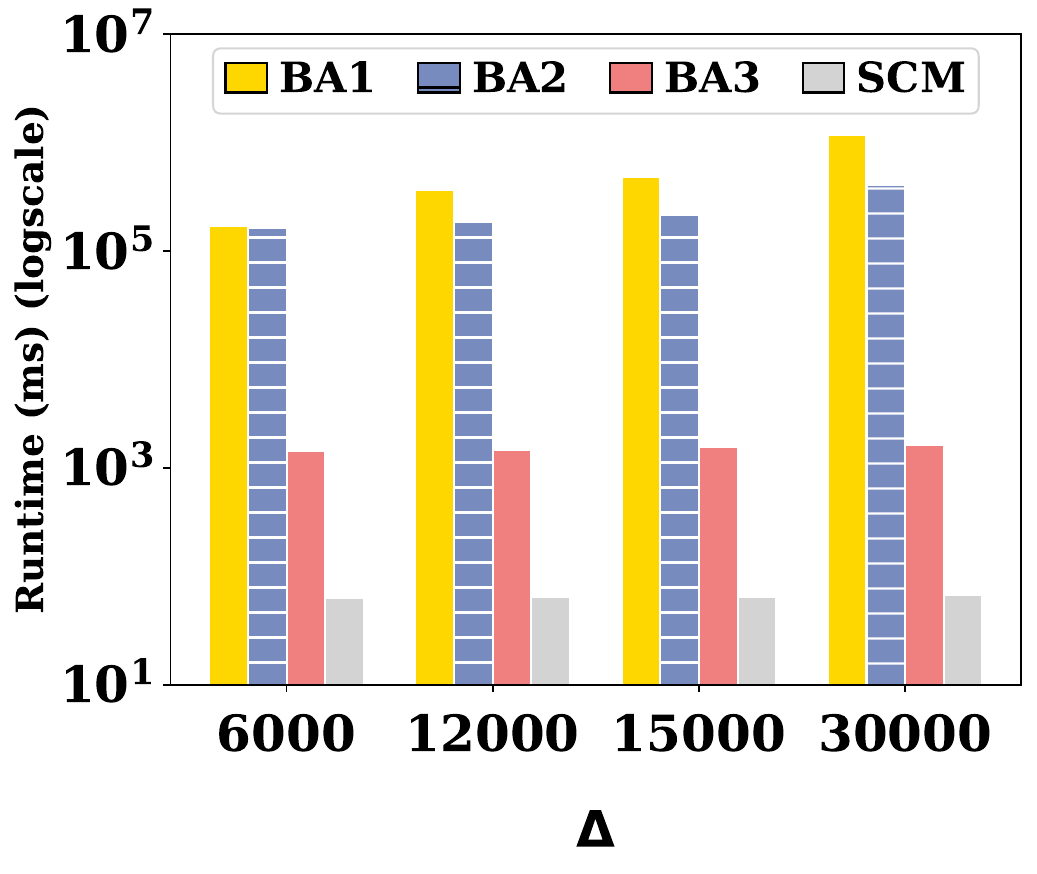}
    \caption{\WebKBsample}
    \label{fig:webkb_small_time_delta}
\end{subfigure}
\begin{subfigure}[b]{0.32\columnwidth}
    \includegraphics[width=1.05\textwidth]{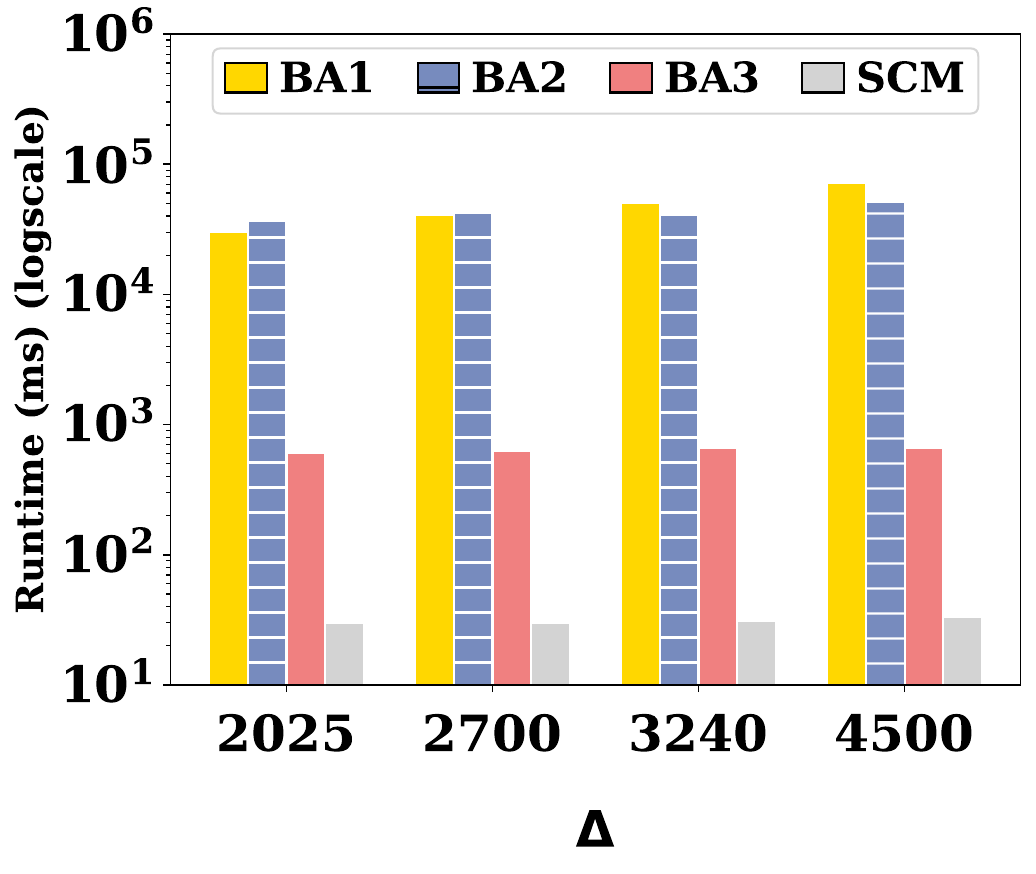}
    \caption{\News}
    \label{fig:news_time_delta}
\end{subfigure}
\begin{subfigure}[b]{0.32\columnwidth}
    \includegraphics[width=1.05\textwidth]{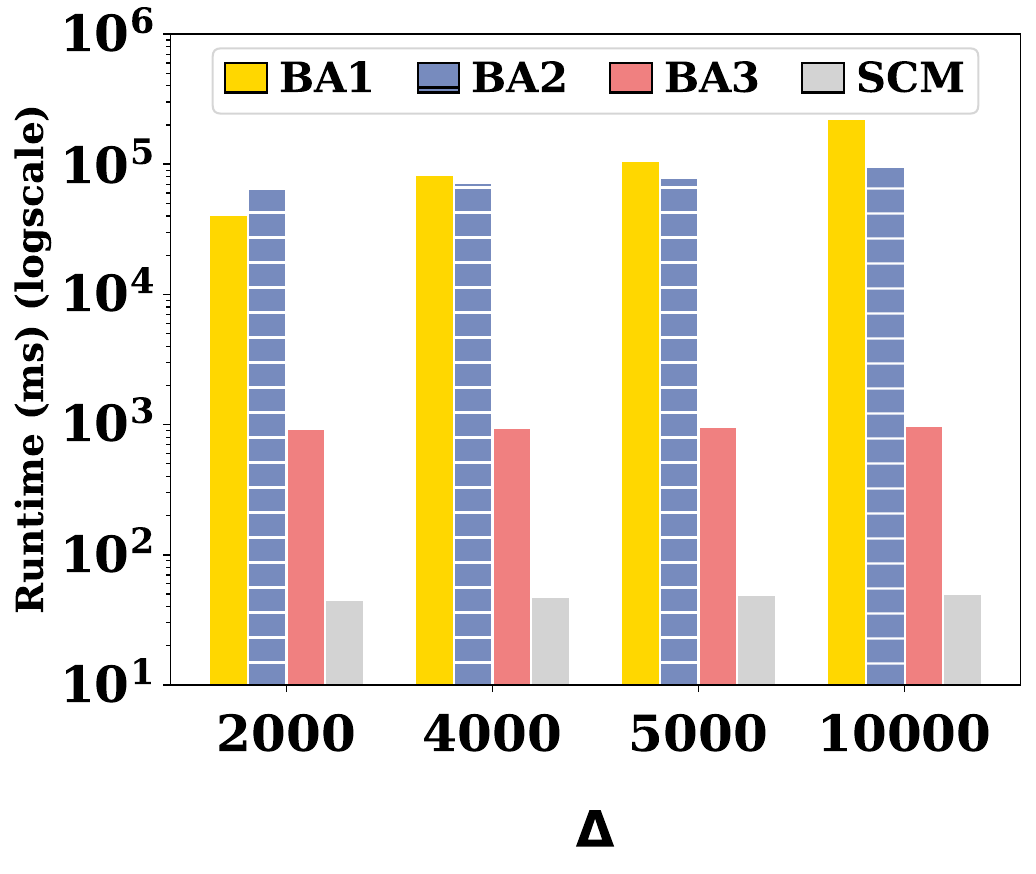}
    \caption{\Genessample}
    \label{fig:gene_small_time_delta}
\end{subfigure}

\begin{subfigure}[b]{0.32\columnwidth}
    \centering
    \includegraphics[width=1.05\textwidth]{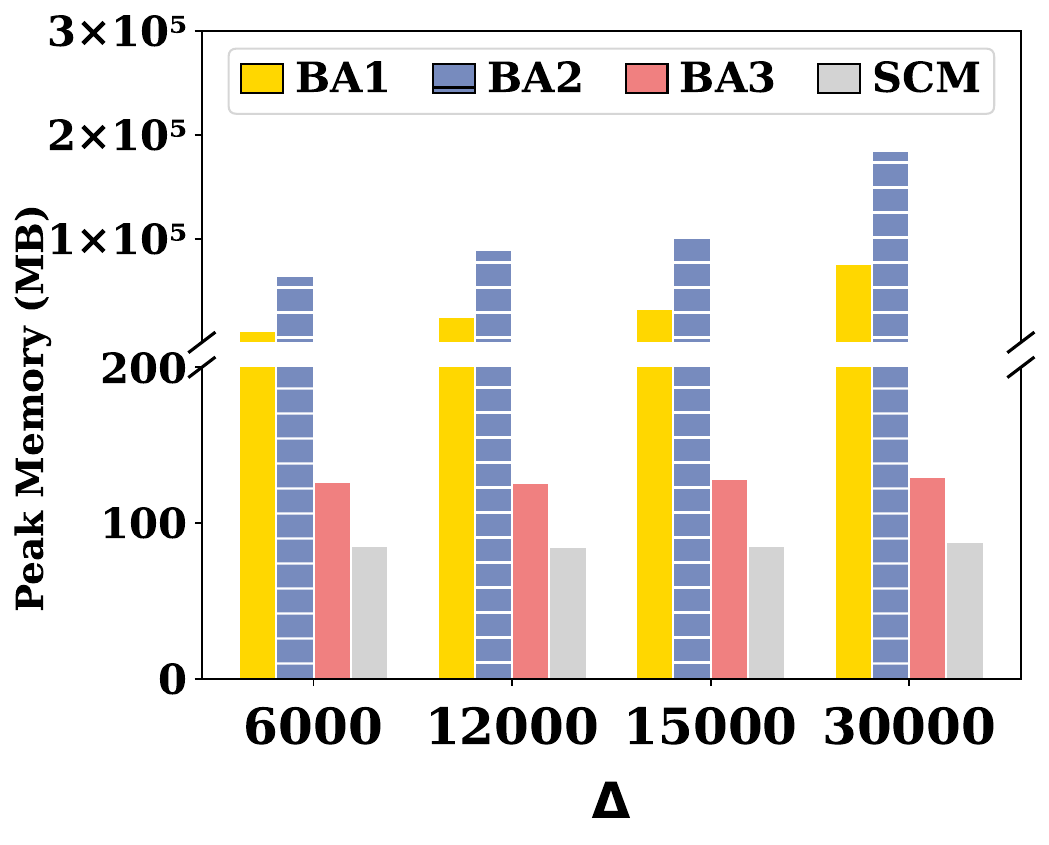}
    \caption{\WebKBsample}
    \label{fig:webkb_small_memory_delta}
\end{subfigure}
\begin{subfigure}[b]{0.32\columnwidth}
    \centering
    \includegraphics[width=1.05\textwidth]{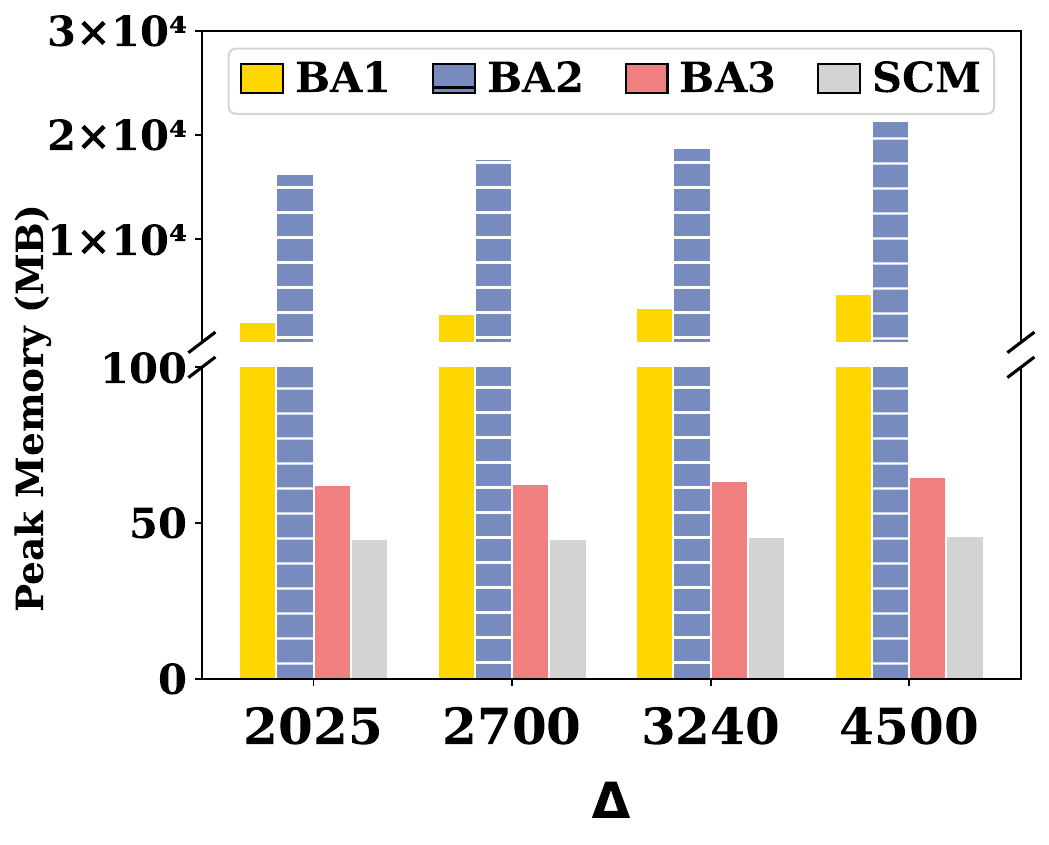}
    \caption{\News}
    \label{fig:news_memory_delta}
\end{subfigure}
\begin{subfigure}[b]{0.32\columnwidth}
    \centering
    \includegraphics[width=1.05\textwidth]{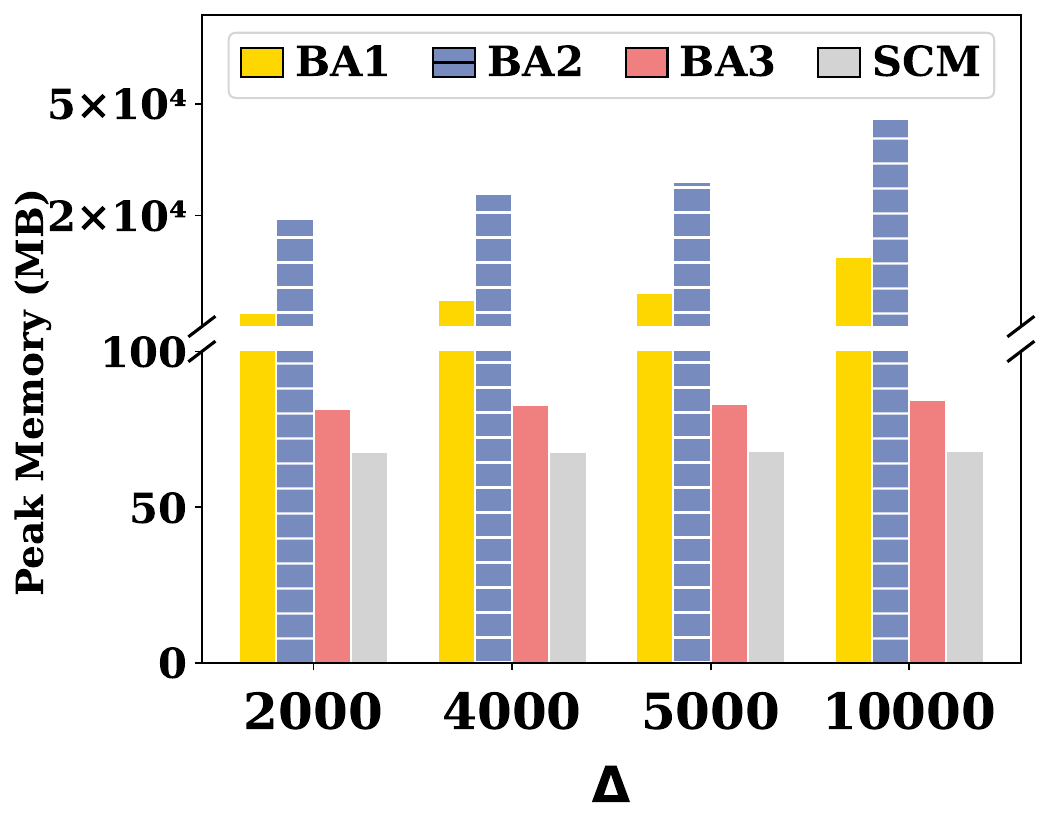}
    \caption{\Genessample}
    \label{fig:gene_small_memory_delta}
\end{subfigure}

\caption{(a-c) Runtime and (d-f) memory vs. $\Delta$.}
\label{fig:runtime-mem-small-Delta}
\end{figure}
\begin{figure*}[htbp]
\begin{subfigure}[b]{0.3\textwidth}
    \includegraphics[width=1.05\textwidth]{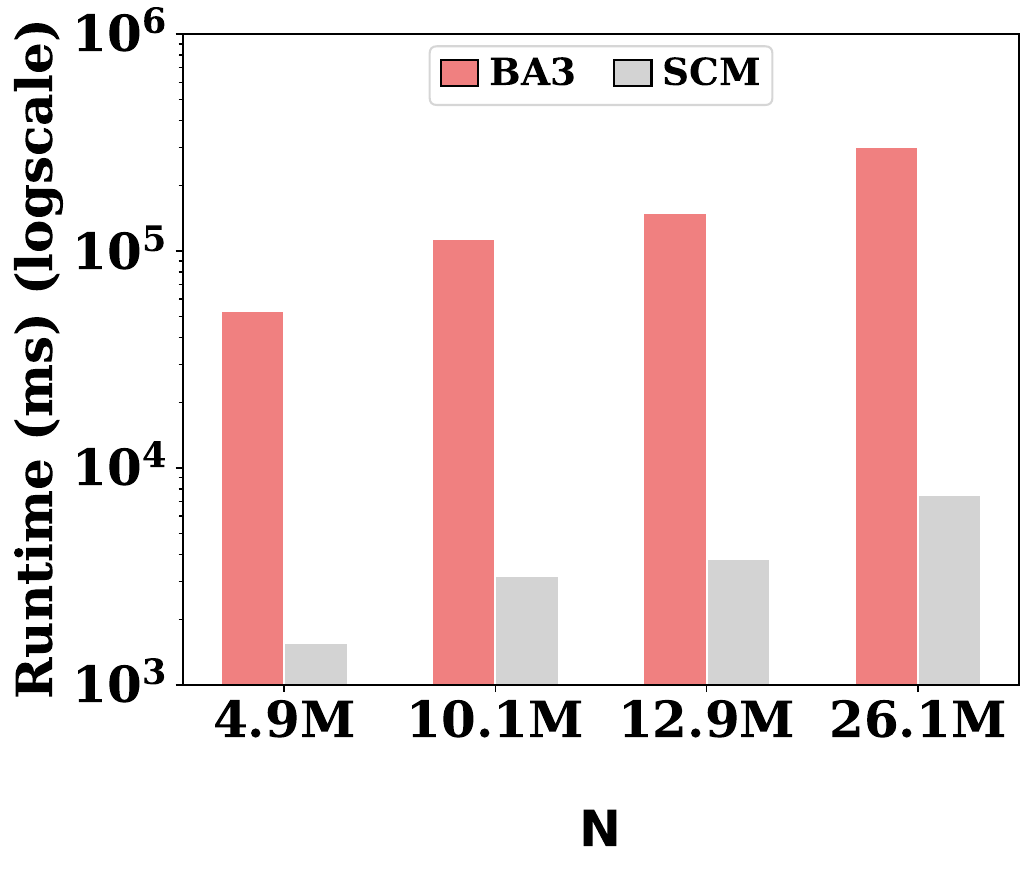}
    \caption{\WebKB}
    \label{fig:webkb_time_n}
\end{subfigure}\hspace{+1mm}
\begin{subfigure}[b]{0.3\textwidth}
    \includegraphics[width=1.05\textwidth]{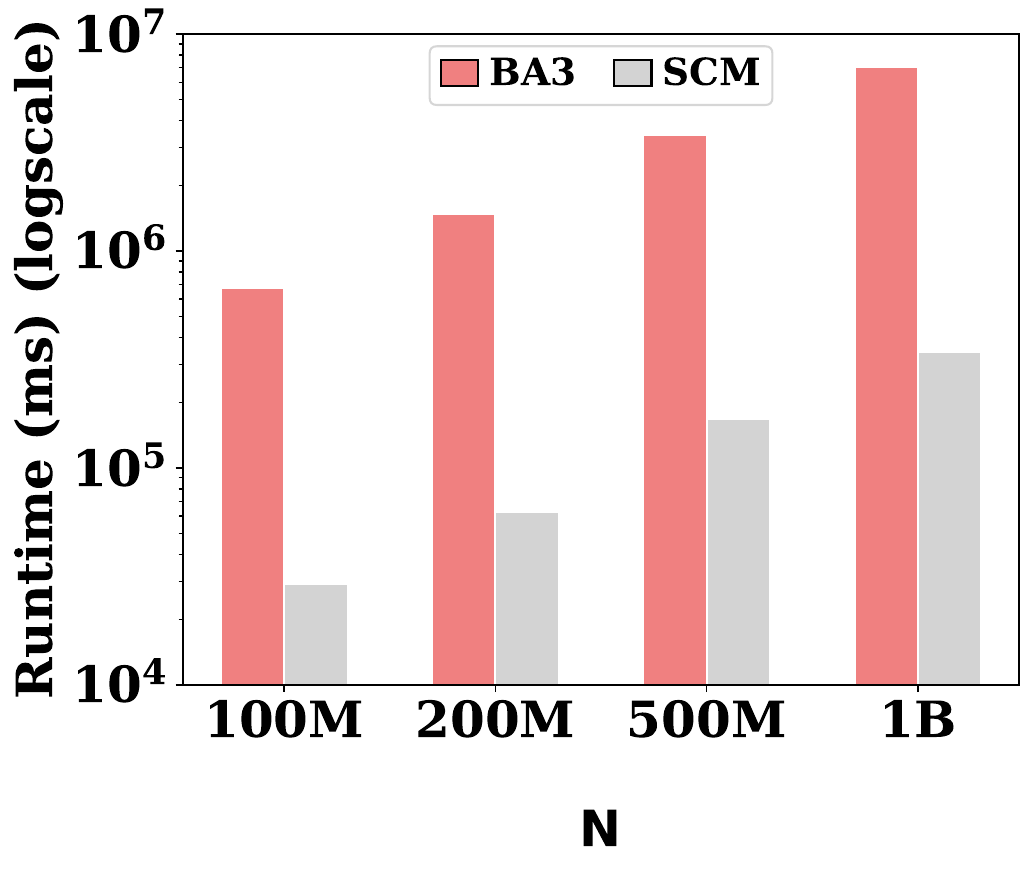}
    \caption{\Genes}
    \label{fig:gene_time_n}
\end{subfigure}\hspace{+1mm}
\begin{subfigure}[b]{0.3\textwidth}
    \includegraphics[width=1.05\textwidth]{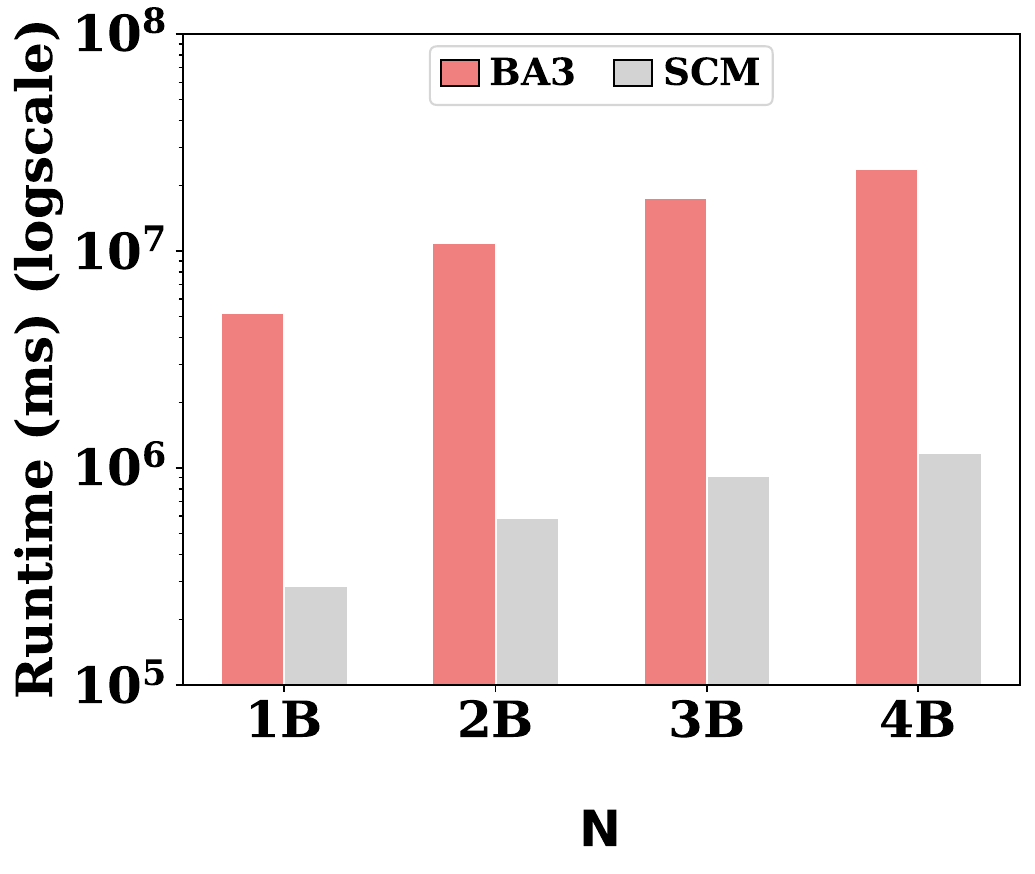}
    \caption{\VIR}
    \label{fig:vir_time_n}
\end{subfigure}\hspace{+1mm}\\
\begin{subfigure}[b]{0.3\textwidth}
    \centering
    \includegraphics[width=1.05\textwidth]{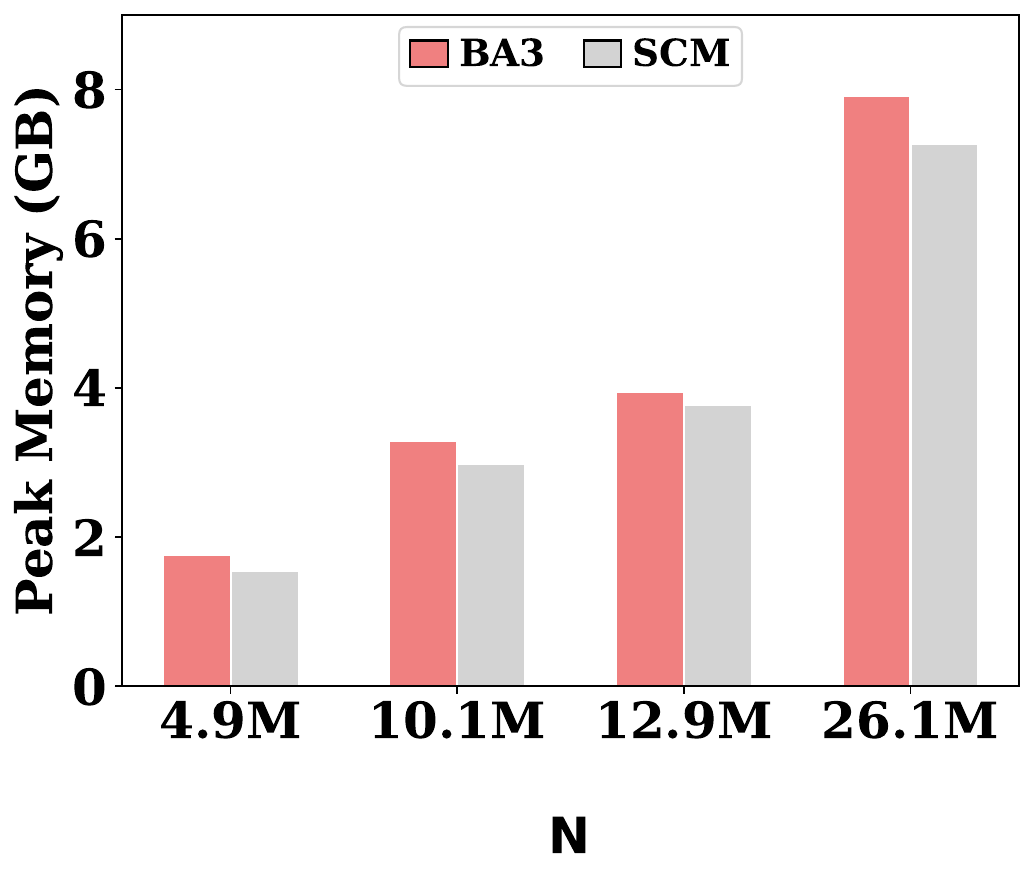}
    \caption{\WebKB}
    \label{fig:webkb_mem_n}
\end{subfigure}\hspace{+1mm}
\begin{subfigure}[b]{0.3\textwidth}
    \centering
    \includegraphics[width=1.05\textwidth]{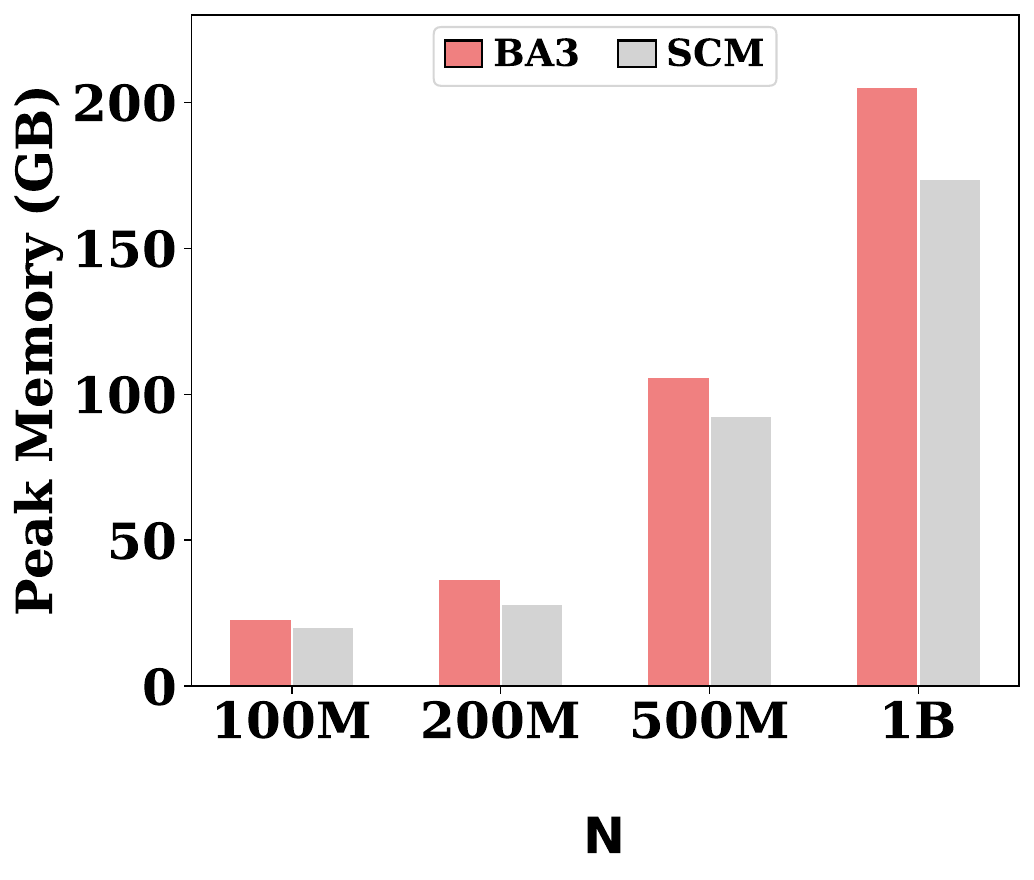}
    \caption{\Genes}
    \label{fig:gene_mem_n}
  \end{subfigure}\hspace{+1mm}
\begin{subfigure}[b]{0.3\textwidth}
    \centering
    \includegraphics[width=1.05\textwidth]{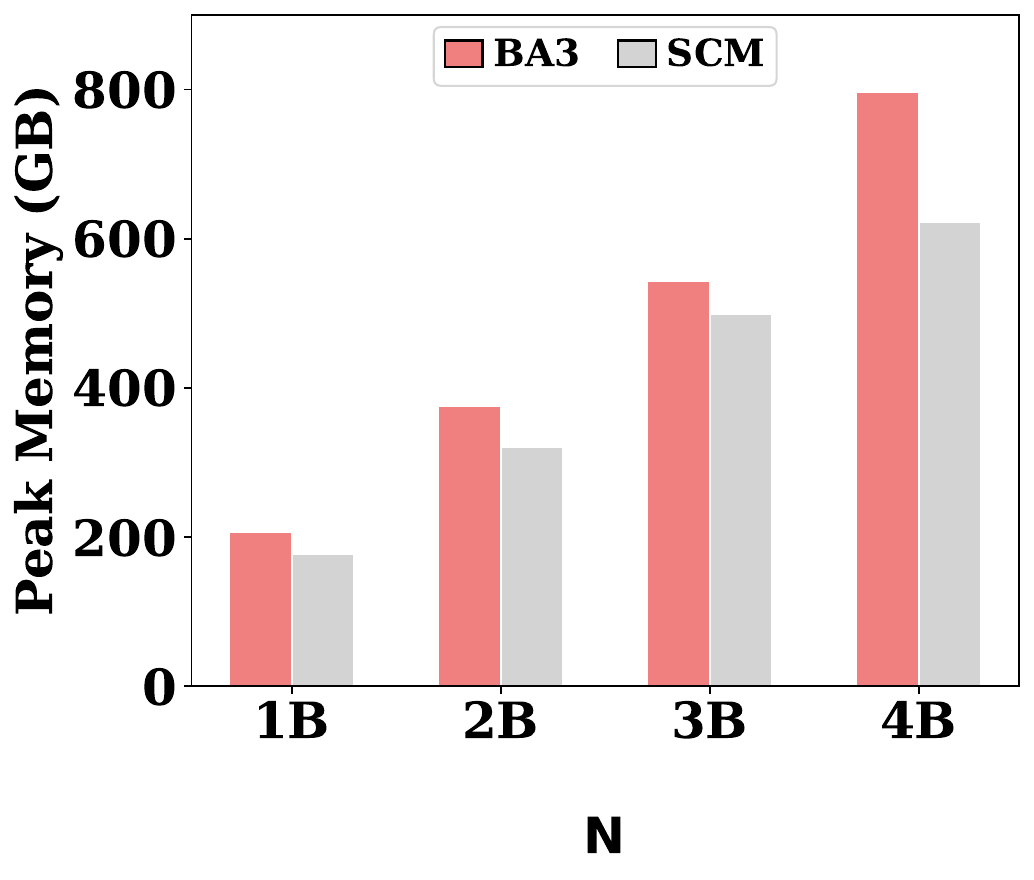}
    \caption{\VIR}
    \label{fig:vir_mem_n}
\end{subfigure}%
\caption{
(a-c) Runtime and (d-f) memory vs. $N$.}
\label{fig:runtime-largeN}
\begin{subfigure}[b]{0.3\textwidth}
    \includegraphics[width=1.05\textwidth]{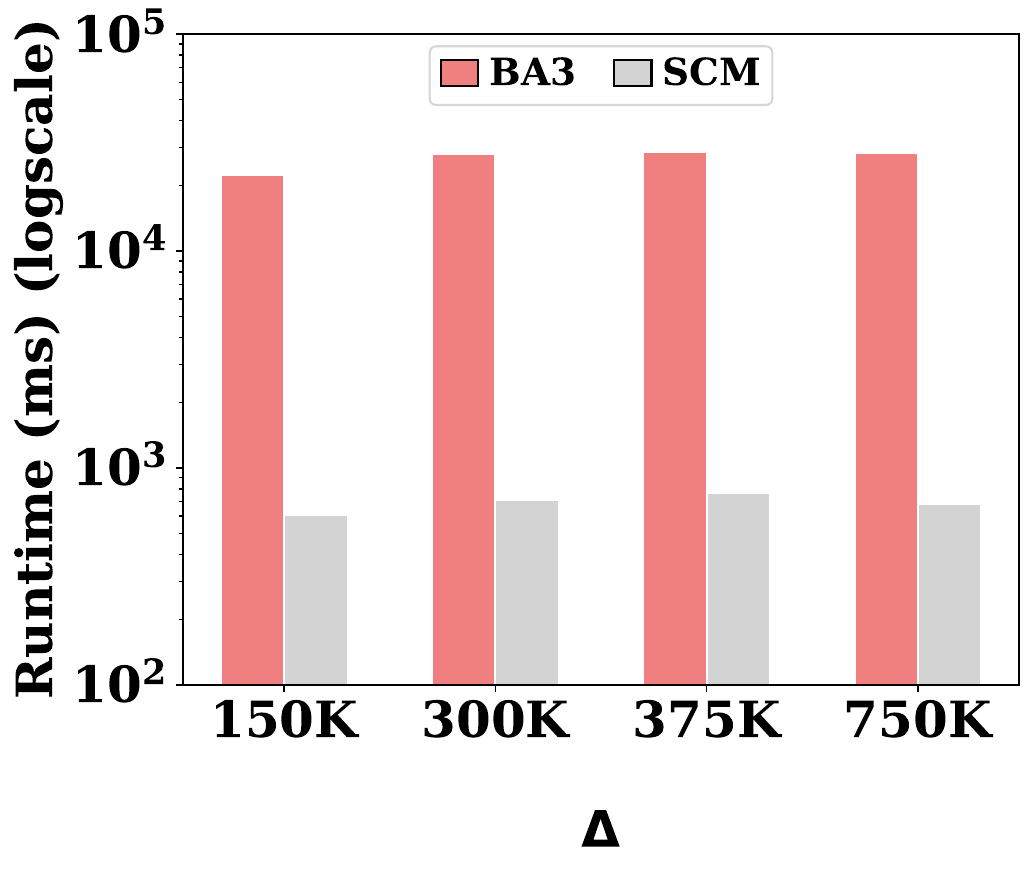}
    \caption{\WebKB}
    \label{fig:webkb_time_delta}
\end{subfigure}\hspace{+1mm}
\begin{subfigure}[b]{0.3\textwidth}
    \includegraphics[width=1.05\textwidth]{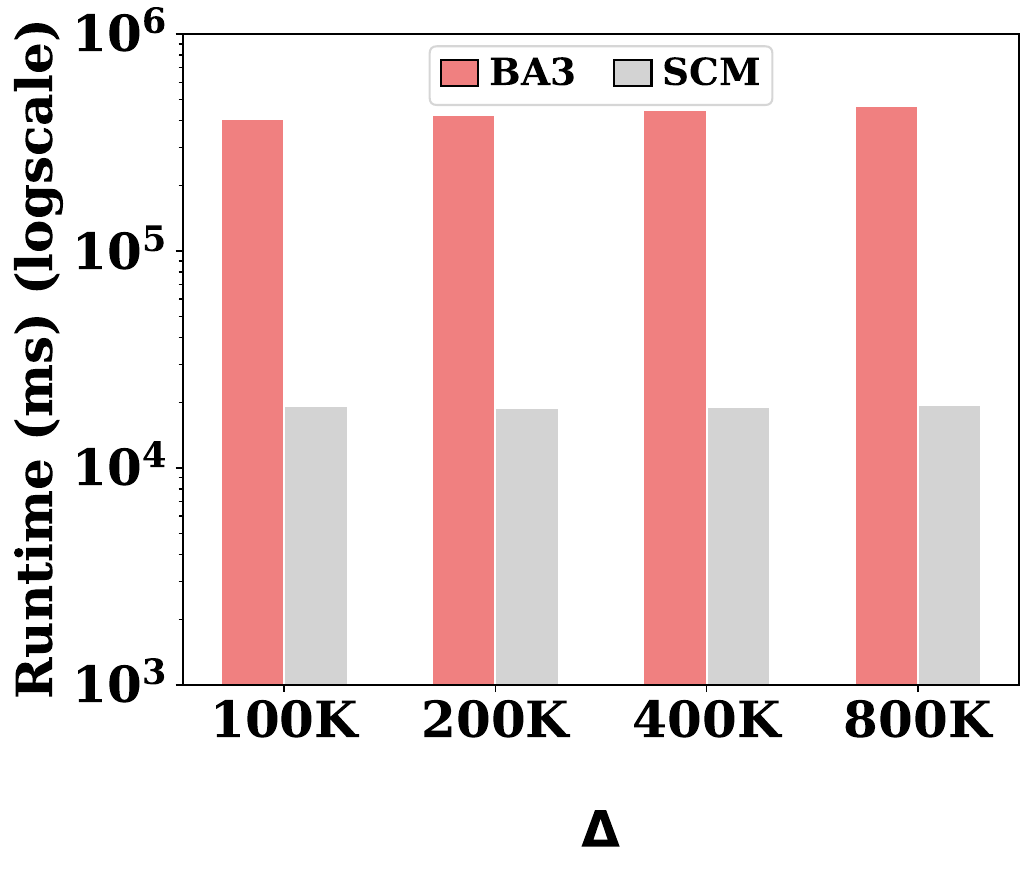}
    \caption{\Genes}
    \label{fig:gene_time_delta}
\end{subfigure}\hspace{+1mm}
\begin{subfigure}[b]{0.3\textwidth}
    \includegraphics[width=1.05\textwidth]{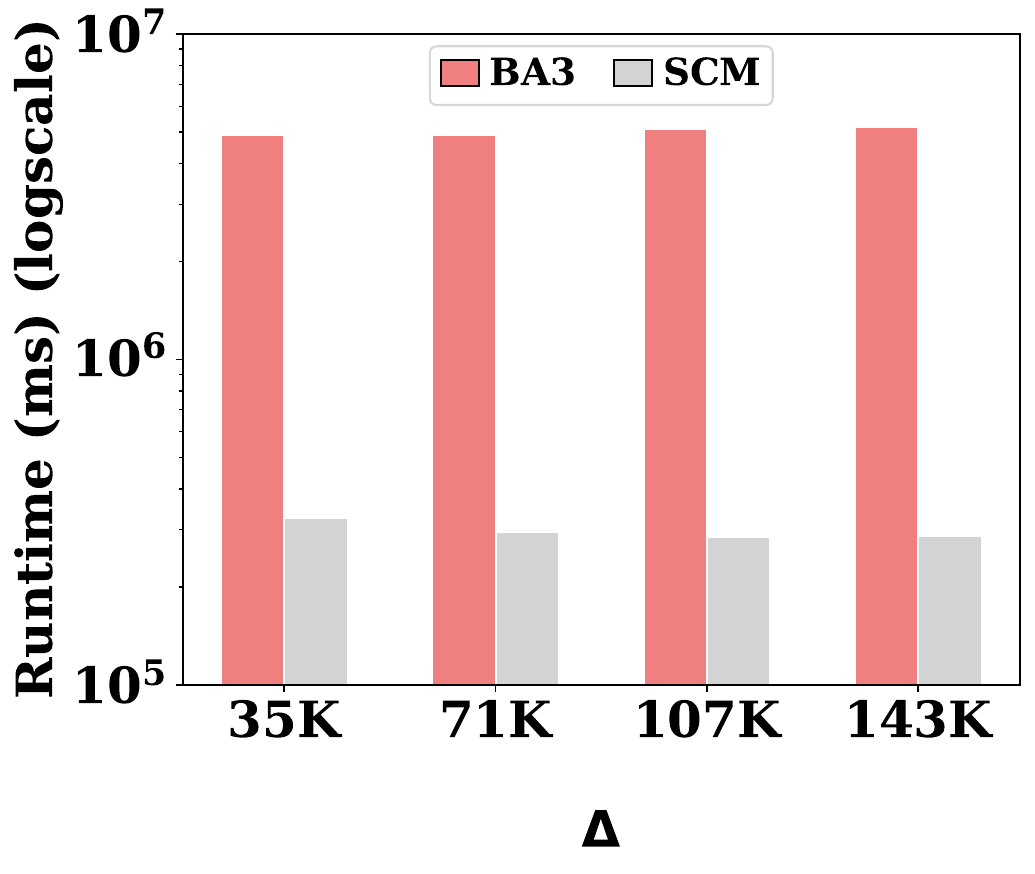}
    \caption{\VIR}
    \label{fig:vir_time_delta}
\end{subfigure}\\
\begin{subfigure}[b]{0.3\textwidth}
    \includegraphics[width=1.05\textwidth]{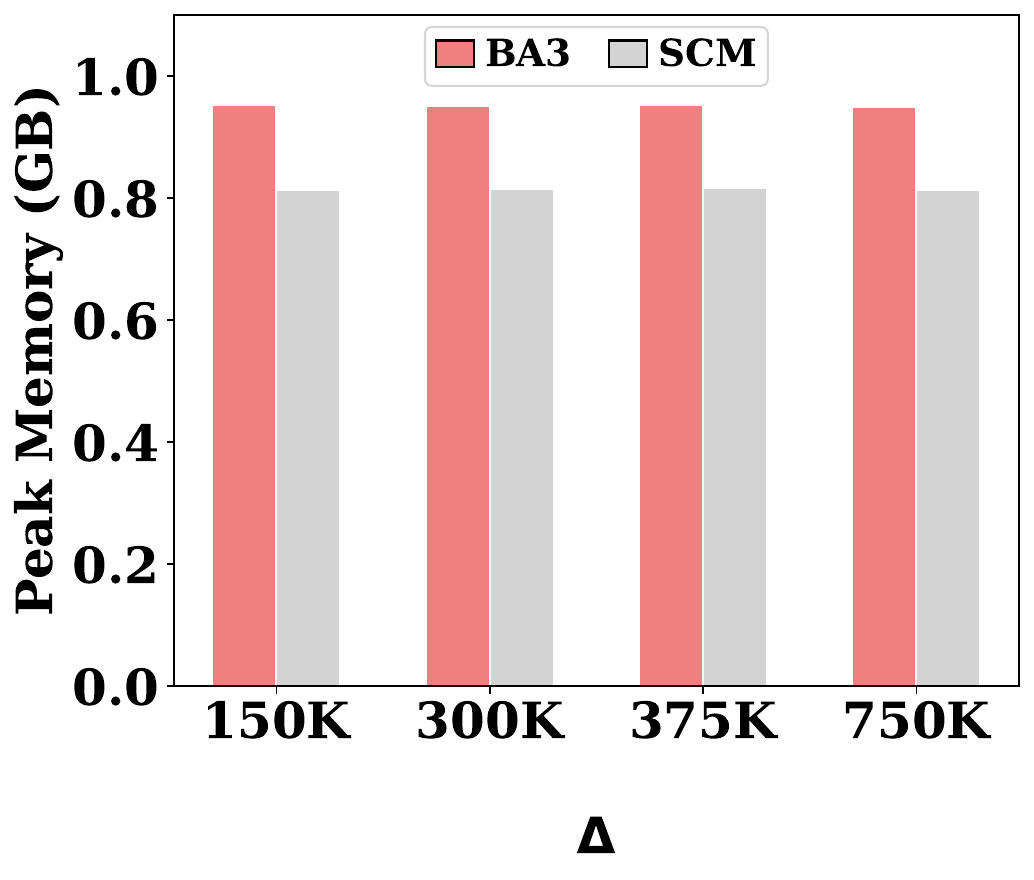}
    \caption{\WebKB}
    \label{fig:webkb_mem_delta}
\end{subfigure}\hspace{+1mm}
\begin{subfigure}[b]{0.3\textwidth}
    \includegraphics[width=1.05\textwidth]{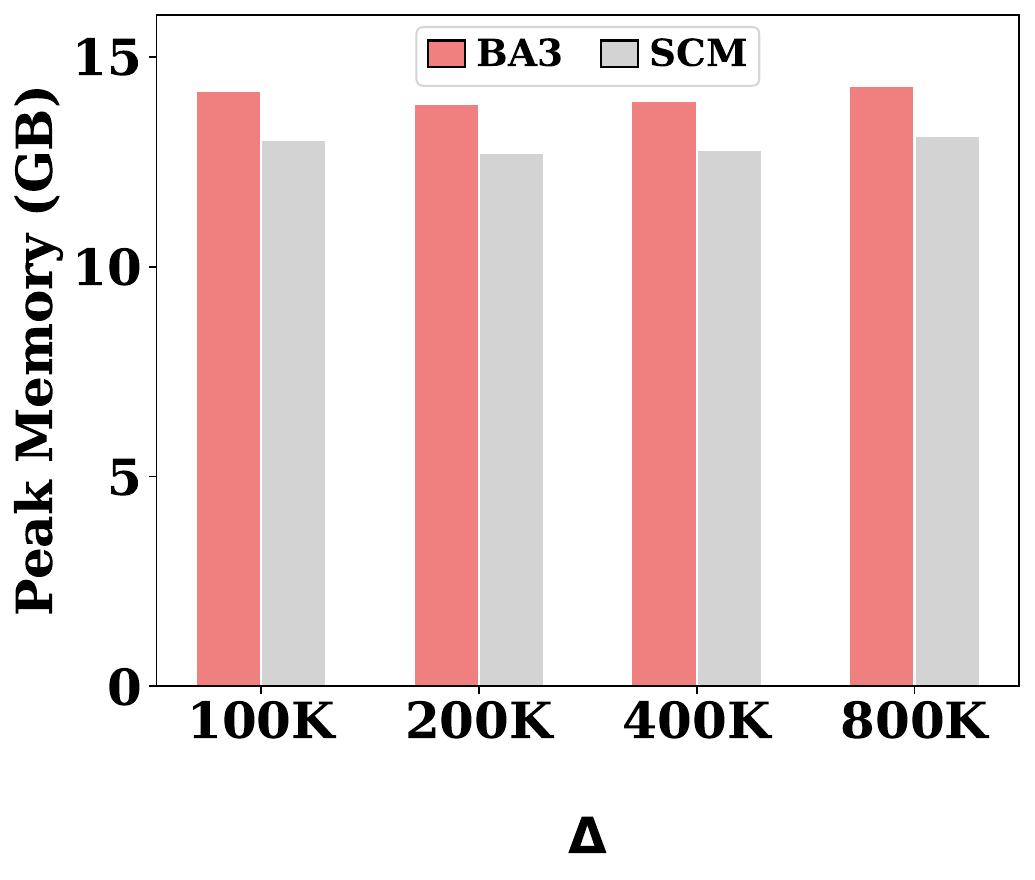}
    \caption{\Genes}
    \label{fig:gene_mem_delta}
\end{subfigure}
\begin{subfigure}[b]{0.3\textwidth}
    \centering
    \includegraphics[width=1.05\textwidth]{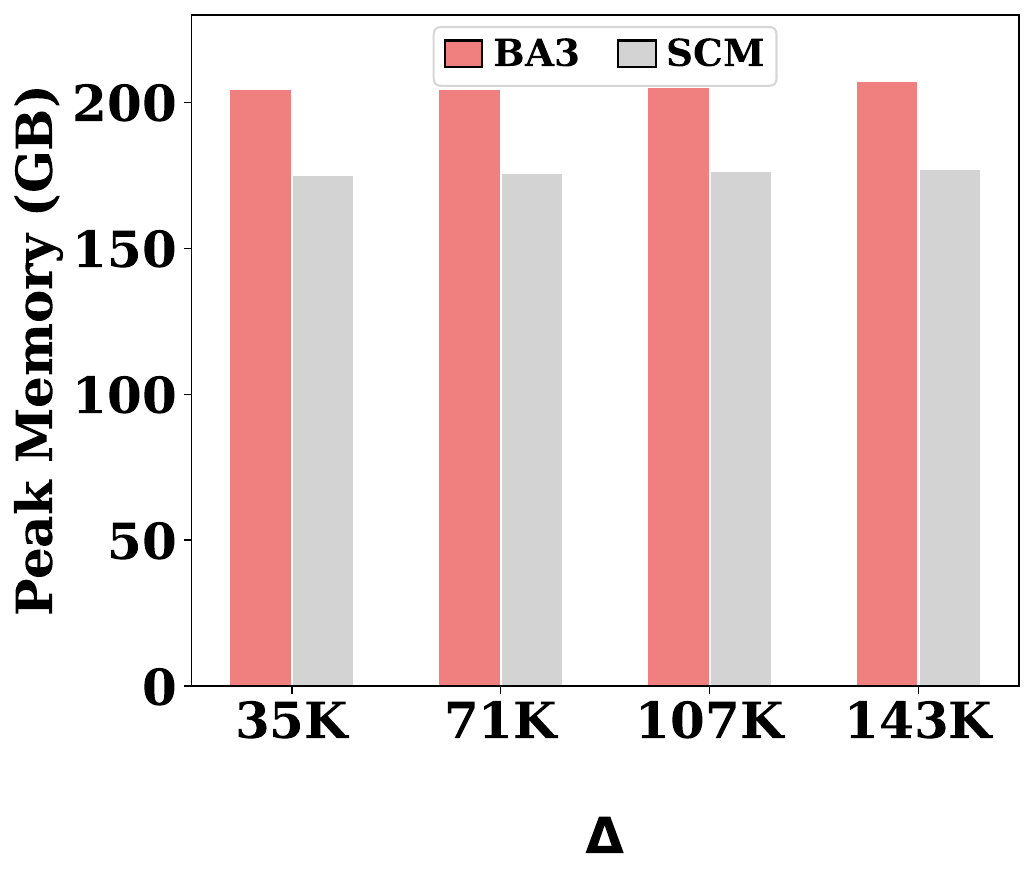}
    \caption{\VIR}
    \label{fig:vir_mem_delta}
\end{subfigure}
\caption{(a-c) Runtime and (d-f) memory vs. $\Delta$.}\label{fig:mem-large}
\end{figure*}

\paragraph{Efficiency on Large Datasets.}~We present results for \SCM and \BAIII, as \BAI and \BAII did not terminate within $48$  hours. The results below are analogous to those for the small datasets. 

\paragraph{Impact of $N$} Figs.~\ref{fig:webkb_time_n} to~\ref{fig:vir_time_n} show the runtime for varying~$N$ and fixed $\Delta$. Both \SCM  and \BAIII scaled linearly with $N$, in line with their time complexities,
but \SCM was  \emph{at least $18$ and up to $40$ times} faster, as \BAIII has an extra $\log \Delta$ term in its time complexity which is $17$ to $20$ depending on the dataset. \SCM is practical; it took less than $20$ minutes when applied to the entire \VIR dataset whose total length is $4.2$ billion letters (and suffix tree has over $7.3$ billion nodes). Figs.~\ref{fig:webkb_mem_n} to~\ref{fig:vir_mem_n} show the peak memory consumption for the experiments of Figs.~\ref{fig:webkb_time_n} to~\ref{fig:vir_time_n}
. \SCM needed \emph{$15\%$ less memory compared to} \BAIII on average, which is in line with the space complexities of the algorithms. This shows again the benefit of the count merging in \SCM compared to the tree merging in \BAIII.

\subparagraph{Impact of $\Delta$.} Figs.~\ref{fig:webkb_time_delta} to~\ref{fig:vir_time_delta} show the runtime for varying~$\Delta$ and fixed $N$; we fixed $N$ by removing letters from each string evenly, so that the remaining strings have total length $N=2,000,000$ for \WebKB,  $N=62,500,000$ for \Genes, and $N=10^9$ for \VIR. Again, \SCM was \emph{faster than \BAIII by at least $15$ and $26$ times on average}. For example, when $\Delta=800,000$ in \Genes (see Fig.~\ref{fig:gene_time_delta}), \SCM took only about $20$ seconds while \BAIII about $8$ minutes. Figs.~\ref{fig:webkb_mem_delta} to~\ref{fig:vir_mem_delta} show the peak memory consumption for the experiments of 
Figs.~\ref{fig:webkb_time_delta} to~\ref{fig:vir_time_delta}.  Again, \SCM was more space-efficient than \BAIII; it needed \emph{at least $9\%$ and up to $17\%$ less memory}. 

\paragraph{Case Study: Uniform Pattern Mining.}~We consider collections whose strings represent different user subpopulations and solve the \CPM problem 
with large $\epsilon$. As will be discussed later, each letter in each string represents a movie genre, book genre, or product price tier, and \CPM discovers \emph{all} substrings with bounded difference in terms of their frequencies in the strings of the collection. There is \emph{no other  restriction} to the type of patterns that can be discovered, which makes \CPM challenging. For example, if one could restrict the length of the mined patterns and/or their items (e.g., to be semantically similar movie genres), then the problem space becomes smaller and other approaches such as those based on inverted indexes may be more suitable. However, actionable patterns (e.g., the pattern comprised of the $4$ non-semantically similar letters/items \texttt{Comedy Drama Romance Comedy} that we mined) may be missed. In our case study, we discover  
patterns revealing behavioral preferences that prevail in these subpopulations. 
For each pattern, we refer to literature supporting that it is in fact prevailing in these subpopulations.

\subparagraph{Datasets.}~We processed three datasets:  \Movielens~\cite{harper2015movielens}, \BookCrossing~\cite{ziegler2005improving}, and \Alibaba~\cite{pei2019value}. The processed datasets are comprised of two strings each, and they can be found at our code and data link. 
Each string corresponds to a different user subpopulation, and in each string  each distinct letter corresponds to a distinct entity (movie genre, book genre, price tier). In \Movielens, one string corresponds to $1,709$ men and the other to $1,709$ women. Each string has length $4,000,033$ and contains genres of rated movies, ordered chronologically. The total number of distinct genres (alphabet size) is $18$. In \BookCrossing,
one string corresponds to $9,164$ teenagers ($10$-$18$ years old) and the other to $4,252$ elderly ($65$-$100$ years old). Each string has length $14,145$ and contains book genres, ordered by users' ratings for book genres from high to low. The genres were found by mapping book ISBNs to genres using ChatGPT-3.5-turbo~\cite{openai2023chatgpt}, which is remarkably good in this task~\cite{raj}.
The total number of distinct book genres (alphabet size) is $10$. In \Alibaba,
one string corresponds to $928,622$ users with high purchase power ($\geq 15$) and the other to $894,770$ users with low purchase power ($<15$). Each string has length $34,651,424$ and contains price tiers of products, ordered in the way users browsed them. The total number of distinct price tiers (alphabet size) is $4$. Specifically, the letters \texttt{L}, \texttt{M}, \texttt{H}, and \texttt{P}, correspond to the Low, Middle, High, and Premium price tier, respectively.

\begin{table*}[t]
\centering
\resizebox{\textwidth}{!}{
\begin{tabular}{llcccl}
\toprule
\textbf{Dataset} & \textbf{Pattern} & \textbf{Frequency in String A} & \textbf{Frequency in String B} & \boldmath{$\epsilon$} & \textbf{Reference} \\ \midrule  
\Movielens & $\seq{Action~Drama~War~Action}$ & 1,551 (Men) & 588 (Women) & 1,000 & \cite{wuhr2017tears, krcmar2005uses} \\ \midrule
\Movielens & $\seq{Comedy~Drama~Romance~Comedy}$ & 899 (Men) & 1,860 (Women) & 1,000 & \cite{wuhr2017tears, infortuna2021inner} \\ \midrule
\Movielens & $\seq{Adventure~Fantasy~Sci\text{-}Fi}$ & 2,008 (Men) & 1,079 (Women) & 1,000 & \cite{wuhr2017tears} \\ \midrule
\BookCrossing & $\seq{Fiction~Mystery~Mystery}$ & 157 (Teenagers) & 256 (Elderly) & 100 & \cite{infortuna2021inner, awriterofhistory_2014, currie2025reading} \\ \midrule 
\BookCrossing & $\seq{Fiction~Fiction~Adventure}$ & 380 (Teenagers) & 102 (Elderly) & 1,000 & \cite{currie2025reading, dubourg2023exploratory} \\ \midrule 
\BookCrossing & $\seq{Fiction~Adventure~Adventure}$ & 316 (Teenagers) & 40 (Elderly) & 1,000 & \cite{infortuna2021inner} \\ \midrule 
\Alibaba & \texttt{H H H} & 1,096,218 (High Purchase Power) & 306,722 (Low Purchase Power) & $10^6$ & \cite{han2010signaling, ngwe2019impact} \\ \midrule 
\Alibaba & \texttt{L L L} & 166,714 (High Purchase Power) & 2,432,915 (Low Purchase Power) & $10^7$ & \cite{ngwe2019impact,  lu2020status} \\ \midrule 
\Alibaba & \texttt{P H P} & 1,340,059 (High Purchase Power) & 155,092 (Low Purchase Power) & $10^7$ & \cite{mundel2017exploratory, berger2010subtle}\\ \bottomrule
\end{tabular}
}
\vspace{+1mm}
\caption{Frequencies of mined patterns across different datasets and user subpopulations, with varying $\epsilon$ thresholds.}\label{tab:within}
\end{table*}

\subparagraph{Uniform Patterns.}~From each dataset, we mined uniform patterns and show the $3$ patterns with the largest difference among their frequencies in the two strings of the dataset. Table~\ref{tab:within} shows the patterns, their frequency in each string of their dataset, the $\epsilon$ value used, and references supporting that indeed each pattern is prevalent in the subpopulation in which it has the largest frequency in our dataset. 
These patterns effectively reveal behavioral differences between different subpopulations. In \Movielens, 
they show which movie genres are preferred mostly by men (first and third pattern) or by women (second pattern). These movie genres are indeed preferred by the respective subpopulations~\cite{wuhr2017tears,krcmar2005uses,infortuna2021inner}. In \BookCrossing, the patterns show which book genres are preferred mostly by elderly (first pattern) or by teenagers (second and third pattern) and indeed this agrees with the literature \cite{awriterofhistory_2014,currie2025reading,infortuna2021inner,dubourg2023exploratory}. Last, in
\Alibaba, the patterns show that customers with high (respectively, low) purchase power browse products in the high or premium price tier (respectively, in the low price tier), and this is again well-supported by the literature~\cite{han2010signaling,ngwe2019impact,lu2020status,mundel2017exploratory,berger2010subtle}.

\paragraph{Case Study: Consistent Query String.}~We solve the \CSQ problem to efficiently distinguish between DNA strings belonging to different biological entities. 

\subparagraph{Dataset.} We used the SARS-CoV-2~\cite{ncbi_sars_cov_2_2024} (\SARS) database as~$\mathcal{S}$. \SARS contains $2,000$ strings, each representing a different genetic variation of SARS-CoV-2. The total length of strings in \SARS is $59,515,733$ and the alphabet size of all these strings is $4$. We randomly selected $50$ strings from \SARS as queries $Q$ and removed them from the database.  
Additionally, we retrieved $50$ Influenza A genomes~\cite{ncbi_influenza_a_2025} (\INFL) and 
sampled $50$ substrings from the human genome~\cite{ncbi_human_genome_2022} (\HUM). Each retrieved string from \INFL and each sampled substring of \HUM was used as a query on \SARS.   
Since the genome size of \INFL is smaller than that of the \SARS virus ($\approx 13$kb and $\approx 30$kb, respectively), we used the entire Influenza A genomes as queries. On the other hand, the genome of \HUM is much larger ($\approx 3$Gb), so we sampled \HUM substrings  that were of very similar length to the \SARS genome.
In particular, the average length of the \SARS, \INFL, and \HUM queries is 29,777, 13,606, and 29,758, respectively, and all queries have roughly the same number of distinct $q$-grams for all tested $q$ values. 
We used $\epsilon \in \{0, 1, 2\}$ and measured the percentage of $\epsilon$-consistent $q$-grams for different length of the query $q\in [4,7]$. 

\begin{figure}[tbp]
\centering
\begin{subfigure}[b]{0.32\columnwidth}
    \centering
    \includegraphics[width=\textwidth]{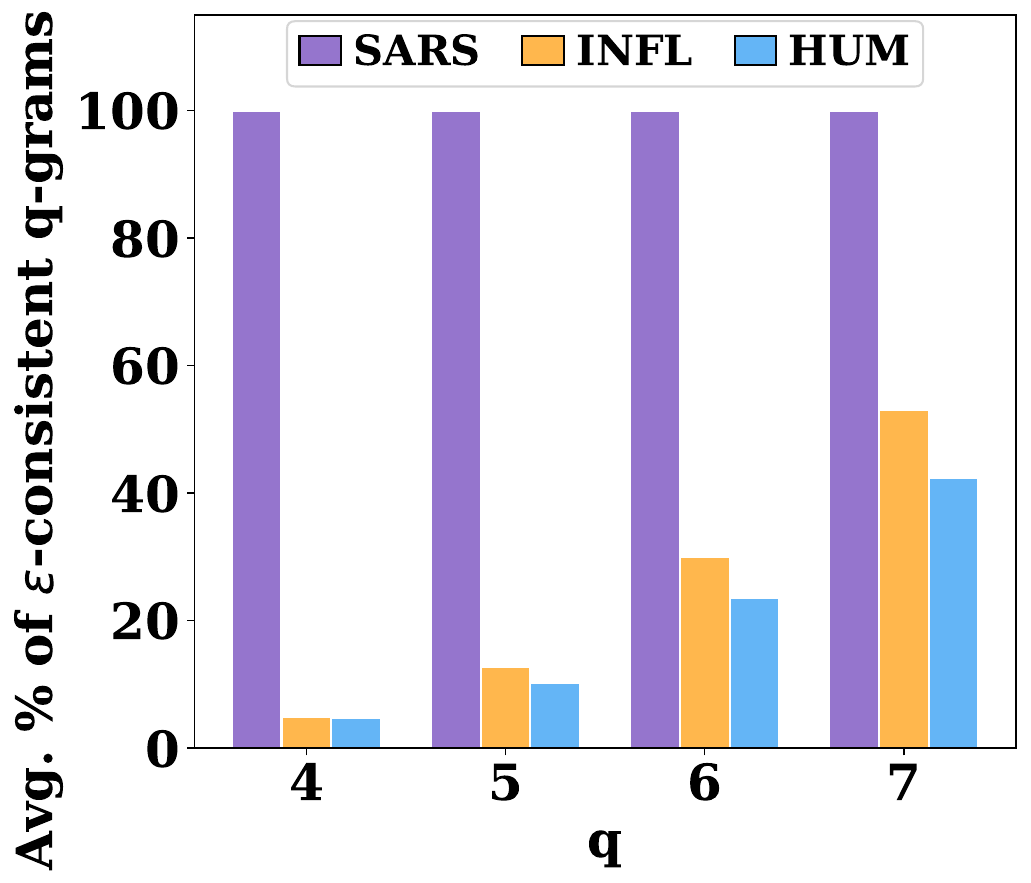}
    \caption{$\epsilon = 0$}
    \label{fig:CSQ_epsilon_0}
\end{subfigure}
\begin{subfigure}[b]{0.32\columnwidth}
    \centering
    \includegraphics[width=\textwidth]{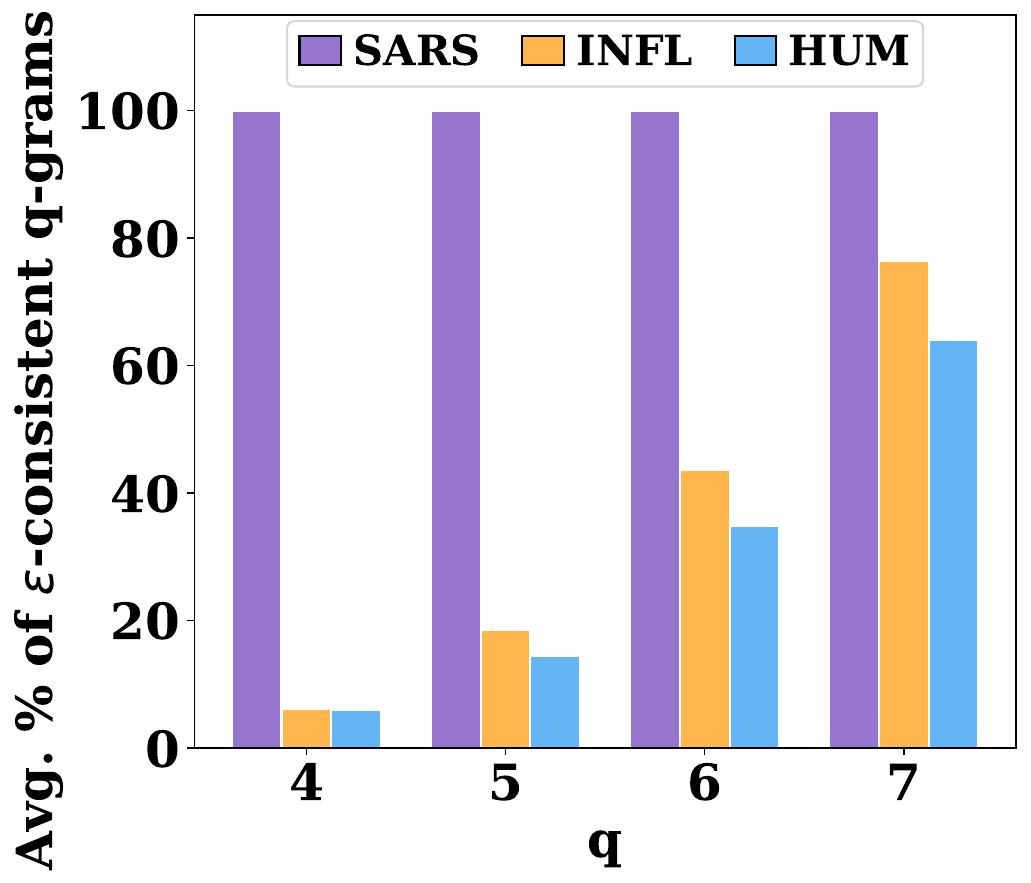}
    \caption{$\epsilon = 1$}
    \label{fig:CSQ_epsilon_1}
\end{subfigure}
\begin{subfigure}[b]{0.32\columnwidth}
    \centering
    \includegraphics[width=\textwidth]{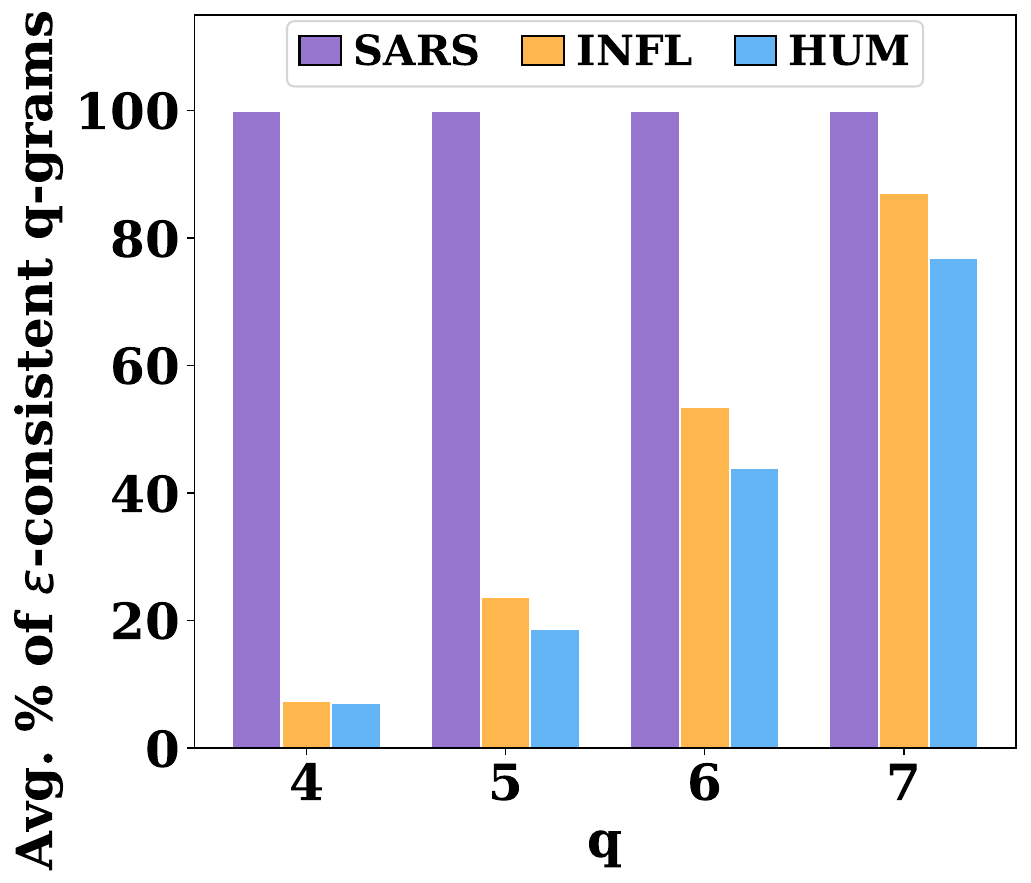}
    \caption{$\epsilon = 2$}
    \label{fig:CSQ_epsilon_2}
\end{subfigure}
\caption{Average percentage of $\epsilon$-consistent $q$-grams for different values of $\epsilon$ using a database of SARS-CoV-2 genomes.}
\label{fig:CSQ_epsilon}
\end{figure}

The results in Fig.~\ref{fig:CSQ_epsilon} show the percentage of $\epsilon$-consistent $q$-grams averaged over the \SARS, \HUM, and \INFL queries. All $q$-grams in the \SARS queries  
are $\epsilon$-consistent for $\epsilon\in \{1, 2\}$ and between $99.5\%$ and $99.8\%$ of them are $\epsilon$-consistent for $\epsilon=0$. 
In contrast, only $4.85\%$ of
$q$-grams on average in the INFL queries are $\epsilon$-consistent for $\epsilon=0$, $6.1\%$ for $\epsilon=1$, and $7.41\%$ for $\epsilon=2$. This shows that indeed the \CSQ problem helps to efficiently distinguish DNA sequences belonging to different biological entities (SARS queries are similar to the strings in the \SARS database unlike \INFL queries). The reason the percentage of $\epsilon$-consistent $q$-grams increases with $\epsilon$ is because the frequency interval in \CSQ gets larger. The percentage of $\epsilon$-consistent $q$-grams also increases with $q$
because the $q$-grams for high $q$ values have low frequency (e.g., $1$ or $2$),  which makes it easier  for them to be $\epsilon$-consistent.
As expected, the number of $\epsilon$-consistent $q$-grams in the HUM queries is even lower  compared to that in the \INFL queries, as the former come from a human while the latter come from a   virus (and the  \SARS database contains the genome of viruses). For example, for $\epsilon=2$ and $q=5$, $18.7\%$ of $q$-grams in the \HUM queries on average are $\epsilon$-consistent while the corresponding percentage for \INFL queries is $23.73\%$. 

We compared our approach to a state-of-the-art method for viral sequence classification~\cite{vir3}, which uses vectors of $q$-gram frequencies as features to train random forests. The method, referred to as \textsf{RF}, was implemented in \textsf{Python} using \textsf{scikit-learn} and configured and trained as in~\cite{vir3}. As \textsf{RF} requires also negative samples (i.e., also \INFL and \HUM strings  in its training set), we used two training sets: (1) $T_1$, comprised of $3,950$ strings; the $1,950$ \SARS strings we used in the experiment of Fig.~\ref{fig:CSQ_epsilon}, \num{1000} randomly selected \INFL strings of average length \num{13208} from~\cite{ncbi_influenza_a_2025}, and \num{1000} randomly selected \HUM strings of average length \num{29758} from the human genome~\cite{ncbi_human_genome_2022}; and (2) $T_2$ constructed from $T_1$ by removing \num{950}, \num{500}, and \num{500} randomly selected \SARS, \INFL, and \HUM  strings, respectively. \SARS strings had a label  \texttt{True} and the others a label \texttt{False}.  Table~\ref{tab:classif1a} shows the time for training \textsf{RF} on $T_1$ or $T_2$ and the time for our approach, for varying $q$. As expected, \textsf{RF} takes more time on $T_1$ and when $q$ is larger, while our approach is insensitive to $q$ and significantly faster. In fact, when we increased the training set size and $q$, our approach was \emph{up to \num{30} times faster}. This can be seen in Table~\ref{tab:classif1b}; the training set $T_3$ for \textsf{RF} was comprised of \num{20000} randomly selected strings from \SARS with average length \num{29773}, which were also used by our approach, and also \num{10000} randomly selected strings from the human genome~\cite{ncbi_human_genome_2022} with average length \num{29758}. Furthermore, \textsf{RF} used \emph{up to \num{3.9} times more memory}. In all experiments, we used the queries of the experiment of Fig.~\ref{fig:CSQ_epsilon} for both \textsf{RF} and our approach, $\epsilon=0$, and classified every query as \SARS when $>99\%$ of its $q$-grams are $\epsilon$-consistent and as not-\SARS otherwise. Both \textsf{RF} and our approach had no misclassification errors. 

\begin{table}[t]
\centering
\begin{subtable}[t]{0.25\linewidth}
\resizebox{1.1\linewidth}{!}{
\begin{tabular}{lccc}
\toprule
\textbf{$q$} & \textsf{RF}  & \textsf{RF} & {Our approach} \\
~ & $T_1$ & $T_2$ & \SARS \\ 
\midrule
$4$ & 87.15  & 46.22  & 16.57 \\
$5$ & 102.06 & 52.97  & 16.57 \\
$6$ & 119.41 & 60.82  & 16.57 \\
$7$ & 154.26 & 77.67  & 16.57 \\
\bottomrule
\end{tabular}}
\caption{}\label{tab:classif1a}
\end{subtable}\hspace{+8mm}
\begin{subtable}[t]{0.25\linewidth}
\resizebox{0.92\linewidth}{!}{
\begin{tabular}{lcc}
\toprule
\textbf{$q$} & \textsf{RF} & {Our approach} \\ 
~ & $T_3$ & $20,000$ \SARS\\\midrule
$6$ & 501.24  & 67.74 \\
$7$ & 657.33 & 67.74 \\
$8$ & 1022.48 & 67.74 \\
$9$ & 2128.33 & 67.74 \\
\bottomrule
\end{tabular}}
\caption{}\label{tab:classif1b}
\end{subtable}
\caption{(a) Runtime (seconds) for \textsf{RF}  trained on $T_1$ or $T_2$ and for \SCM using the \SARS database. (b) Runtime (seconds) for \textsf{RF} trained on $T_3$ and for \SCM using $20,000$ \SARS strings.\label{tab:classif1}}
\label{tab:runtime}
\end{table}

\begin{figure}[htbp]
\hspace{-2mm}
\begin{subfigure}[b]{0.31\columnwidth}
    \includegraphics[width=1\textwidth]{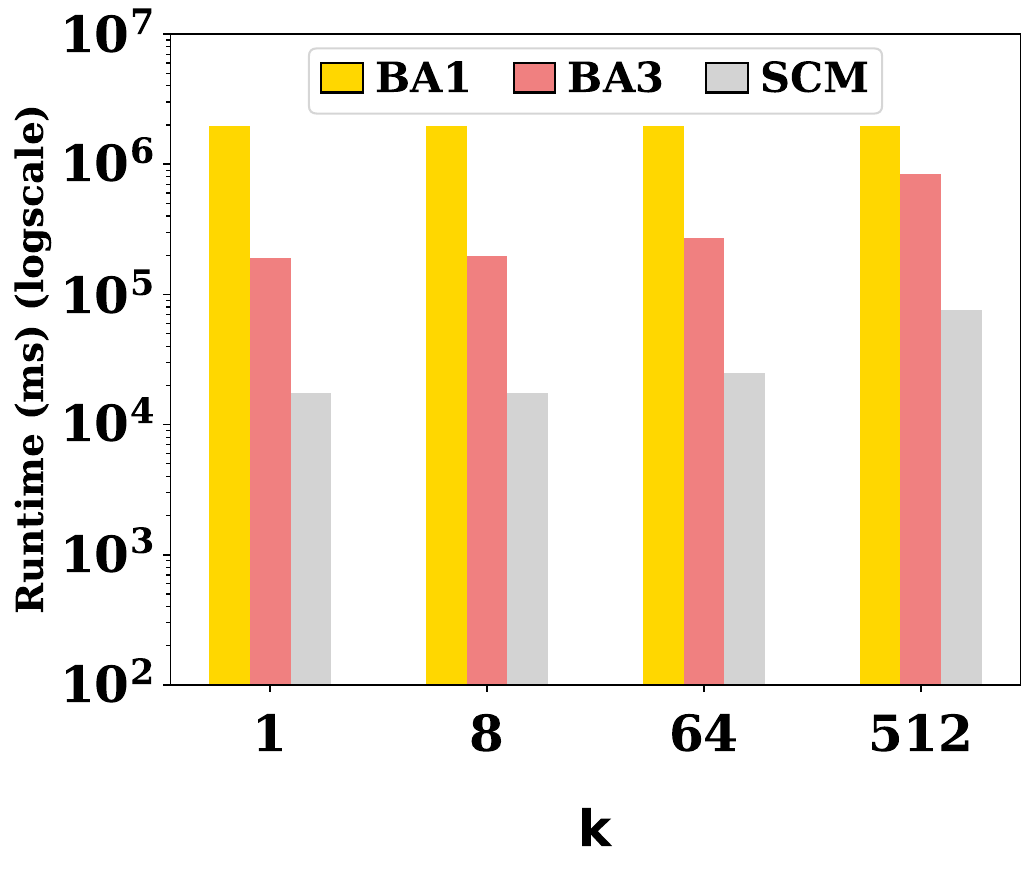}
    \caption{\SARS}
    \label{fig:sars_time_1}
\end{subfigure}\hspace{+1mm}
\begin{subfigure}[b]{0.31\columnwidth}
    \includegraphics[width=1.04\textwidth]{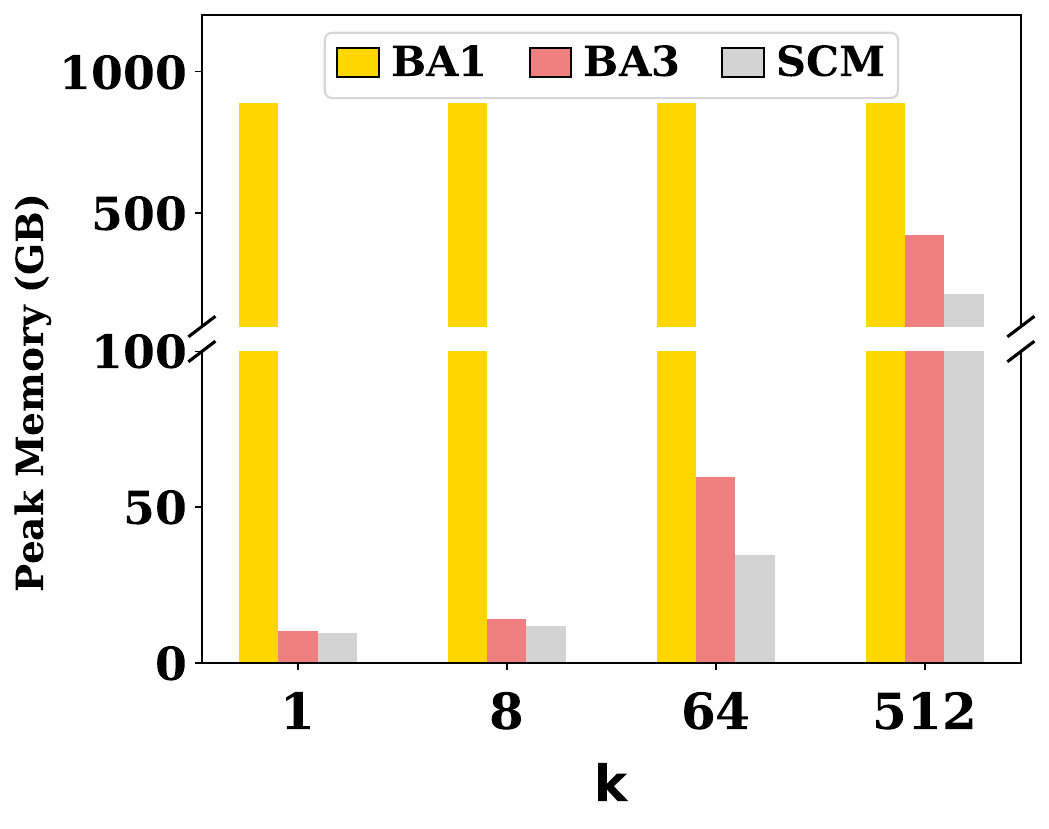}
    \caption{\SARS}
    \label{fig:sars_space_1}
\end{subfigure}\hspace{+1mm}
\begin{subfigure}[b]{0.31\columnwidth}
    \includegraphics[width=1\textwidth]{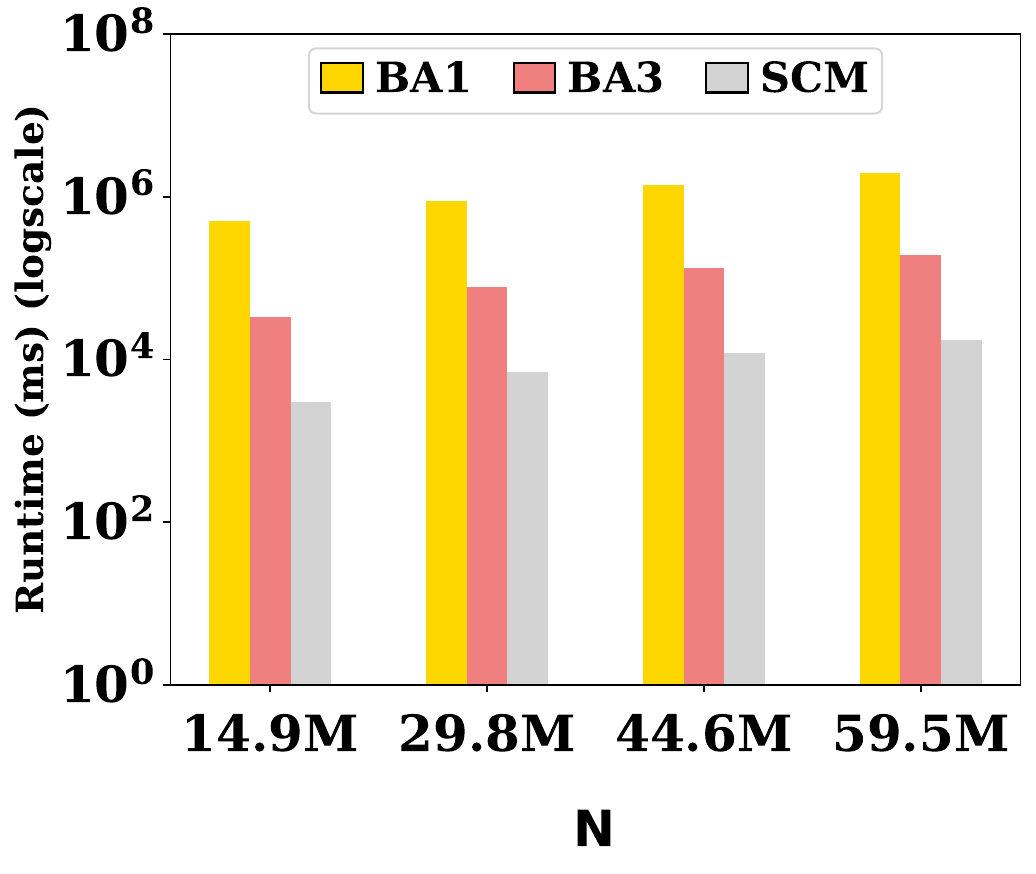}
    \caption{\SARS}
    \label{fig:sars_time_3}
\end{subfigure}\\
\begin{subfigure}[b]{0.31\columnwidth}
    \includegraphics[width=1.06\textwidth]{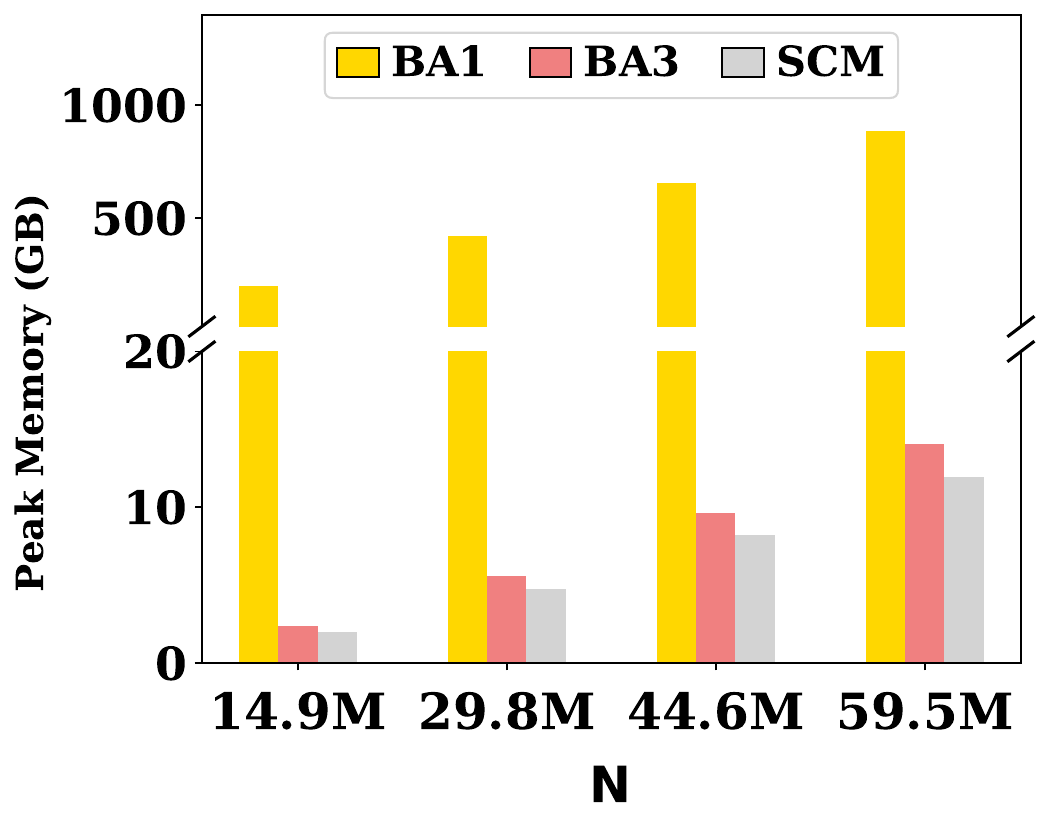}
    \caption{\SARS}
    \label{fig:sars_space_3}
\end{subfigure}\hspace{+1mm}
\begin{subfigure}[b]{0.31\columnwidth}
    \includegraphics[width=0.98\textwidth]{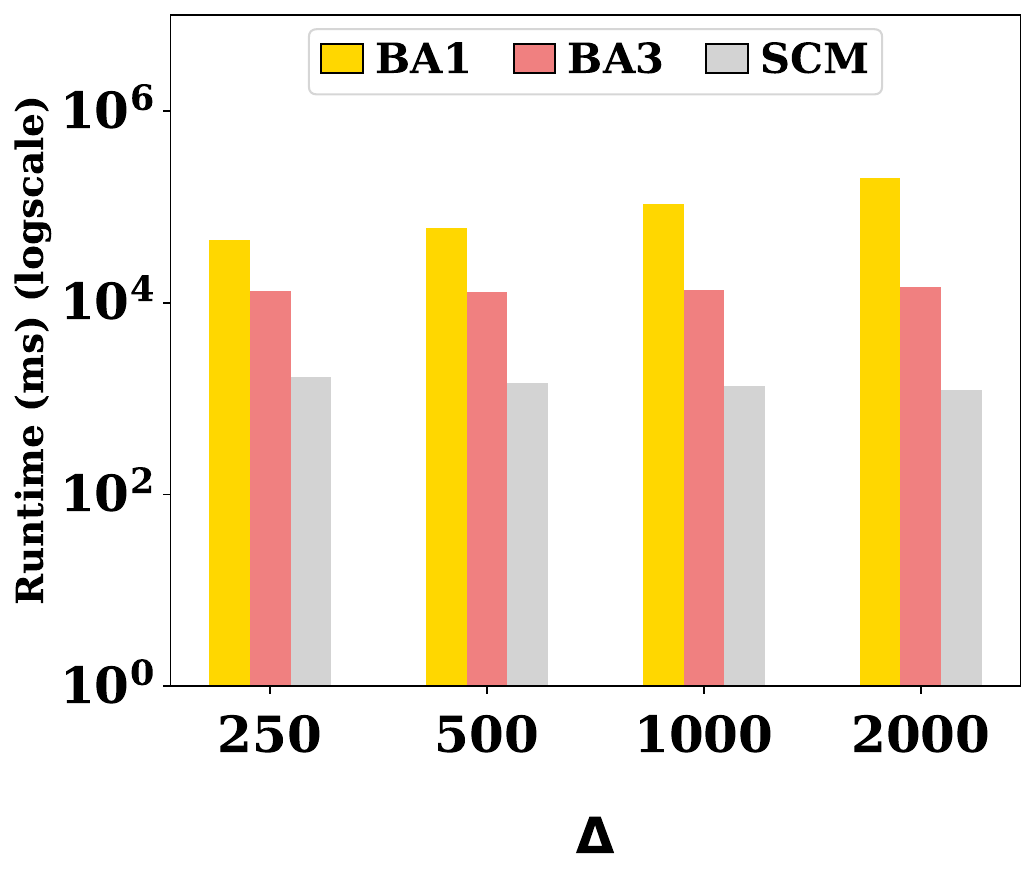}
    \caption{\SARS}
    \label{fig:sars_time_2}
\end{subfigure}\hspace{+1mm}
\begin{subfigure}[b]{0.31\columnwidth}
    \includegraphics[width=1.06\textwidth]{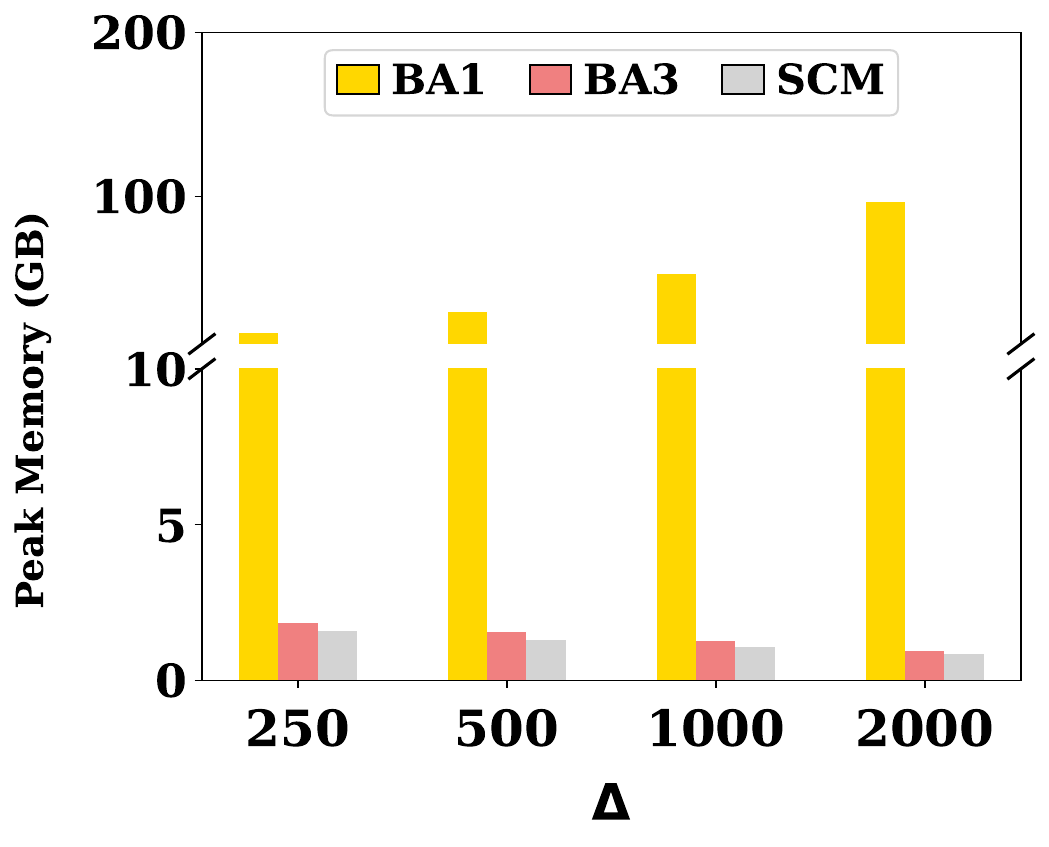}
    \caption{\SARS}
    \label{fig:sars_space_2}
\end{subfigure}\hspace{+1mm}
\caption{Runtime and memory for \SARS vs. (a-b) $k$, (c-d) $N$ for $k=8$, and (e-f) $\Delta$, for $k=8$ and $N=7,439,466$.}
\label{fig:runtime-mem-sars}
\end{figure}

\paragraph{\kDRQ.} 
We applied the modified \SCM algorithm (see \cref{cor:k-most-frequent}) on the SARS database. Recall from Section~\ref{sec:related} that there are no practical methods for \kDRQ (or $k$-\SM). Also, it is unclear if and how these problems could be addressed based on \BAII. Thus, we compared our algorithm to: (1) a modified version of \BAI, which derives the sorted list $(c_1, f_1), \ldots, (c_k, f_k)$ per node $v$ based on array $A_v$ (see \cref{cor:k-most-frequent} and \cref{thm:bl1}), and it takes $\cO(N\Delta)$ time and space, and (2) a modified version of \BAIII following \cref{cor:k-most-frequent} which takes $\cO(N\log \Delta +kN)$ time and $\cO(kN)$ space. Fig.~\ref{fig:sars_time_1} shows that, for all $k$ values, our algorithm is faster by \emph{at least one order of magnitude} than the modified \BAIII and even faster than the modified \BAI. It also uses \emph{at least $6.4\%$ and $75.7\%$ less memory} than the modified \BAIII and \BAI, respectively (see Fig.~\ref{fig:sars_space_1}). For varying $N$ (respectively, $\Delta$), our algorithm is $11$ (respectively, $10$) times faster on average than the fastest baseline (see Figs.~\ref{fig:sars_time_3} and~\ref{fig:sars_time_2}), and it uses $15.3\%$ (respectively, $13.8\%$) less memory on average (see Figs.~\ref{fig:sars_space_3} and~\ref{fig:sars_space_2}). All results are in line with the complexities of the algorithms and show that our algorithm is the most time- and space-efficient.

\paragraph{Phylogenetic Tree Annotation.} We show that our \SCM algorithm helps detecting  candidate genotype-phenotype links.  We used the phylogenetic tree of~\cite[Fig. S1]{10193} which is constructed on \emph{Treponema pallidium} (a syphilis-causing bacterium) data from the study of~\cite{Arora2017OriginSyphilis}. This tree has $75$ leaves, each representing a bacteria strain annotated by: (1)  antibiotic resistance ($3$ colors: \emph{sensitive}, \emph{resistant}, or \emph{n/a}), and (2) SNP state ($10$ colors: \texttt{G}, \texttt{A}, \texttt{T}, \texttt{C}, \texttt{N}, each  \emph{supporting} or \emph{non-supporting}) at a genomic position; there are in total $4$ positions. We applied our algorithm $5$ times; once using the antibiotic resistance coloring, and $4$ using the SNP state colorings. This yielded $5$ (mode, frequency) pairs per \emph{clade} (i.e., internal node that represents similar bacteria strains indicating transmission between patients). The pairs were used to annotate the tree; a part of the tree annotated only by  antibiotic resistance and SNP state for one position is in Fig.~\ref{fig:phylocs}. 
By inspecting the (entire) annotated tree, a domain expert (the fourth author of this paper) identified SNPs whose allele patterns correlate with clades of mostly sensitive or mostly resistant strains. This revealed associations worth further statistical testing (e.g., to find out which mutations corresponding to non-supporting alleles cause antibiotic resistance)~\cite{10193,Arora2017OriginSyphilis}. Our algorithm took less than a second to produce the annotation information, while manual computation of this information with the aid of a state-of-the-art tool~\cite{itol6}  took about $2$ hours. 

\begin{figure}[htbp]
\centering
    \includegraphics[width=0.7\columnwidth]{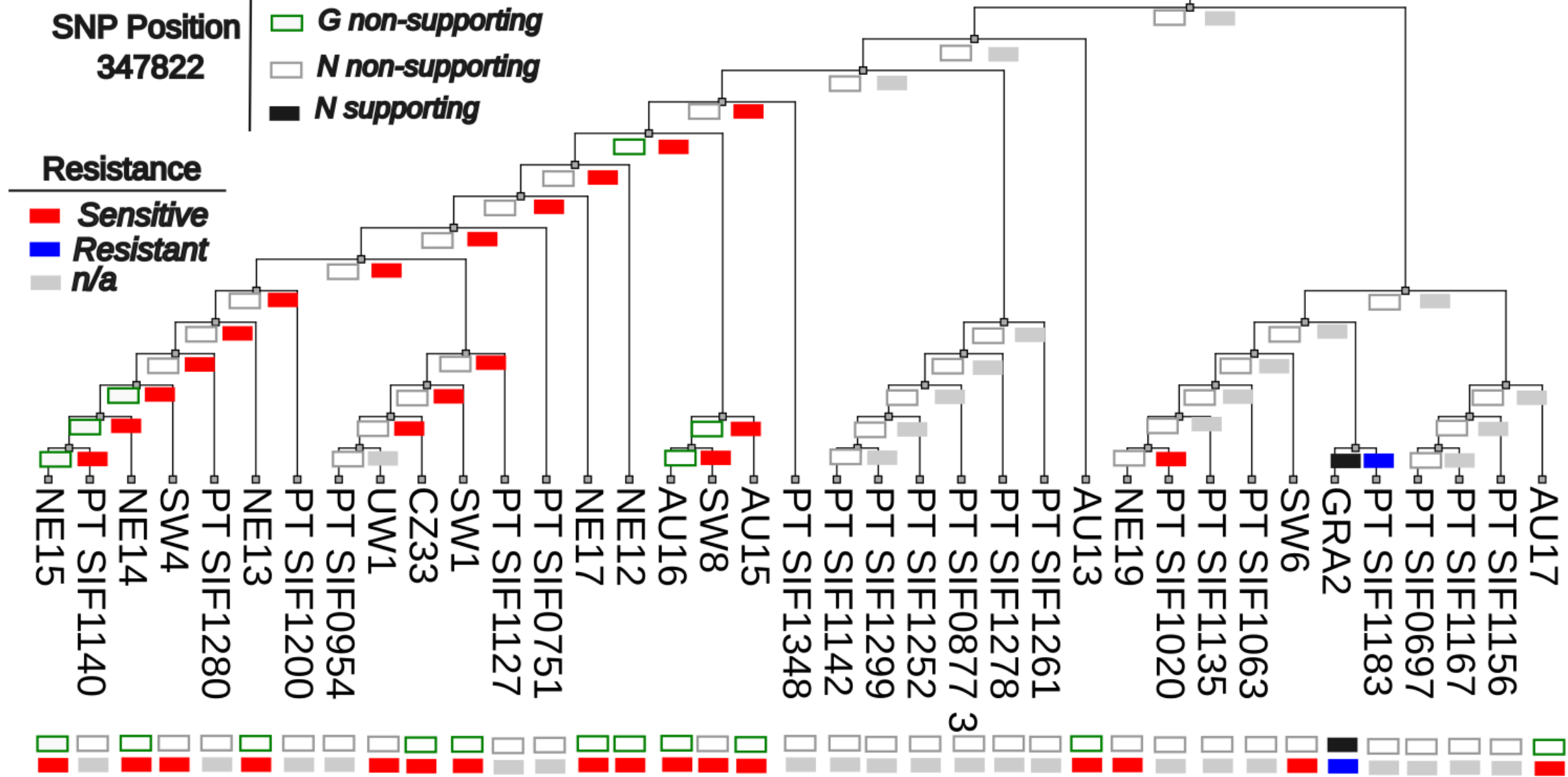}
    \caption{Part of the tree of~\cite[Fig. S1]{10193} annotated with antibiotic resistance status and SNP state for position \num{347822}.}\label{fig:phylocs}
\end{figure}

\section{Conclusion}\label{sec:conclusion}

We introduced the \SM problem and proposed \SCM, a time-optimal $\cO(N)$-time algorithm to solve it. This algorithm forms the basis of time-optimal solutions for document retrieval, pattern mining, and sequence-to-database search applications. We also studied natural generalizations of \SM that work on node-colored trees, or ask for the $k$ most 
frequent colors in the leaves of the subtree of a given node. Furthermore, we proved that the analogous problem to \SM where the input is a sink-colored DAG is 
highly unlikely to be solved as fast as \SM. 
Our experiments showed that \SCM is much faster and space-efficient than two natural baselines and an $\cO(N\log \Delta)$-time variant of it, while it can be used to discover meaningful uniform patterns and to efficiently distinguish between DNA sequences belonging to different biological entities. As future work, we aim to study generalizations of the \SM problem and their applications. Another interesting direction for future work is to study
dynamic versions of \SM, where the underlying tree can be updated (its elements can be changed, inserted, or deleted) and queries should still be answered efficiently.

\printbibliography[title={References}]


\end{document}